\documentclass[a4paper,12pt]{fullarticle}
\usepackage[british]{babel}
\usepackage{csquotes}
\usepackage{sciencestuff}
\usepackage[algoruled,algosection,vlined,shortend,linesnumbered]{algorithm2e}

\title{%
Functional Renormalization Group for a Rank-4 Renormalizable Tensorial Group Field Theory with Derivative Necklace Couplings
}

\author[1]{Seke Fawaaz Zime Yerima\emailfoot{fawaaz.zimeyerima@uac.bj}}
\author[2]{Vincent Lahoche\emailfoot{vincent.lahoche@cea.fr}}
\author[1]{Dine Ousmane Samary\emailfoot{dine.ousmanesamary@uac.bj}}

\affil[1]{%
	Faculté des Sciences et Techniques (ICMPA-UNESCO Chair)
	\protect \\
	Université d'Abomey-Calavi, 072 BP 50, Benin
}

\affil[2]{%
  Université Paris-Saclay, CEA, 
  \protect \\
  Palaiseau, F-91120, France
}

\date{}

\hypersetup{%
  pdftitle={Non-perturbative fixed point for a rank $4$ renormalizable Tensorial Group Field Theory with derivative couplings},
  pdfkeywords={%
  functional renormalization group,
  theoretical physics,
  data science,
  signal analysis,
  signal detection,
    random matrix theory
    },
  pdfsubject={signal detection}
}

\newtheorem{remark}{Remark}
\newtheorem{definition}{Definition}
\newtheorem{theorem}{Theorem}
\newtheorem{proposition}{Proposition}
\newtheorem{corollary}{Corollary}
\newtheorem{lemma}{Lemma}
\newtheorem{ansatz}{Ansatz}

\newcommand{\sym}{\mathrm{Sym}}

\begin{document}

\maketitle

\begin{abstract}

We apply the functional renormalization group to an Abelian Group Field Theory extended beyond the branched-polymer (melonic) sector by including interactions that are subdominant from a power-counting perspective but enhanced by derivative couplings. Focusing on a rank-4 model, we consider a class of non-melonic interactions with a necklace structure. Due to their index contraction pattern, their leading-order behavior is analogous to that of large-N random matrix models and is associated with a planar graph structure. Within this setting, we identify the emergence of a nontrivial ultraviolet fixed point, reminiscent of mechanisms previously observed in matrix models, and discuss its reliability within the present truncation. The robustness of this fixed point will be further investigated through modified Ward identities, following strategies previously developed in the melonic sector.
\end{abstract}

\highlights{%
We employ the Non-Perturbative Renormalization Group (NPRG) framework to investigate the phase structure of an Abelian tensorial group field theory (TGFT), characterized by a competition between melonic and necklace-type bubble interactions.
}


\clearpage

{\small\tableofcontents}

\clearpage


\section{Introduction}

Group Field Theories (GFT) represent a class of field theories defined on a group manifold, characterized by a specific form of non-locality in their interactions. This non-locality endows their Feynman diagrams with the structure of cellular complexes rather than simple graphs \cite{Freidel_2005,baratin2012ten,https://doi.org/10.48550/arxiv.1110.5606,https://doi.org/10.48550/arxiv.gr-qc/0607032,https://doi.org/10.48550/arxiv.1210.6257}. Fundamentally, a GFT is defined over a group manifold $G^d$, where $d$ corresponds to the number of field variables. In the context of quantum gravity research, this group is generally the Lorentz group, its Riemannian counterpart, or one of their subgroups, such as the compact groups $SU(2)$ or $SO(2,1)$. GFTs can be understood through several complementary lenses: as a generalization of random matrix models \cite{Francesco_1995}, or as an enrichment of random tensor models \cite{gurau2017random} with pre-geometric data. Furthermore, they emerge as a second-quantized version of Loop Quantum Gravity (LQG). At the perturbative level, their Feynman graphs appear as simplicial complexes weighted by spin-foam amplitudes; thus, GFTs can be viewed as a completion of both the LQG and spin-foam approaches to covariant quantum general relativity see \cite{https://doi.org/10.48550/arxiv.1310.7786,oriti2015group}. Tensorial Group Field Theories (TGFTs) are a subclass of GFTs where interactions are termed \textit{tensorial} \cite{carrozza2014renormalization2,carrozza2014tensorial}, in the sense defined by colored random tensor theories \cite{bonzom2012random,gurau2017random}. Indeed, rank-$d$ tensor models of size $N$ are characterized by a $U(N)^{\otimes d}$ invariance, providing a $1/N$ expansion similar to that of matrix models. For tensor models, the $1/N$ expansion is governed by an index called the \textit{Gurau degree}, which plays a role analogous to the genus in matrix models. The leading-order contributions correspond to vanishing degree; it has been established that these correspond to specific spherical triangulations called \textit{melons}, which play a role similar to planar graphs in matrix models.

The $U(N)^{\otimes d}$ invariance in tensor models arises from the specific contraction pattern of tensor indices. This same scheme is employed in TGFTs, where sums over indices are replaced by integrations over group variables. As with tensor models, this construction allows for a power-counting analysis of the theory. Moreover, the tensorial structure of interactions enables the definition of a locality principle known as \textit{traciality} \cite{carrozza2014renormalization2,carrozza2014tensorial,rivasseau2016random,carrozza2014renormalization,https://doi.org/10.48550/arxiv.1111.4997,carrozza2016flowing}, provided the propagator is modified by a Laplacian-type coupling (a power of the Laplacian on $G^d$). Ultimately, renormalization and the physics of phase transitions play a crucial role in the fundamental question of the emergence of spacetime. Indeed, the primary challenge for background-independent approaches to quantum gravity is to understand the transition from the discrete quantum degrees of freedom of a fundamental theory candidate to the description of a continuous, or at least semi-classical, spacetime through an emergence scenario. In the GFT approach, a significant step toward understanding this mechanism has been taken in several recent works \cite{oriti2015generalized,gielen2014quantum,gielen2022effective,oriti2021tensorial,jercher2022emergent,oriti2018black,de2017dynamics,kegeles2018inequivalent}. In the first part of these studies, the authors derive the Friedmann cosmological equations for a homogeneous and isotropic Universe from the Schwinger-Dyson equations, under the hypothesis that such a spacetime corresponds to a condensate of microscopic degrees of freedom, by analogy with condensed matter physics. In the second part, the authors discuss phase transition aspects and the role of matter degrees of freedom, justifying some of the hypotheses regarding condensate emergence.

However, this choice remains an intuition, motivated by parallels with quantum optics but still lacking a rigorous justification derived from the theory itself. Within this framework, the non-perturbative Functional Renormalization Group (FRG) approach becomes a natural tool, offering a robust framework to investigate phase transition physics and justify the occurrence of a condensed phase, as anticipated in \cite{marchetti2021phase}. This constitutes the primary physical motivation for the present work.

Over the past decade, numerous studies have explored the non-perturbative renormalization group for TGFTs \cite{Carrozza_2015a,Lahoche_2017bb,Benedetti_2016,Benedetti_2015,Geloun_2016,Ben_Geloun_2015,Geloun_2018}. The main difficulty of this approach lies in the non-local nature of the interactions, which limits the application of standard field theory methods \cite{Delamotte_2012}. Furthermore, the non-trivial Ward identities of these models add to their specificity and have been extensively studied in a series of papers \cite{Lahoche_2019bb,Lahoche_2019a,Lahoche_2020d,Lahoche_2021c,Lahoche_2020b,https://doi.org/10.48550/arxiv.1701.03029,lahoche2021no,Lahoche:2018oeo}, focusing on the symmetric phase while exploring the entirety of the non-branched melonic sector. Other methods, proposed in \cite{pithis2021no,pithis2020phase}, bypass Ward identities and instead explore the non‑symmetric regime. A doubt persists to this day regarding the existence of a melonic fixed point supporting a pre-geometric transition. In particular, it appears that accounting for Ward identities suppresses these Wilson-Fisher type fixed points, which tended to manifest in early studies that ignored these constraints.\\

In this paper, we propose to apply FRG methods to a rank-4 TGFT model, including interactions beyond the melonic sector. These interactions exhibit a cyclic and balanced index structure (with two links between each node) and go by the evocative name of \textit{necklaces}. While these necklaces are sub-dominant compared to melons under ordinary power counting, their influence can be "enhanced" by carefully chosen derivative couplings. This places them on a par with melons in the hierarchy of canonical dimensions, thereby escaping the 'curse' of branched polymers. Similar procedures have been considered in colored tensor models (where $U(N)^{\otimes d}$ invariance is not broken by a propagator) in \cite{bonzom2015enhancing}; our work can be seen as an extension of this research to the GFT framework. The primary conclusion in the reference above was that, due to the enhancement of non-melonic interactions which appear "matrix-like" in structure the leading order remains tree-like, as in the strictly melonic sector, but with an additional structure called \textit{disks}, appearing as planar configurations. Like the melonic tree sector, the planar sector possesses its own critical behavior the model exhibiting a transition between them with a positive entropy exponent. The main interest of such a model, combining matrix-like behavior with a tensorial structure, is that it allows for phase transitions from a "branched polymer" phase (corresponding to the tensorial tree sector) to a planar phase, characterized by the proliferation of planar graphs along the trees, thus enriching the graph topology. From a physical perspective, this connects to the question of the emergence of matter degrees of freedom. In several approaches to this problem, matter is not simply added to gravitational microscopic degrees of freedom but emerges as fluctuations around the vacuum of the fundamental quantum gravity theory. Such an approach in the GFT framework has been explored, for instance, in \cite{fairbairn20073d}, where it was shown that fluctuations behave as an effective matrix field. This aligns with the standard view that trees are suitable building blocks for a more elaborate theory, providing a second motivation for our enhanced field theory.\\

The paper is organized as follows: In section \ref{sectiondef} we defined the model and the conventions used throughout this paper. In particular, we define the classical action and the quantized theory via the path integral formalism and Feynman diagrams. In section \ref{sectionpower} we investigated the power counting of the theory, allowing us to specify the renormalizable sector, which includes both melons and necklaces. In section \ref{sectionNP} we introduced the FRG formalism and the approximation scheme, providing a detailed derivation of the flow equations. Sections \ref{numeric} and \ref{conclusion} present the results and offer a conclusion opening onto future research.Finally, appendix \ref{App0} provides a basic definition of the colored graphs which are used throughout the manuscript, and appendix \ref{App1} presents an alternative (perturbative) computation using the Polchinski equation.


\section{A Group Field Theory beyond the world of melons}
\label{sectiondef}

We consider the simplest case of a GFT over the Abelian group manifold $U(1)^4$. From the perspective of tensor models and tensor field theories, this choice is quite natural: when the size $N$ of the tensor goes to infinity, the classical action is reduced to that of a group field theory on $U(1)^d$, where the tensor components become the Fourier coefficients of the fields. From a purely GFT point of view, our theory can be viewed as a toy model for quantum gravity, implementing more topological ingredients than traditional melonic TGFTs. A group field $\psi$ over $U(1)^4$ is a smooth map $\psi: U(1)^4 \to \mathbb{C}$. A field theory is then defined by the choice of a generating functional:
\begin{equation}\label{genfunc}
\mathcal{Z}[J,\bar{J}]=\int d\psi d\bar{\psi} e^{-S[\psi,\bar{\psi}]+\int_{U(1)^4}\bar{J}\psi+\int_{U(1)^4}\bar{\psi}J},
\end{equation}
where the sources $J$ and $\bar{J}$ are smooth mappings from $U(1)^4$ to $\mathbb{C}$, and the classical action $S[\psi,\bar{\psi}]$ splits into a kinetic part and an interaction part: $S[\psi,\bar{\psi}]=S_{kin}[\psi,\bar{\psi}]+S_{int}[\psi,\bar{\psi}]$. The first term encodes the dynamical structure of the field theory, while the second encodes the ways in which the fields interact with one another. For our purposes, we choose:
\begin{equation}\label{skin}
S_{kin}[\psi,\bar{\psi}]=\int_{U(1)^4}\prod_{i=1}^4 dg_i\bar{\psi}(g_1,g_2,g_3,g_4)\bigg(-\sum_{i=1}^4\Delta_i^{\eta}+m^{2\eta}\bigg)\psi(g_1,g_2,g_3,g_4),
\end{equation}
where $\Delta$ is the Laplace-Beltrami operator on $U(1)$, $dg$ is the Haar measure over $U(1)$, and $1/2 \leq \eta \leq 1$ is chosen in agreement with Osterwalder-Schrader positivity. In the history of TGFT, the introduction of a propagator of this form was motivated by radiative corrections in GFTs \cite{BenGeloun:2011jnm}, it appears as the simplest requirement to initiate a renormalization program using the tools of quantum or statistical field theory. In this context, Osterwalder-Schrader positivity plays the role of a natural constraint on the theories, especially from the perspective that a time dimension is already contained within the elementary random geometry theory though this remains an assumption.\\

\noindent

Let us now turn to the interaction part $S_{int}[\psi, \bar{\psi}]$. Such a theory is referred to as ``tensorial'' if these terms are invariant under independent unitary transformations acting on each index. To achieve this, each variable in the interaction terms attached to a field $\psi$ must be identified with a corresponding variable of a field $\bar{\psi}$ and integrated over the group manifold. As an example, we consider the following $T^4$ (quartic) interaction:
\begin{align}\label{intexample}
\int \prod_{i=1}^4 dg_idg_i'&\psi(g_1,g_2,g_3,g_4)\bar{\psi}(g_1',g_2,g_3,g_4)\psi(g_1',g_2',g_3',g'_4)\bar{\psi}(g_1,g_2',g_3',g'_4).
\end{align}

Such an interaction, exhibiting tensorial invariance, is called a bubble. A GFT whose interactions are of this type is referred to as a TGFT. Each bubble can be conveniently depicted as a bipartite, regular, 4-colored graph. In this representation, the fields $\psi$ and $\bar{\psi}$ are represented by white and black vertices, respectively. Each vertex is incident to four edges, representing the four group variables of the field. The connectivity of these edges follows the pattern of contractions (sums over variables) between the fields, as illustrated in Figure \ref{fig1}a for the bubble in Eq. \eqref{intexample}.
\begin{center}
\includegraphics[scale=0.85]{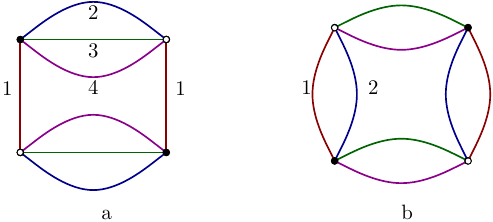} 
\captionof{figure}{The two possible configurations for quartic interactions. (a) Bipartite graph associated with the interaction given in Eq. \eqref{intexample}. (b) The non-melonic configuration.}\label{fig1}
\end{center}
We focus exclusively on quartic interactions, of which there are two distinct types, as depicted in Figures \ref{fig1}a and \ref{fig1}b. There are four bubbles of the first type (\ref{fig1}a), denoted by $b_i$ for $i \in \{1, \dots, 4\}$, where the index $i$ refers to the color of the 'singular' line. There are also three bubbles of the second type (\ref{fig1}b), denoted by $b_{1i}$ for $i \in \{2, 3, 4\}$, where the index $i$ indicates the color of the edge paired with the edge of color $1$. For instance, Figure \ref{fig1}b corresponds to the bubble $b_{12}$. We introduce the notation $\Tr_b[\psi, \bar{\psi}]$ to indicate that the field variables are contracted according to the scheme defined by the bubble $b$. Explicitly:
\begin{equation}
\Tr_{b_{12}}[\psi,\bar{\psi}]=\int \prod_{i=1}^4 dg_idg_i'\psi(g_1,g_2,g_3,g_4)\bar{\psi}(g_1',g_2',g_3,g_4)\psi(g_1',g_2',g_3',g'_4)\bar{\psi}(g_1,g_2,g_3',g'_4).
\end{equation}											

One can show from its definition that the Gurau degree vanishes for interactions of type (\ref{fig1}a). Since a melon is defined as a graph with a vanishing Gurau degree, these interactions are termed melonic. As discussed in the introduction, melonic graphs constitute the leading-order contributions, without a modification of the interactions of type (\ref{fig1}b), their contributions remain sub-leading. In the model studied in \cite{bonzom2015enhancing}, the authors addressed this issue by rescaling the interaction with an additional factor $N$, the tensor size. This rescaling ensures that non-melonic interactions carry the same weight as melonic ones in the power counting. In a field theory framework, such a 'boost' corresponds to a derivative coupling\footnote{This will be clarified in the next section.}, a situation highly reminiscent of non-Abelian gauge theories. Given the non-local structure of tensor fields, there is more than one possible derivative coupling. Consequently, we choose the most balanced one, which increases the weight of each color by one.
\begin{equation}\label{fou}
\int \prod_{i=1}^4 dg_idg_i'\sum_{i=1}^4(|\nabla_{g_i}|+|\nabla_{g'_i}|)\psi(g_1,g_2,g_3,g_4)\bar{\psi}(g_1',g_2',g_3,g_4)\psi(g_1',g_2',g_3',g'_4)\bar{\psi}(g_1,g_2,g_3',g'_4),
\end{equation}
where $\nabla_{g_i}$ denotes the gradient on $U(1)$ acting on the variable $g_i$. In Fourier space, denoting the Fourier coefficients of $\psi$ by $T_{\vec{p}}$, \eqref{fou} is reduced to:
\begin{equation}
\sum_{\vec{p},\vec{p}^{\prime}\in\mathbb{Z}^4}T_{p_1,p_2,p_3,p_4}\bar{T}_{p_1',p_2',p_3,p_4}T_{p_1',p_2',p_3',p_4'}\bar{T}_{p_1,p_2,p_3',p_4'}\sum_{i=1}^4(|p_i|+|p_i'|),
\end{equation}
which we denote by $\Tr_{b_{12}}^{\prime}[\psi,\bar{\psi}]$. The following interaction term $S_{int}$ is then chosen as:
\begin{equation}\label{interaction}
S_{int}[\psi,\bar{\psi}]:=\lambda_1\sum_{i=1}^4\Tr_{b_i}[\psi,\bar{\psi}]+\lambda_2\sum_{i=2}^4\Tr_{b_{1i}}^{\prime}[\psi,\bar{\psi}],
\end{equation}

Thus, all the components of our theory are defined by Eq. \eqref{interaction} and Eq. \eqref{skin} at least formally, if one disregards the divergences typically occurring in the computation of Feynman amplitudes.

\noindent

A regularization prescription is required in order to prevent the appearance of divergences in the course of computing Feynman amplitudes. For this purpose, we have chosen to work with Schwinger regularization. The propagator that is associated with the kinetic action (\ref{skin}) will be denoted by $C_{\vec{p},\vec{p}^{\prime}}$. In the momentum representation, this propagator is diagonal, meaning that it connects only momentum modes that are equal. Concretely, we have the relation $C_{\vec{p},\vec{p}^{\prime}} = \delta_{\vec{p},\vec{p}^{\prime}} C(\vec{p})$, where the Kronecker symbol $\delta_{\vec{p},\vec{p}^{\prime}}$ enforces the conservation of momentum, and $C(\vec{p})$ is a function that depends solely on $\vec{p}$. The latter is given by the following expression: \begin{equation}
C(\vec{p}\,)=\dfrac{1}{\sum_i|p_i|^{2\eta}+m^{2\eta}}=\int_{0}^{\infty}d\alpha \, e^{-\alpha \left(\sum_i|p_i|^{2\eta}+m^{2\eta}\right)},
\end{equation}

and the Schwinger regularization introduces a cut-off $1/\Lambda^{2\eta}$ in the $\alpha$-integration, thereby defining a regularized propagator $C_{\Lambda}$:
\begin{equation}
C_{\Lambda}(\vec{p}\,):=\int_{1/\Lambda^{2\eta}}^{\infty}d\alpha \, e^{-\alpha \left(\sum_i|p_i|^{2\eta}+m^{2\eta}\right)},
\end{equation}

so that the contribution of high momenta decreases exponentially. The perturbative theory is entirely captured by Feynman graphs, whose vertices are interaction bubbles. The edges, corresponding to Wick contractions between fields, are associated with an additional color, '0', and depicted by dotted lines. Consequently, such a Feynman diagram admits a natural colored extension as a bipartite, regular, 5-colored graph, an example of which is given in Figure \ref{fig2}. Within this perturbative framework, the connected Schwinger function $S_N$ with $N$ external lines can be expanded as a sum of Feynman amplitudes indexed by line-connected graphs $\mathcal{G}_N$:
\begin{equation}
S_N=\sum_{\mathcal{G}_N}\dfrac{1}{s(\mathcal{G}_N)}(-\lambda_1)^{V_1}(-\lambda_2)^{V_2}\mathcal{A}_{\mathcal{G}_N},
\end{equation}

where $V_i$ denotes the number of vertices of type $i$, $s(\mathcal{G}_N)$ is the symmetry factor, and $\mathcal{A}_{\mathcal{G}_N}$ is the Feynman amplitude, whose explicit expression follows directly from the Feynman rules:
\begin{align}\label{amplitude}
\mathcal{A}_{\mathcal{G}_N}=&\int_{1/\Lambda^{2\eta}}^{\infty}\prod_{l\in \mathcal{L}(\mathcal{G}_N)}d\alpha_le^{-\alpha_lm^{2\eta}}\prod_{f\in\mathcal{F}(\mathcal{G}_N)}\sum_{p_f\in\mathbb{Z}}p_f^{\eta(\partial f)}e^{-(\sum_{l\in\partial f}\alpha_l)p_f^{2\eta}}\cr
&\times \prod_{f\in\mathcal{F}_{ext}(\mathcal{G}_N)}p_f^{\eta(\partial f)}e^{-\left(\sum_{l\in\partial f}\alpha_l\right)p^{2\eta}_f},
\end{align}

where $\mathcal{F}$, $\mathcal{F}_{ext}$, and $\mathcal{L}$ denote the sets of internal faces, external faces, and lines of the graph $\mathcal{G}_N$, respectively. $F$ represents the cardinality of $\mathcal{F}$, while $\eta(\partial f)$ is the momentum power arising from the derivative couplings along the face $f$, satisfying: $\sum_f \eta(\partial f) = V_2$.
\begin{figure}
\begin{center}
\includegraphics[scale=0.8]{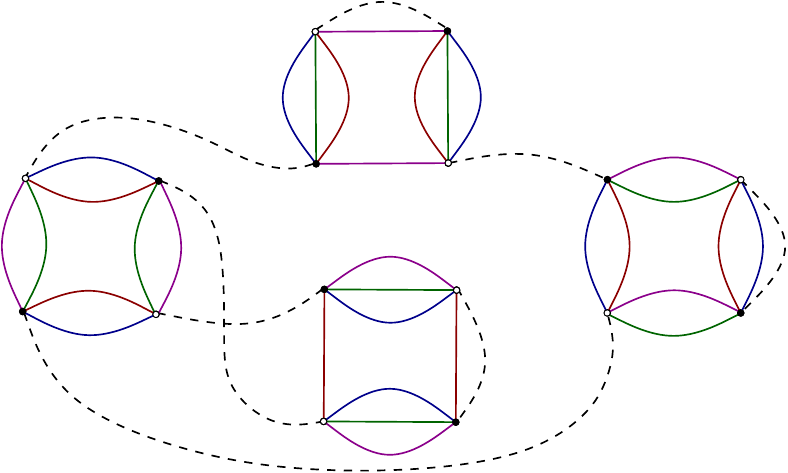}
\end{center}
\caption{Vacuum Feynman graph with four vertices.}\label{fig2}
\end{figure}

\section{Power counting}\label{sectionpower}

Power counting provides a criterion for classifying subgraphs as either divergent or convergent. Generally, there are two sources of divergences: the ultraviolet (UV) sector, occurring as $\alpha \to 0$, and the infrared (IR) sector, occurring as $\alpha \to \infty$. In our model, IR divergences are eliminated by the choice of a non-zero mass parameter and, further, by the compactness of the $U(1)$ group structure, consequently, divergences arise solely from the UV sector. In this section, we employ a slice decomposition multi-scale analysis to extract the exact power counting. Subsequently, we use this power counting to classify models according to the choice of parameters. Note that we extend the group structure to $U(1)^D$ for the more general case, as in references \cite{carrozza2014tensorial,samary2014just}.

\subsection{Multi-scale analysis}


With this choice of the group $U(1)^D$, the momenta $p_f$ appearing in Eq. \eqref{amplitude} become $D$-dimensional vectors, and the sums over each face are performed over $\mathbb{Z}^D$. The multi-scale decomposition introduces a slicing of the $\alpha$-integration. Choosing a scale parameter $M > 1$, we decompose the regularized propagator $C_{\Lambda}$ into slices $[M^{-2\eta i}, M^{-2\eta(i-1)}]$ as follows:
\begin{align}\label{slice}
&C_{0}(\mathbb{P})=\int_{1}^{\infty}d\alpha e^{-\alpha m^{2\eta}}\prod_{i=1}^4e^{-\alpha\sum_{I=1}^Dp_{iI}^{2\eta}},\\
\forall i\geq 1, \quad &C_i(\mathbb{P})=\int_{M^{-{2\eta}i}}^{M^{-{2\eta}(i-1)}}d\alpha e^{-\alpha m^{2\eta}}\prod_{i=1}^4e^{-\alpha\sum_{I=1}^Dp_{iI}^{2\eta}}, 
\end{align}
so that:
\begin{equation}
C_{\Lambda}=\sum_{i=0}^{\rho} C_i,
\end{equation}
where the integer $\rho$, the scale parameter $M$, and the cut-off $\Lambda$ are related by $\Lambda = M^{\rho}$. Furthermore, the elements of the $(4 \times D)$ matrix $\mathbb{P}$ appearing in Eq. \eqref{slice} are defined as $\mathbb{P}_{iI} = p_{iI}$. Notably, the propagator $C_i$ restricted to slice $i$ admits the following trivial uniform bound:
\begin{equation}\label{bound1}
|C_i(\mathbb{P})|\leq KM^{-2\eta i}e^{-\delta M^{-i}|\mathbb{P}|},
\end{equation}
for some positive constants $K$ and $\delta$, with the norm defined by $|\mathbb{P}| := \sum_{i,I} |p_{iI}|$. Given a tensor graph $\mathcal{G}$ and its corresponding amplitude $\mathcal{A}_{\mathcal{G}}$, one can define the multi-scale expansion: $\mathcal{A}_{\mathcal{G}} = \sum_{\mu} \mathcal{A}_{\mathcal{G},\mu}$, where $\mu = \{i_{l_1}, \dots, i_{l_{L(\mathcal{G})}}\}$ denotes a discrete scale attribution for each line of $\mathcal{G}$, and $\mathcal{A}_{\mathcal{G},\mu}$ is the amplitude restricted to that attribution. From the bound (\ref{bound1}) and the explicit expression (\ref{amplitude}), we obtain the following theorem:
\begin{theorem}\label{th1}
The amplitude $\mathcal{A}_{\mathcal{G},\mu}$ for a given scale attribution $\mu$ admits the following uniform bound:
\begin{equation}
|\mathcal{A}_{\mathcal{G},\mu}|\leq K\prod_{i}\prod_{k=1}^{k(i)}M^{\omega(\mathcal{G}_i^k)},
\end{equation}
where $K$ is a constant depending on the graph $\mathcal{G}$, and $\mathcal{G}_i^k$ denotes the $k$-th connected component of the subgraph $\mathcal{G}_i$, which contains only the lines with a scale assignment $j \geq i$. The divergent degree (or superficial degree of divergence) $\omega$ is then defined by:
\begin{equation}
\omega(\mathcal{G}_i^k)=-2\eta L(\mathcal{G}_i^k) +D\bigg(F(\mathcal{G}_i^k)+\sum_{f\in\mathcal{F}(\mathcal{G}_i^k)}\eta(\partial f)\bigg)\,.
\end{equation}
\end{theorem}
\textbf{Proof}: 
Thanks to the uniform bound (\ref{bound1}), the overall bound of the amplitude for a given scale attribution $\mu$ involves the following contribution:
\begin{equation}
\prod_{l\in\mathcal{L}(\mathcal{G})}M^{-2\eta i_l}=\prod_{l\in\mathcal{L}(\mathcal{G})}\prod_{i=1}^{i_l}M^{-2\eta}\,.
\end{equation}
We then define the subgraphs $\mathcal{G}_i^k$ as the connected components, labeled by $k=1, \dots, k(i)$, of the subgraph of $\mathcal{G}$ consisting of all lines with a scale assignment $j \geq i$. From this definition, we obtain:
\begin{align}
\prod_{l\in\mathcal{L}(\mathcal{G})}M^{-2\eta i_l}=\prod_i\prod_{l\in\mathcal{L}(\cup_{k=1}^{k(i)}\mathcal{G}_i^k)}M^{-2\eta}=\prod_{i,k}M^{-2\eta L(\mathcal{G}_i^k)}\,.\label{cont11}
\end{align}
A second contribution arises from the summation over the momenta along each face. Letting $i(f)$ denote the scale assignment of the lowest-scale line in the boundary of face $f$, we obtain the following trivial bounds:
\begin{equation}
\prod_{l\in\partial f}e^{-M^{-i_{l}}\sum_{I=1}^D|p_{f,I}|}\leq e^{-M^{-i(f)}\sum_{I=1}^D|p_{f,I}|}\,.
\end{equation}
With this choice, the factor generated by the summation over momenta is minimized, ensuring that the bound is optimal. The sums over $p_f$ produce factors that depend on $\eta(\partial f)$, which can be absorbed into a global constant specific to the graph. Consequently, we only consider the scaling dependence on $M$. Since:
\begin{equation}
\sum_{p\in\mathbb{Z}}|p|^n e^{-K|p|}\simeq K^{n+1}\,,
\end{equation}
one finds the following contribution:
\begin{equation}
\prod_{f\in\mathcal{F}} M^{D(\eta(\partial f)+1)i_f}\,,
\end{equation}
and following the same strategy as for the contribution (\ref{cont11}), one obtains:
\begin{align}
\nonumber\prod_{f\in\mathcal{F}} M^{D(\eta(\partial f)+1)i_f}&=\Big[\prod_{i,k}\prod_{f\in\mathcal{F}(\mathcal{G}_i^k)}M^D\Big]\times \Big[\prod_{i,k}\prod_{f\in\mathcal{F}(\mathcal{G}_i^k)}M^{D\eta(\partial f)}\Big]\\
&=\prod_{i\in\mathbb{N}}\prod_{i=1}^{k(i)}M^{D\big(F(\mathcal{G}_i^k)+\sum_{f\in\mathcal{F}(\mathcal{G}_i^k)}\eta(\partial f)\big)}\,.\label{cont222Bis}
\end{align}
\noindent
Finally, merging the contributions \eqref{cont11} and \eqref{cont222Bis} yields the scaling dependence on $M$ as stated in the theorem.
\begin{flushright}
$\square$
\end{flushright}

\subsection{Classification of models: The necklace bubbles}

Consider a vacuum graph $\mathcal{G}$ and its associated colored extension $\mathcal{G}_c$, see appendix \ref{App0}, definition \ref{coloreddef}, for readers unfamiliar with these notions. Our objective is to derive an expression for the number of faces $F(\mathcal{G})$ which corresponds, up to the number of external faces, to the count of faces of color type $0i$ ($i \neq 0$) in $\mathcal{G}_c$, denoted as $|\mathcal{F}_{0i}(\mathcal{G}_c)|$ in terms of the Gurau degree $\varpi(\mathcal{G}_c)$. Lemma \ref{propdeg} of appendix \ref{App0} provides the fundamental link between the face count and the degree. Furthermore, because of the definition \ref{coloreddef} of the appendix, the total number of faces in the colored extension $|\mathcal{F}(\mathcal{G}_c)|$ decomposes as:

\begin{equation}
|\mathcal{F}(\mathcal{G}_{c})|=F(\mathcal{G})+|\mathcal{F}_c^{\neq 0}(\mathcal{G}_c)|\,.
\end{equation}

The faces in the set $\mathcal{F}_c^{\neq 0}(\mathcal{G}_c)$ do not contain any edges of color $0$. Let $\mathcal{G}_c^{(\hat{0})}$ be the subgraph obtained from $\mathcal{G}_c$ by deleting all edges of color $0$. This subgraph consists of $|\mathcal{V}(\mathcal{G})|$ connected components, denoted as $\mathcal{B}_{k}$, where $k$ labels the original vertices of $\mathcal{G}$. Each component $\mathcal{B}_{k}$ represents a $3$-bubble (an interaction vertex). The total number of internal faces $|\mathcal{F}_c^{\neq 0}(\mathcal{G}_c)|$ is the sum of the faces within these bubbles. Applying Lemma \ref{propdeg} for $d=3$ to each bubble, we deduce that:
\begin{equation}
|\mathcal{F}_c^{\neq 0}(\mathcal{G}_c)|=\sum_{k=1}^{V(\mathcal{G})}\bigg(3p_k+3-\varpi(\mathcal{G}_c^{(0)k})\bigg)\,,
\end{equation}
and because : $\sum_{k=1}^{V(\mathcal{G})}p_k=L(\mathcal{G})$, we obtain:
\begin{equation}
|\mathcal{F}_c^{\neq 0}(\mathcal{G}_c)|=\bigg(3L(\mathcal{G})+3V(\mathcal{G})-\sum_{k=1}^{V(\mathcal{G})}\varpi(\mathcal{G}_c^{(0)k})\bigg)\,.
\end{equation}
Then, using the lemma \ref{propdeg}:
\begin{equation}
F(\mathcal{G})=3\big[L(\mathcal{G})-V(\mathcal{G})+1\big]+1+\bigg[\sum_{k=1}^{V(\mathcal{G})}\varpi(\mathcal{G}_c^{(0)k})-\frac{1}{3}\varpi(\mathcal{G}_c)\bigg]\,,
\end{equation}
and finally:
\begin{equation}
F(\mathcal{G})=3\big[L(\mathcal{G})-V(\mathcal{G})+1\big]+1+\bigg[\sum_{k=1}^{V(\mathcal{G})}\varpi(\mathcal{G}_c^{(0)k})-\frac{1}{3}\varpi(\mathcal{G}_c)\bigg]\,.
\end{equation}
Using of the topological constraint $2V=L$, one finds, with the definition given in Theorem \ref{th1}:
\begin{equation}
\omega(\mathcal{G})=(3D-4\eta)V(\mathcal{G})+3D+DV_2(\mathcal{G})+D\sum_{k=1}^{V(\mathcal{G})}\varpi(\mathcal{G}_c^{(0)k})+D\rho(\mathcal{G})\,,
\end{equation}
where we make use of the constraint $\sum_f \eta(\partial f) = V_2$, and define:
\begin{equation}
\rho(\mathcal{G})=1-\frac{1}{3}\varpi(\mathcal{G}_c)\,.
\end{equation}
Because they are melonic, the Gurau degree vanishes for interaction bubbles of type \ref{fig1}a ($\varpi(\mathcal{B}_1) = 0$). Conversely, it is straightforward to show that the degree is equal to $3$ for interaction bubbles of type \ref{fig1}b ($\varpi(\mathcal{B}_2) = 3$). Consequently, we have:
\begin{equation}
\sum_{k=1}^{V(\mathcal{G})}\varpi(\mathcal{G}_c^{(0)k})=3V_2(\mathcal{G})\,,
\end{equation}
and:
\begin{equation}
\omega(\mathcal{G})=(3D-4\eta)V(\mathcal{G})+3D+4DV_2(\mathcal{G})+D\rho(\mathcal{G})\,.
\end{equation}
We restrict our attention to the case $D=1$, and the previous analysis can be extended to non-vacuum graphs in a similar fashion. Consider $\mathcal{G}$ as a non-vacuum graph. A vacuum graph $\hat{\mathcal{G}}$ is called a closure of $\mathcal{G}$ if $\mathcal{G}$ can be obtained from $\hat{\mathcal{G}}$ by cutting a subset of its internal lines. Furthermore, the closure is said to be maximal if the resulting number of faces is maximized. For a maximal closure $\hat{\mathcal{G}}$, it follows that:
\begin{equation}
\omega(\hat{\mathcal{G}})=\omega(\mathcal{G})-\frac{3}{4}N(\mathcal{G})+\Delta F(\hat{\mathcal{G}})+\Delta\eta(\hat{\mathcal{G}})\,,
\end{equation}
where $\Delta F(\hat{\mathcal{G}})$ and $\Delta\eta(\hat{\mathcal{G}})$ represent the variation in the number of faces and the number of derivative couplings contained within those faces, respectively. Consequently, it follows that:
\begin{equation}
\Delta F(\hat{\mathcal{G}})=F_{ext}-\mathcal{C}(\hat{\mathcal{G}}/\mathcal{G})=\dfrac{4}{2}N(\mathcal{G})-\mathcal{C}(\hat{\mathcal{G}}/\mathcal{G})\,,
\end{equation}
where $\mathcal{C}(\hat{\mathcal{G}}/\mathcal{G})$ denotes the difference between the number of face-connected components of the graph $\mathcal{G}$ and its maximal closure $\hat{\mathcal{G}}$. Specifically, we define:$$\mathcal{C}(\hat{\mathcal{G}}/\mathcal{G}) = \text{Comp}_f(\mathcal{G}) - \text{Comp}_f(\hat{\mathcal{G}})\,,$$
This term accounts for the topological fragmentation that occurs when internal lines are cut to form external legs. As a direct consequence, the degree of divergence $\omega(\mathcal{G})$ of a graph $\mathcal{G}$ in the $1/N$ expansion is given by:
\begin{equation}
\omega(\mathcal{G}) = \omega_{\text{melo}} - \frac{2}{(d-1)!} \varpi(\hat{\mathcal{G}}_c) - \mathcal{C}(\hat{\mathcal{G}}/\mathcal{G})\,,
\end{equation}
where $\omega_{\text{melo}}$ is the perturbative degree of divergence of the corresponding melonic vacuum diagrams.\\

\paragraph{Melonic sector.} 
First, we investigate the melonic sector, that implies $\varpi(\mathcal{G}_c) = 0$ and $\varpi(\mathcal{G}_c^{(0)k}) = 0$ for all $k$, which consequently leads to $V_2 = 0$. Furthermore, since melonic graphs are face-connected (see Corollary \ref{cor1}), a moment of reflection shows that $\Delta F(\hat{\mathcal{G}}) = 3N/2 + 1$. Substituting these values into our general expression, we find:
\begin{equation}
\omega(\mathcal{G})=(3-4\eta)V(\mathcal{G})-\frac{N(\mathcal{G})}{2}(3-2\eta)+3\,.
\end{equation}
The theory is then renormalizable only if $3 - 4\eta \leq 0 \implies \eta \geq 3/4$. Within the domain $\eta \in [3/4, 1]$, the condition $3 - 2\eta \geq 1$ holds. For the specific value $\eta = 3/4$, the melonic sector becomes just-renormalizable, and the degree of divergence reduces to:
\begin{equation}
\omega(\mathcal{G})=3-\frac{3}{4}N(\mathcal{G})\,.
\end{equation}
As a result, only graphs with 2 and 4 external lines remain divergent. For $\eta > 3/4$, the degree of divergence depends explicitly on the number of vertices, making the theory super-renormalizable. At this stage, it is essential to clarify the claim that the melonic sector is just-renormalizable. This stems from the fact that the melonic sector is, in a sense, isolated from other contributions: the contraction of a melonic subgraph cannot generate an effective vertex of Type 2. Consequently, the renormalization flow remains confined within the melonic boundary.

\noindent
\paragraph{Necklace sector.} The second limiting case that can be easily investigated is the 'matrix sector', where $V = V_2$. In this case, the topology is dominated by matrix-like ribbon structures rather than melonic trees, and the degree of divergence becomes:
\begin{equation}
\omega(\mathcal{G})\leq(7-4\eta)V(\mathcal{G})+\bigg[(1+\rho)-N(\mathcal{G})(1-\eta)\bigg]\,.
\end{equation}
It can be proved, from the lemma \ref{propdeg} that : $\varpi(\mathcal{G}_c)\geq 4\sum_k\sum_{k=1}^{V(\mathcal{G})}\varpi(\mathcal{G}_c^{(0)k})=12V$, and:
\begin{equation}
\omega(\mathcal{G})\leq(3-4\eta)V(\mathcal{G})+\bigg[2-N(\mathcal{G})(1-\eta)\bigg]\,.
\end{equation}
Once again, the dependence on the vertex count is eliminated by choosing $\eta = 3/4$. This choice appears to be the most natural for fostering the proliferation of both phases, as it ensures that the different topological sectors compete on an equal footing at the critical point.\\

The preceding analyses demonstrate that for $\eta = 3/4$, the theory becomes just-renormalizable within these two limiting sectors. In the melonic case, the renormalizable sector is restricted to quartic interactions ($p=2$). Conversely, in the 'matrix' limit, the renormalizable sector is limited to octic interactions ($p=4$).

\subsection{Canonical dimension}
\label{sectioncanonical}

The canonical dimension corresponds to the power of $e^{s}$ by which the couplings must be multiplied to make the flow equations autonomous (i.e., independent of $s$). One way to determine these powers is to identify the transformations that render the flow equations autonomous, which can be done directly from the equations derived in the following section. They can also be fixed beforehand, thereby avoiding the ambiguities that sometimes arise with the former strategy. First, let us note that, unlike standard field theories defined on a background spacetime, there is no notion of an external scale to the flow here. In other words, it is the behavior of the flow in the vicinity of the Gaussian fixed point and more precisely its local structure determines these dimensions. This behavior differs across regimes (UV or IR), and there is no global notion of dimension: it is scale-dependent. An extended discussion can be found in \cite{Lahoche:2018oeo,Benedetti_2015}, as well as in references \cite{lahoche2024functional}.\\

Following the method described in \cite{Carrozza_2017a}, we define the canonical dimension from the (saturated) superficial degree of divergence, setting $\eta=3/4$. Thus, for any melonic bubble $b$ with coupling $\lambda_b$,
\begin{equation}
[ \lambda_b] := 3-\frac{3}{4} N(b)\,,
\end{equation}

 where $N(b)$ denotes the number of fields in the bubble $b$. The quartic melon is therefore dimensionless. Similarly, for any necklace bubble $n$:
\begin{equation}
[\lambda_n]:= 2-N(b)/4\,,
\end{equation}
and the octic necklace interaction is dimensionless. Under this definition, renormalizable interactions always correspond to dimensionless couplings, and for the rest of this paper, we shall restrict our focus to the renormalizable sector.

Based on this power-counting, it is possible to provide a preliminary perturbative calculation of the renormalization group. This analysis, detailed in Appendix \ref{App1}, shows that: (i) the Gaussian fixed point is UV-unstable (not asymptotically free), raising the question of its UV completion, and (ii) non-Gaussian fixed points emerge, albeit beyond the strictly perturbative regime. Consequently, searching for fixed points requires the use of non-perturbative methods. It should be emphasized that the loss of asymptotic freedom is a novel feature brought about by the necklace enrichment, as discussed in  \cite{lahoche2015renormalization}.

\section{Non-perturbative functional renormalization group}
\label{sectionNP}

In this section, we apply the FRG formalism to the TGFTs introduced previously. We first derive the Wetterich equation within this specific tensorial context, subsequently employing it to extract the flow equations under a suitable truncation. For comprehensive reviews and diverse applications of the FRG, we refer the reader to \cite{Delamotte_2012}.

\subsection{Wetterich equation for tensor fields}

Starting from the generating functional defined in Eq. \eqref{genfunc}, we introduce the following one-parameter family of models:
\begin{equation}\label{family}
\mathcal{Z}_{s}[\bar{J},J]:=\int d\mu_{C}(\bar{\psi},\psi)e^{-S_{int}(\bar{\psi},\psi)-\Delta S_{s}[\bar{\psi},\psi]+\langle \bar{J},\psi \rangle+\langle\bar{\psi},J\rangle}.
\end{equation}
The additional term $\Delta S_{s}$, which we refer to as the IR regulator, is chosen to be ultralocal in the momentum (Fourier) representation and have the following structure:
\begin{equation}
\Delta S_{s}[\bar{\psi},\psi]:= \langle\bar{\psi}, R_{s} \psi \rangle= \sum_{\vec{p} \in \mathbb{Z}^4}R_{s}(\vec{p}\,)\bar{T}_{\vec{p}}T_{\vec{p}}.
\end{equation}

As is standard, in order for the partition function to be well-defined, we assume the presence of a UV regulator for instance, a sharp cutoff on the momenta $|p| \leq \Lambda$ (using the standard norm $|p| = \sqrt{\vec{p} \cdot \vec{p}}$). In practice, one often proceeds by taking the limit $\Lambda \to \infty$, as the Wetterich equation remains well-defined in this regime, even though its formal derivation from the path integral becomes more subtle.

\noindent

The function $R_{s}(\vec{p}\,)$ is positive definite, and we briefly recall its standard properties (where $\Lambda$ denotes the fundamental ultraviolet cutoff):

\noindent
$\bullet$ $R_{s}(\vec{p}\,)\geq 0$ for all $\vec{p}\in \mathbb{Z}^d$ and $s\in(-\infty,+\infty)$.\\

\noindent
$\bullet$ $\lim_{s\to-\infty} R_s(\vec{p}\,) =  0$: ensuring the boundary condition \eqref{family}: $\mathcal{Z}_{s=-\infty}[\bar{J},J]=\mathcal{Z}[\bar{J},J].$ Physically, it means that the original model is recovered when all the fluctuations are integrated out.  \\

\noindent
$\bullet$  $\lim_{s\to\ln\Lambda} R_s(\vec{p}\,) =  +\infty$, ensuring that all the fluctuations are frozen when $e^s=\Lambda$. As a consequence, the bare action will be represented by the initial condition for the flow at $s=\ln\Lambda$.\\

\noindent
$\bullet$ For $-\infty<s<\ln \Lambda$, the cutoff $R_{s}$ is chosen so that 
\begin{equation}
R_{s}(|p|>e^s)\ll 1\,,
\end{equation}
The first condition ensures that the UV modes ($|p| > e^s$) remain largely unaffected by the regulator, while the requirement that $R_{s}(|p| < e^s) \sim 1$ (or $R_{s}(|p| < e^s) \gg 1$) guarantees that IR modes ($|p| < e^s$) are effectively decoupled.\\

\noindent
Finally, in addition to the standard properties discussed above, we impose the requirement that $\partial_s R_s(\vec{p}) \leq 0$ for all $\vec{p} \in \mathbb{Z}^d$ and $s \in (-\infty, +\infty)$. This ensures that high-momentum modes are never suppressed more than low-momentum modes, thereby maintaining a consistent coarse-graining procedure throughout the RG flow.

\noindent

As it stands, \eqref{family} defines an infinite-dimensional deformation of the original partition function. However, the specific choice of the IR regulator, provided it satisfies the aforementioned requirements, is secondary to the parametric dependence on $s$. From a Wilsonian perspective, the former represents a choice of coarse-graining scheme, while the latter corresponds to the coarse-graining scale. Since our primary interest lies in the scale dependence of the theory, we assume that a specific cutoff function has been fixed and treat \eqref{family} as a one-parameter family of theories. Although approximation schemes generally introduce a spurious dependence of physical quantities on the chosen regularization, we will not discuss this point further here and refer the reader to the literature on optimization techniques \cite{DePolsi:2022wyb,Balog_2019,canet2003optimization,Berges_2002}.

\noindent
We define the one-parameter family of connected Schwinger functionals (or free energies) as:
\begin{equation}\label{free}
W_{s}:=\log \mathcal{Z}{s}[\bar{J},J],
\end{equation}
and, by taking their Legendre transform, a one-parameter family of effective actions $\Gamma{s}$, referred to as the effective average action and defined as:
\begin{equation}\label{legendre}
\Gamma_{s}[\bar{\phi},\phi]+\langle \bar{\phi}, R_{s}\phi\rangle=\langle \bar{J}, \phi\rangle+\langle \bar{\phi},J \rangle-W_{s}[\bar{J},J],
\end{equation}
where the source $J$ is expressed as a function of the classical field (or mean field) $\phi$ by inverting the relation:
\begin{equation}\label{meanfield}
\phi =\dfrac{\delta W_{s}}{\delta \bar{J}}.
\end{equation}

Interestingly, when translated into the effective average action framework, the properties of the IR regulator $\Delta S_{s}$ imply the following limits:

\noindent
$\bullet$ $\Gamma^{int}_{s=\ln \Lambda}=S_{int}$, so that when all the fluctuations are frozen, the effective average action coincides with the initial microscopic action.\\

\noindent
$\bullet$ $\Gamma_{s=-\infty}=\Gamma$, meaning that when all the fluctuations are integrated out, the effective average action coincides with the full effective action. \\

\noindent
By differentiating the effective average action \eqref{legendre} with respect to the renormalization scale $s$, we obtain the exact functional renormalization group equation, known as the Wetterich equation:
\begin{equation}\label{Wetterich}
\partial_s\Gamma_s=\sum_{\vec{p}\in \mathbb{Z}^4}\partial_s R_s(\vec{p}\,)[\Gamma_s^{(2)}+R_s]^{-1}(\vec{p},\vec{p}\,) \quad where \quad 
\Gamma_s^{(2)}(\vec{p},\vec{p}\,^\prime)=\frac{\partial^2\Gamma_s}{\partial T_{\vec{p}}\partial\bar{T}_{\vec{p}\,'}}\,.
\end{equation}
In this equation, $\Gamma_s^{(2)} + R_s$ corresponds to the inverse of the connected two-point function $\delta^2 W_s[J, \bar{J}] / \delta \bar{J} \delta J$. Extracting non-perturbative information from the exact flow equation \eqref{Wetterich} requires a suitable approximation scheme, typically involving the restriction of the flow to a finite-dimensional functional subspace. A standard strategy, known as the truncation method, consists of proposing an ansatz for $\Gamma_s$ that recovers the original action in the UV limit. In this approach, one retains only the terms on the right-hand side of the Wetterich equation that project onto the operators already present in the chosen ansatz. For our purposes, we adopt the following truncation:
\begin{align}\label{ansatz} 
\nonumber\Gamma_s[T,\bar{T}]:=\sum_{\vec{p}\in\mathbb{Z}^4}&\bar{T}_{\vec{p}}\bigg(Z_1(s)\sum_{i=1}^4|p_i|^{2\eta}+Z_2(s)\sum_{i=1}^4|p_i|+m^{2\eta}(s)\bigg)T_{\vec{p}}+\bigg[\lambda_1(s)\sum_{i=1}^4\sum_{\{\vec{p}_l\}}\mathcal{W}^{(i), {\vec{\bar{p}}}_1,{\vec{\bar{p}}}_2}_{melo\,\vec{p}_1,\vec{p}_2}\\\nonumber
&+\sum_{i=2}^4\sum_{\{\vec{p}_l\}}\mathcal{W}^{(1i), {\vec{\bar{p}}}_1,{\vec{\bar{p}}}_2}_{neck\,\vec{p}_1,\vec{p}_2}\big(\lambda_2(s)+\lambda_3(s)\sum_{i=1}^4(|p_{1i}|+|p_{2i}|)\big)\bigg]\prod_{l=1}^2T_{\vec{p}_l}\bar{T}_{{\vec{\bar{p}}}_l}\\\nonumber
&+\lambda_4\sum_{i=2}^4\sum_{\{\vec{p}_l\}}\mathcal{W}^{(1i), {\vec{\bar{p}}}_1,{\vec{\bar{p}}}_2,{\vec{\bar{p}}}_3}_{neck\,\vec{p}_1,\vec{p}_2,\vec{p}_3}\prod_{l=1}^3T_{\vec{p}_l}\bar{T}_{{\vec{\bar{p}}}_l}+\lambda_5\sum_{i=2}^4\sum_{\{\vec{p}_l\}}\mathcal{W}^{(1i), {\vec{\bar{p}}}_1,{\vec{\bar{p}}}_2,{\vec{\bar{p}}}_3,{\vec{\bar{p}}}_4}_{neck\,\vec{p}_1,\vec{p}_2,\vec{p}_3,\vec{p}_4}\prod_{l=1}^4T_{\vec{p}_l}\bar{T}_{{\vec{\bar{p}}}_l}\\
&=:\Gamma_{s\,kin}[T,\bar{T}]+\Gamma_{s\,int}[T,\bar{T}]\,,
\end{align}
where the kernels $\mathcal{W}^{(i), {\vec{\bar{p}}}_1,{\vec{\bar{p}}}_2}_{melo\,\vec{p}_1,\vec{p}_2}$ and $\mathcal{W}^{(1i), {\vec{\bar{p}}}_1,{\vec{\bar{p}}}_2}_{neck\,\vec{p}_1,\vec{p}_2}$, ..., are the product of Kronecker deltas that materialize the melonic and necklace interactions as defined above. The final step involves the choice of the regulator $R_s$, for which we adopt the following slightly modified (spectrally adapted) Litim-type cutoff \cite{litim2000optimisation}:
\begin{equation}\label{regulatorchoice}
R_s(\vec{p}\,)=\bigg[Z_1\bigg(e^{2\eta s}-\sum_{i=1}^4|p_i|^{2\eta}\bigg)+Z_2\bigg(e^s-\sum_{i=1}^4|p_i|\bigg)\bigg]\Theta\bigg(e^{2\eta s}-\sum_{i=1}^4|p_i|^{2\eta}\bigg)\,,
\end{equation}
where $\Theta$ denotes the Heaviside step function. The specific momentum dependence, combined with the inclusion of the wave function renormalization factor $Z_s$, leads to substantial simplifications that facilitate analytical calculations. Given these choices for the truncation and the regulator, we are now in a position to derive the flow equations for each parameter from the Wetterich equation \eqref{Wetterich}, which is the subject of the next section. For the sake of clarity, we shall derive these flow equations in two asymptotic regimes: the deep UV limit (large cutoff) and the deep IR limit (small cutoff). In the former, only the leading-order (LO) graphs are retained, while in the latter, the discrete sums reduce to their simplest terms (zero or one).

\begin{remark}
In the following, we shall employ $Z_1$ as the 'true' wave function renormalization, from which the renormalized coupling constants are defined. Consequently, $Z_2$ will be treated as an additional coupling. This approach is consistent in the UV limit, as the choice $\eta = 3/4$ determines the renormalizability of the theory and thus governs the Gaussian behavior in the ultraviolet. However, it should be noted that renormalizing with respect to $Z_2$ could be of interest in the infrared regime, which will be the subject of a subsequent analysis.
\end{remark}

\subsection{Flow equations in the deep UV limit}

We define:
\begin{align}
[\mathcal{K}_s]_{\vec{p},\vec{p}\,'}&=\bigg[Z_1(s)\sum_{i=1}^4|p_i|^{2\eta}+Z_2(s)\sum_{i=1}^4|p_i|+m^{2\eta}(s)\bigg]\delta_{\vec{p},\vec{p}\,'},\\
[\mathcal{F}_s]_{\vec{p},\vec{p}\,'}&=\frac{\partial^2\Gamma_{s\,int}}{\partial T_{\vec{p}}\partial\bar{T}_{\vec{p}\,'}}=\sum_{\mu=1}^3[\mathcal{F}_{s,2\mu}]_{\vec{p},\vec{p}\,'},
\end{align}
where $\mathcal{F}_{s,2\mu}$ denotes the contribution to $\mathcal{F}_{s}$ involving $2\mu$ fields. Thus, we have $\Gamma^{(2)}_s + R_s = \mathcal{K}_s + \mathcal{F}_s$. Letting $\Gamma_{s,2\mu}$ represent the terms in the truncated effective action \eqref{ansatz} containing $2\mu$ fields, we expand the right-hand side of the Wetterich equation \eqref{Wetterich} in powers of $\mathcal{F}_s$. We then obtain, in compact form (where the functional trace $\text{Tr}$ acts over matrices with indices in $\mathbb{Z}^4$):
\begin{align}\label{flow1}
\partial_s\Gamma_{s,2}&=-\Tr\big(\partial_sR_s\mathcal{K}_s^{-1}\mathcal{F}_{s,2}\mathcal{K}_s^{-1}\big),\\
\partial_s\Gamma_{s,4}&=-\Tr\big(\partial_sR_s\mathcal{K}_s^{-1}\mathcal{F}_{s,4}\mathcal{K}_s^{-1}\big)+\Tr\big(\partial_sR_s\mathcal{K}_s^{-1}\mathcal{F}_{s,2}\mathcal{K}_s^{-1}\mathcal{F}_{s,2}\mathcal{K}_s^{-1}\big),\\
\partial_s\Gamma_{s,6}&=-\Tr\big(\partial_sR_s\mathcal{K}_s^{-1}\mathcal{F}_{s,6}\mathcal{K}_s^{-1}\big)+\Tr\big(\partial_sR_s\mathcal{K}_s^{-1}\mathcal{F}_{s,2}\mathcal{K}_s^{-1}\mathcal{F}_{s,4}\mathcal{K}_s^{-1}\big)+ 2\leftrightarrow 4 \\\nonumber
&\quad-\Tr\big(\partial_sR_s\mathcal{K}_s^{-1}\mathcal{F}_{s,2}\mathcal{K}_s^{-1}\mathcal{F}_{s,2}\mathcal{K}_s^{-1}\mathcal{F}_{s,2}\mathcal{K}_s^{-1}\big),\\
\partial_s\Gamma_{s,8}&=\Tr\big(\partial_sR_s\mathcal{K}_s^{-1}\mathcal{F}_{s,2}\mathcal{K}_s^{-1}\mathcal{F}_{s,6}\mathcal{K}_s^{-1}\big)+ 2\leftrightarrow 6+\Tr\big(\partial_sR_s\mathcal{K}_s^{-1}\mathcal{F}_{s,4}\mathcal{K}_s^{-1}\mathcal{F}_{s,4}\mathcal{K}_s^{-1}\big)\\\nonumber
&\quad-\Tr\big(\partial_sR_s\mathcal{K}_s^{-1}\mathcal{F}_{s,2}\mathcal{K}_s^{-1}\mathcal{F}_{s,2}\mathcal{K}_s^{-1}\mathcal{F}_{s,4}\mathcal{K}_s^{-1}\big)+\sym(2\longrightarrow 2 \longrightarrow 4)\\\nonumber
&\quad+\Tr\big(\partial_sR_s\mathcal{K}_s^{-1}\mathcal{F}_{s,2}\mathcal{K}_s^{-1}\mathcal{F}_{s,2}\mathcal{K}_s^{-1}\mathcal{F}_{s,2}\mathcal{K}_s^{-1}\mathcal{F}_{s,2}\mathcal{K}_s^{-1}\big).
\end{align}

We shall now expand these expressions to extract the flow equations for each parameter. We focus specifically on the derivation of the flow equations in the deep UV limit.

\subsubsection{Derivation of flow equations}

Computing $\partial_sR_s$, we find:
\begin{align}\label{regulatorderive}
\nonumber &\partial_sR_s(\vec{p}\,)=\bigg[\bigg(\partial_sZ_1(s)\bigg(e^{2\eta s}-\sum_{i=1}^4|p_i|^{2\eta}\bigg)+2\eta Z_1(s)e^{2\eta s}\bigg)+\bigg(\partial_sZ_2(s)\bigg(e^{s}-\sum_{i=1}^4|p_i|\bigg)\\
&\quad+Z_2(s)e^{s}\bigg)\bigg]\Theta\bigg(e^{2\eta s}-\sum_{i=1}^4|p_i|^{2\eta}\bigg)
+2\eta e^{2\eta s}Z_2(s)\bigg(e^{s}-\sum_{i=1}^4|p_i|\bigg)\delta\bigg(e^{2\eta s}-\sum_{i=1}^4|p_i|^{2\eta}\bigg)\,,
\end{align}\label{regulatorderive}
\noindent
and it will be useful to define:
\begin{align}
\nonumber \bar{\partial}_sR_s(\vec{p})\,:=&\bigg[\bigg(\partial_sZ_1(s)\bigg(e^{2\eta s}-\sum_{i=1}^4|p_i|^{2\eta}\bigg)+2\eta Z_1(s)e^{2\eta s}\bigg)+\bigg(\partial_sZ_2(s)\bigg(e^{s}-\sum_{i=1}^4|p_i|\bigg)\\
&\quad+Z_2(s)e^{s}\bigg)\bigg]\Theta\bigg(e^{2\eta s}-\sum_{i=1}^4|p_i|^{2\eta}\bigg)\,,
\end{align}
and
\begin{equation}
\partial \partial_sR_s(\vec{p}\,)\,:=2\eta e^{2\eta s}Z_2(s)\bigg(e^{s}-\sum_{i=1}^4|p_i|\bigg)\delta\bigg(e^{2\eta s}-\sum_{i=1}^4|p_i|^{2\eta}\bigg)\,.
\end{equation}
Interestingly, the presence of the Heaviside step function restricts the summations on the right-hand side of the Wetterich equation to a finite domain. Within this region, the regulated kinetic term $\mathcal{K}_s$ becomes independent of the momentum $\vec{p}$:
\begin{equation}
\mathcal{K}_s=Z_1(s)e^{2\eta s}+Z_2e^s+m^{2\eta}(s).
\end{equation}
\noindent
\paragraph{Flow equations for $2$-valent bubbles couplings}
We begin with the $2$-valent bubble couplings, namely $m(s)$, $Z_1(s)$, and $Z_2(s)$. Accounting for the momentum independence of $\mathcal{K}_s$, the first flow equation \eqref{flow1} can be written as:
\begin{align}
\nonumber&\sum_{\vec{p}\in\mathbb{Z}^4}\bar{T}_{\vec{p}}\bigg(\partial_sZ_1\sum_{i=1}^4|p_i|^{2\eta}+\partial_sZ_2\sum_{i=1}^4|p_i|+\partial_sm^{2\eta}\bigg)T_{\vec{p}}=-2\Bigg\{\sum_{\vec{p}\in\mathbb{Z}^4}\frac{\bar{\partial}_sR_s(\vec{p}\,)+\partial \partial_sR_s(\vec{p}\,)}{[Z_1e^{2\eta s}+Z_2e^s+m^{2\eta}]^2}\\\nonumber
&\times \sum_{\vec{p}_1,\vec{p}_2}\bigg[\lambda_1(s)\sum_{i=1}^4\sym\mathcal{W}^{(i), {\vec{\bar{p}}}_1,{\vec{\bar{p}}}_2}_{melo\,\vec{p}_1,\vec{p}_2}+\sum_{i=2}^4\sym\mathcal{W}^{(1i), {\vec{\bar{p}}}_1,{\vec{\bar{p}}}_2}_{neck\,\vec{p}_1,\vec{p}_2}\big(\lambda_2(s)+\lambda_3(s)\sum_{i=1}^4(|p_{1i}|+|p_{2i}|)\big)\bigg]\Bigg\}T_{\vec{p}_1}\bar{T}_{\vec{p}_2},
\end{align}\label{equationflow2points}
with
\begin{equation}\label{one}
\sym\mathcal{W}^{{\vec{\bar{p}}}_1,{\vec{\bar{p}}}_2}_{melo\,\vec{p}_1,\vec{p}_2}:=\mathcal{W}^{{\vec{\bar{p}}}_1,{\vec{\bar{p}}}_2}_{\,\vec{p}_1,\vec{p}_2}+\mathcal{W}^{{\vec{{p}}}_1,{\vec{\bar{p}}}_2}_{\,\vec{\bar{p}}_1,\vec{p}_2}.
\end{equation}
\begin{figure}[htbp]
  \centering
\includegraphics[scale=0.9]{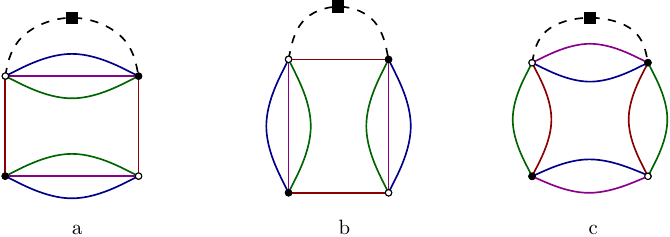} 
\caption{
Graphical representation of the typical contributions to the right-hand side of the flow equation \eqref{one}. The dotted line with a black square represents the insertion of the scale-derivative of the regulator, $\partial_s R_s$. Diagrams (a) and (b) illustrate the leading-order contributions arising from melonic interactions, characterized by their specific face-tracking structure. Diagram (c) depicts the contribution from necklace interactions, highlighting the different topological connectivity of these vertices.}
\label{fig3}  
\end{figure}

The r.h.s. of this equation naturally splits the contributions to $\partial_s Z_i(s)$ and $\partial_s m^{2\eta}(s)$ into two parts: the first arising from melonic interactions and the second from necklace interactions. We shall treat each contribution separately. These terms can be represented graphically as tadpoles, where the one-loop structure originates from the contraction of two vertices with the regulator insertion $\partial_s R_s(\vec{p}\,)\delta_{\vec{p}\vec{p}^\prime}$ (see Figure \ref{fig3}). Notably, diagram \ref{fig3}b is non-melonic and yields a sub-leading contribution, which may be neglected in the deep UV limit.

\noindent
$\bullet$\textbf{Melonic sector}.
We start with the melonic sector of equation \ref{one}, defined as:
\begin{align}\label{equationflow2pointsmelo}
\nonumber\sum_{\vec{p}\in\mathbb{Z}^4}&\bar{T}_{\vec{p}}\bigg(\partial_sZ_{1\,melo}\sum_{i=1}^4|p_i|^{2\eta}+\partial_sZ_{2\,melo}\sum_{i=1}^4|p_i|+\partial_sm^{2\eta}_{melo}\bigg)T_{\vec{p}}\\\nonumber
&=-\frac{2\lambda_1}{[Z_1e^{2\eta s}+Z_2e^s+m^{2\eta}]^2}\times \sum_{\vec{p}\in\mathbb{Z}^4}\Bigg\{\bigg[\bigg(\partial_sZ_1(s)\times
\bigg(e^{2\eta s}-\sum_{i=1}^4|p_i|^{2\eta}\bigg)+2\eta Z_1(s)e^{2\eta s}\bigg)\\\nonumber
&+\bigg(\partial_sZ_2(s)\bigg(e^{s}-\sum_{i=1}^4|p_i|\bigg)\quad+Z_2(s)e^{s}\bigg)\bigg]\Theta\bigg(e^{2\eta s}-\sum_{i=1}^4|p_i|^{2\eta}\bigg)
\\
&+2\eta e^{2\eta s}Z_2(s)\bigg(e^{s}-\sum_{i=1}^4|p_i|\bigg)\delta\bigg(e^{2\eta s}-\sum_{i=1}^4|p_i|^{2\eta}\bigg)\Bigg\}\sum_{\vec{p}_1,\vec{p}_2}\sum_{i=1}^4\mathcal{W}^{(i), {\vec{\bar{p}}}_1,{\vec{\bar{p}}}_2}_{melo\,\vec{p}_1,\vec{p}_2}T_{\vec{p}_1}\bar{T}_{\vec{p}_2}.
\end{align}
To extract the flow equations for $\partial_s Z_1$, $\partial_s Z_2$, and $\partial_s m^{2\eta}$, we introduce a specific set of test tensors: $T^{(k)}_{\vec{p}} := \prod_{j=1}^4 \delta_{p_j, k}$. For these configurations, the l.h.s. of equation \eqref{equationflow2pointsmelo} takes the form:
\begin{align}
\nonumber \sum_{\vec{p}\in\mathbb{Z}^4}\bar{T}_{\vec{p}}\bigg(\partial_sZ_{1\,melo}&\sum_{i=1}^4|p_i|^{2\eta}+\partial_sZ_{2\,melo}\sum_{i=1}^4|p_i|+\partial_sm^{2\eta}_{melo}\bigg)T_{\vec{p}}\\
&=4\partial_sZ_{1\,melo}|k|^{2\eta}+4\partial_sZ_{2\,melo}|k|+\partial_sm^{2\eta}_{melo},\label{projection}
\end{align}
and the r.h.s becomes:
\begin{align}
r.h.s=-\frac{2\lambda_1(s)}{[Z_1e^{2\eta s}+Z_2e^s+m^{2\eta}]^2}\big[(r.h.s)_1(k)+(r.h.s)_2(k)+(r.h.s)_3(k)\big],
\end{align}
where $\vec{k}:=(k,k,k,k)$ and:
\begin{equation}
(r.h.s)_1(k):=\sum_{\sum_{i=1}^4|p_i|^{2\eta}\leq e^{2\eta s}}\sum_{i=1}^4\mathcal{W}^{(i), {\vec{\bar{p}}}_1,{\vec{k}}}_{melo\,\vec{p}_1,\vec{k}}\bigg(\partial_sZ_1\bigg(e^{2\eta s}-\sum_{i=1}^4|p_i|^{2\eta}\bigg)+2\eta Z_1e^{2\eta s}\bigg),
\end{equation}
\begin{equation}
(r.h.s)_2(k):=\sum_{\sum_{i=1}^4|p_i|^{2\eta}\leq e^{2\eta s}}\sum_{i=1}^4\mathcal{W}^{(i), {\vec{\bar{p}}}_1,{\vec{k}}}_{melo\,\vec{p}_1,\vec{k}}\bigg(\partial_sZ_2\bigg(e^{s}-\sum_{i=1}^4|p_i|\bigg)+Z_2e^{s}\bigg),
\end{equation}
\begin{equation}
(r.h.s)_3(k):=2\eta Z_2e^{2\eta s}\sum_{\sum_{i=1}^4|p_i|^{2\eta}=e^{2\eta s}}\sum_{i=1}^4\mathcal{W}^{(i), {\vec{\bar{p}}}_1,{\vec{k}}}_{melo\,\vec{p}_1,\vec{k}}\bigg(e^{s}-\sum_{i=1}^4|p_i|\bigg).
\end{equation}
\noindent
Setting $k=0$ selecting only the contributions to $\partial_sm^{2\eta}$, we find:
\begin{align}
\nonumber \partial_sm^{2\eta}_{melo}=-&\frac{2\lambda_1(s)}{[Z_1e^{2\eta s}+Z_2e^s+m^{2\eta}]^2}\big[(r.h.s)_1(0)+(r.h.s)_2(0)+(r.h.s)_3(0)\big].
\end{align}
The r.h.s. involves several constrained summations, and we shall express the flow equations in terms of these quantities. In the melonic sector, the relevant sums are defined as follows:
\begin{align}\label{keysums}
S_1(k)&=\sum_{\sum_{i=2}^4|p_i|^{2\eta}\leq e^{2\eta s}-k^{2\eta}}1\,,\\
S_2(k)&=\sum_{\sum_{i=2}^4|p_i|^{2\eta}\leq e^{2\eta s}-k^{2\eta}}\sum_{i=2}^4|p_i|^{2\eta}\,,\\
S_3(k)&=\sum_{\sum_{i=2}^4|p_i|^{2\eta}\leq e^{2\eta s}-k^{2\eta}}\sum_{i=2}^4|p_i|\,,\\
\partial S_1(k)&=\sum_{\sum_{i=2}^4|p_i|^{2\eta}= e^{2\eta s}-k^{2\eta}}1\,,\\
\partial S_3(k)&=\sum_{\sum_{i=2}^4|p_i|^{2\eta}= e^{2\eta s}-k^{2\eta}}\sum_{i=2}^4|p_i|\,.
\end{align}
and $\partial_sm^{2\eta}_{melo}$ writes as:
\begin{align}\label{betammelon}
\nonumber\partial_sm^{2\eta}_{melo}=-&\frac{8\lambda_1(s)}{[Z_1e^{2\eta s}+Z_2e^s+m^{2\eta}]^2}\bigg\{(\partial_sZ_1e^{2\eta s}+\partial_sZ_2e^{s})S_1(0)-(\partial_sZ_1S_2(0)+\partial_sZ_2S_3(0))\\
&+(2\eta Z_1e^{2\eta s}+Z_2e^s)S_1(0)+2\eta Z_2(e^s\partial S_1(0)-\partial S_3(0))\bigg\}\,.
\end{align}
Following the same strategy, we extract the contributions to $\partial_s Z_{1, \text{melo}}$ and $\partial_s Z_{2, \text{melo}}$ by identifying the coefficients of $|k|^{2\eta}$ and $|k|$, respectively, on the r.h.s. of the flow equation. Let $S_i^{(2\eta)}(0)$ and $S_i^{(1)}(0)$ denote the coefficients of $k^{2\eta}$ and $k$ in the expansion of the sum $S_i(k)$ around the vanishing external momentum. We then find:
\begin{align}\label{anomalous1melon}
\nonumber\partial_sZ_{1\,melo}&=\dfrac{-2\lambda_1(s)}{[Z_1e^{2\eta s}+Z_2e^s+m^{2\eta}]^2}\bigg\{\partial_sZ_1\big(e^{2\eta s}S_1^{(2\eta)}-S_2^{(2\eta)}-S_1\big)+2\eta Z_1e^{2\eta s}S_1^{(2\eta)}\\
&+\partial_sZ_2\big(e^{s}S_1^{(2\eta)}-S_3^{(2\eta)}-S_1^{(2\eta-1)}\big)+Z_2e^{s}S_1^{(2\eta)}+2\eta Z_2\big(e^s\partial S_1^{(2\eta)}-\partial S_3^{(2\eta)}\big)\bigg\}\,,
\end{align}
\begin{align}\label{anomalous2melon}
\nonumber\partial_sZ_{2\,melo}&=\dfrac{-2\lambda_1(s)}{[Z_1e^{2\eta s}+Z_2e^s+m^{2\eta}]^2}\bigg\{\partial_sZ_1\big(e^{2\eta s}S_1^{(1)}-S_2^{(1)}-S_1^{(1-2\eta)}\big)+2\eta Z_1e^{2\eta s}S_1^{(1)}\\
&+\partial_sZ_2\big(e^{s}S_1^{(1)}-S_3^{(1)}-S_1\big)+Z_2e^{s}S_1^{(1)}+2\eta Z_2\big(e^s\partial S_1^{(1)}-\partial S_3^{(1)}\big)\bigg\}\,.
\end{align}
\noindent
$\bullet$\textbf{Necklace sector}
Turning now to the necklace contributions, we extract $\partial_s Z_{1, \text{neck}}$, $\partial_s Z_{2, \text{neck}}$, and $\partial_s m^{2\eta}_{\text{neck}}$ using the same strategy as for the melonic sector. For the necklace interaction associated with the coupling $\lambda_3$, the leading-order contributions arise from graphs involving a $|p_i|$ term in one of their two internal faces. Furthermore, since the canonical dimension of $\lambda_2$ is equal to $1$, all contributions involving this coupling must be retained. Consequently, by projecting the flow equation onto the test tensor $T^{(k)}$ and defining the following sums:
\begin{align}\label{keysums2}
S_4(k)&=\sum_{\sum_{i=3}^4|p_i|^{2\eta}\leq e^{2\eta s}-2k^{2\eta}}1\,,\\
S_5(k)&=\sum_{\sum_{i=3}^4|p_i|^{2\eta}\leq e^{2\eta s}-2k^{2\eta}}\sum_{i=3}^4|p_i|^{2\eta}\,,\\
S_6(k)&=\sum_{\sum_{i=3}^4|p_i|^{2\eta}\leq e^{2\eta s}-2k^{2\eta}}|p_3|\, \\
S_7(k)&=\sum_{\sum_{i=3}^4|p_i|^{2\eta}\leq e^{2\eta s}-2k^{2\eta}}\sum_{i=3}^4|p_i|^{2\eta}|p_3|\,,\\
S_8(k)&=\sum_{\sum_{i=3}^4|p_i|^{2\eta}\leq e^{2\eta s}-2k^{2\eta}}\sum_{i=3}^4|p_i||p_3|\,,\\
\partial S_4(k)&=\sum_{\sum_{i=3}^4|p_i|^{2\eta}= e^{2\eta s}-2k^{2\eta}}1\,,
\end{align}
\begin{align}
\partial S_6(k)&=\sum_{\sum_{i=3}^4|p_i|^{2\eta}= e^{2\eta s}-2k^{2\eta}}|p_3|\,,\\
\partial S_8(k)&=\sum_{\sum_{i=3}^4|p_i|^{2\eta}= e^{2\eta s}-2k^{2\eta}}\sum_{i=3}^4|p_i||p_3|\,.
\end{align}
we obtain the following expressions for $\partial_s m^{2\eta}_{\text{neck}}$, $\partial_s Z_{1, \text{neck}}$, and $\partial_s Z_{2, \text{neck}}$:
\begin{align}\label{betamneacklace}
\nonumber\partial_sm^{2\eta}_{\text{neck}}=&-\frac{12}{[Z_1e^{2\eta s}+Z_2e^s+m^{2\eta}]^2}\bigg\{\partial_sZ_1\big[2\lambda_3(e^{2\eta s}S_6-S_7)+\lambda_2(e^{2\eta s}S_4-S_5)\big]\\\nonumber
&+2\eta Z_1e^{2\eta s}(2\lambda_3S_6+\lambda_2S_4)+\partial_sZ_2\big[2\lambda_3(e^{2\eta s}S_6-S_8)+\lambda_2(e^{2\eta s}S_4-2S_6)\big]\\
&+Z_2e^{s}(2\lambda_3S_6+\lambda_2S_4)+2\eta Z_2\big[2\lambda_3\big(e^s\partial S_6-\partial S_8\big)+\lambda_2\big(e^s\partial S_4-2\partial S_6\big)\big]\bigg\}\,,
\end{align} 
\begin{align}\label{anomalousneacklace1}
\nonumber\partial_sZ_{1\,\text{neck}}&=-\frac{3}{[Z_1e^{2\eta s}+Z_2e^s+m^{2\eta}]^2}\bigg\{\partial_sZ_1\big[2\lambda_3(e^{2\eta s}S_6^{(2\eta)}-S_7^{(2\eta)}-2 S_6)\\\nonumber
&+\lambda_2(e^{2\eta s}S_4^{(2\eta)}-S_5^{(2\eta)}-2S_4)\big]+2\eta Z_1e^{2\eta s}(2\lambda_3S_6^{(2\eta)}+\lambda_2S_4^{(2\eta)})\\\nonumber
&+\partial_sZ_2\big[2\lambda_3(e^{2\eta s}S_6^{(2\eta)}-S_8^{(2\eta)}-2S_6^{(2\eta-1)})+\lambda_2(e^{2\eta s}S_4^{(2\eta)}-2S_6^{(2\eta)}\\\nonumber
&-2S_4^{(2\eta-1)})\big]+Z_2e^{s}(2\lambda_3S_6^{(2\eta)}+\lambda_2S_4^{(2\eta)})+2\eta Z_2\big[2\lambda_3\big(e^s\partial S_6^{(2\eta)}-\partial S_8^{(2\eta)}-2\partial S_6^{(2\eta-1)}\big)\\
&+\lambda_2\big(e^s\partial S_4^{(2\eta)}+2\partial S_6^{(2\eta)}-2\partial S_4^{(2\eta-1)}\big)\big]\bigg\}\,,
\end{align} 
\begin{align}\label{anomalousneacklace2}
\nonumber &\partial_sZ_{2\,\text{neck}}=-\frac{3}{[Z_1e^{2\eta s}+Z_2e^s+m^{2\eta}]^2}\bigg\{\partial_sZ_1\big[2\lambda_3(e^{2\eta s}S_6^{(1)}-S_7^{(1)}-S_6^{(1-2\eta)})\\\nonumber
&+2\lambda_3(e^{2\eta s}S_4-S_5-2S_4^{(-2\eta)})+\lambda_2(e^{2\eta s}S_4^{(1)}-S_5^{(1)}-2S_4^{(1-2\eta)})\big]+2\eta Z_1e^{2\eta s}(2\lambda_3(S_6^{(1)}+S_4)+\lambda_2S_4^{(1)})\\\nonumber
&+\partial_sZ_2\big[2\lambda_3(e^{2\eta s}S_6^{(1)}-S_8^{(1)}-2S_6)+2\lambda_3(e^{2\eta s}S_4-S_5-2S_4^{(-1)})+\lambda_2(e^{2\eta s}S_4^{(1)}-2S_6^{(1)}-2S_4)\big]\\\nonumber
&+Z_2e^{s}(2\lambda_3(S_6^{(1)}+S_4)+\lambda_2S_4^{(1)})+2\eta Z_2\big[2\lambda_3\big(e^s\partial S_6^{(1)}-\partial S_8^{(1)}-2\partial S_6\big)+2\lambda_3\big(e^s\partial S_4-2\partial S_6^{(1)}-2\partial S_4\big)\\
&+\lambda_2\big(e^s\partial S_4^{(1)}+2\partial S_6^{(1)}-2\partial S_4\big)\big]\bigg\}\,.
\end{align} 

where the factor $12$ in the first equation accounts for: the factor $2$ arising from the definition of $\mathcal{F}_s$, the factor $2$ coming from the symmetrization of the vertices, and the factor $3$ representing the number of necklace interactions. The factor $3$ on the r.h.s. of the second and third equations is obtained as $12/4$, where the divisor $4$ originates from the l.h.s. of the flow equation (see \eqref{projection}).\\

\noindent
By combining the contributions \eqref{betammelon} and \eqref{betamneacklace} for $\partial_s m^{2\eta}$, along with \eqref{anomalous1melon}, \eqref{anomalous2melon}, \eqref{anomalousneacklace1}, and \eqref{anomalousneacklace2} for $\partial_s Z_1$ and $\partial_s Z_2$, we obtain the following coupled system of equations:

\begingroup
\small
\begin{align}\label{flowm}
\nonumber\partial_sm^{2\eta}=&-\frac{12}{[Z_1e^{2\eta s}+Z_2e^s+m^{2\eta}]^2}\bigg\{\partial_sZ_1\big[\frac{2}{3}\lambda_1(e^{2\eta s}S_1-S_2)+2\lambda_3(e^{2\eta s}S_6-S_7)+\lambda_2(e^{2\eta s}S_4-S_5)\big]\\\nonumber
&+2\eta Z_1e^{2\eta s}(\frac{2}{3}\lambda_1S_1+2\lambda_3S_6+\lambda_2S_4)+\partial_sZ_2\big[\frac{2}{3}\lambda_1(e^sS_1-S_3)+2\lambda_3(e^{ s}S_6-S_8)+\lambda_2(e^{ s}S_4-2S_6)\big]\\
&+Z_2e^{s}(\lambda_1S_1+2\lambda_3S_6+\lambda_2S_4)+2\eta Z_2\big[\frac{2}{3}\lambda_1(e^s\partial S_1-\partial S_3)+2\lambda_3\big(e^s\partial S_6-\partial S_8\big)+\lambda_2\big(e^s\partial S_4-2\partial S_6\big)\big]\bigg\}\,,
\end{align}
\begin{align}
\nonumber\partial_sZ_1=&-\dfrac{3}{[Z_1e^{2\eta s}+Z_2e^s+m^{2\eta}]^2+A(\lambda_1,\lambda_2,\lambda_3)}\bigg\{2\eta Z_1e^{2\eta s}(\frac{2}{3}\lambda_1S_1^{(2\eta)}+2\lambda_3S_6^{(2\eta)}+\lambda_2S_4^{(2\eta)})\\\nonumber
&+\partial_sZ_2\big[\frac{2}{3}\lambda_1(e^sS_1^{(2\eta)}-S_3^{(2\eta)}-S_1^{(2\eta-1)})+2\lambda_3(e^{2\eta s}S_6^{(2\eta)}-S_8^{(2\eta)}-2S_6^{(2\eta-1)})\\\nonumber
&+\lambda_2(e^{2\eta s}S_4^{(2\eta)}-2S_6^{(2\eta)}-S_4^{(2\eta-1)})\big]+Z_2e^{s}(\frac{2}{3}\lambda_1S_1^{(2\eta)}+2\lambda_3S_6^{(2\eta)}+\lambda_2S_4^{(2\eta)})\\\nonumber
&+2\eta Z_2\big[\frac{2}{3}\lambda_1(e^s\partial S_1^{(2\eta)}-\partial S_3^{(2\eta)})+2\lambda_3\big(e^s\partial S_6^{(2\eta)}-\partial S_8^{(2\eta)}-2\partial S_6^{(2\eta-1)}\big)\\
&+\lambda_2\big(e^s\partial S_4^{(2\eta)}+2\partial S_6^{(2\eta)}-2\partial S_4^{(2\eta-1)}\big)\big]\bigg\}\label{flowZ1},
\end{align}

\begin{align}
\nonumber\partial_sZ_2=&-\dfrac{3}{[Z_1e^{2\eta s}+Z_2e^s+m^{2\eta}]^2+B(\lambda_1,\lambda_2,\lambda_3)}\bigg\{2\eta Z_1e^{2\eta s}(\frac{2}{3}\lambda_1S_1^{(1)}+2\lambda_3(S_6^{(1)}+S_4)+\lambda_2S_4^{(1)})\\\nonumber
&+\partial_sZ_1\big[\frac{2}{3}\lambda_1(e^{2\eta s}S_1^{(1)}-S_2^{(1)}-S_1^{(1-2\eta)})+2\lambda_3(e^{2\eta s}S_6^{(1)}-S_7^{(1)}-S_6^{(1-2\eta)})\\\nonumber
&+2\lambda_3(e^{2\eta s}S_4-S_5-2S_4^{(-2\eta)})+\lambda_2(e^{2\eta s}S_4^{(1)}-S_5^{(1)}-2S_4^{(1-2\eta)})\big]\\\nonumber
&+Z_2e^{s}(\frac{2}{3}\lambda_1S_1^{(1)}+2\lambda_3(S_6^{(1)}+S_4)+\lambda_2S_4^{(1)})+2\eta Z_2\big[\frac{2}{3}\lambda_1(e^s\partial S_1^{(1)}-\partial S_3^{(1)})\\
&+2\lambda_3\big(e^s\partial S_6^{(1)}-\partial S_8^{(1)}-2\partial S_6\big)
+\lambda_2\big(e^s\partial S_4^{(1)}+2\partial S_6^{(1)}-2\partial S_4\big)\big]\bigg\},
\end{align}\label{flowZ2}
with:
\begin{align}
\nonumber A(\lambda_1,\lambda_2,\lambda_3):=& 2\lambda_1\big(e^{2\eta s}S_1^{(2\eta)}-S_2^{(2\eta)}-S_1\big)+3\big(2\lambda_3(e^{2\eta s}S_6^{(2\eta)}-S_7^{(2\eta)}-2S_6)\\
&+\lambda_2(e^{2\eta s}S_4^{(2\eta)}-S_5^{(2\eta)}-2S_4)\big)\,,
\end{align}
\begin{align}
\nonumber B(\lambda_1,\lambda_2,\lambda_3):=&2\lambda_1\big(e^{2\eta s}S_1^{(1)}-S_2^{(1)}-S_1\big)+3\big(2\lambda_3(e^{2\eta s}S_6^{(1)}-S_8^{(1)}-2S_6)\\
&+2\lambda_3(e^{2\eta s}S_4-S_5-2S_4^{(1)})+\lambda_2(e^{2\eta s}S_4^{(1)}-2S_6^{(1)}-2S_4)\big)\,.
\end{align}
\endgroup

\bigskip
\noindent
\paragraph{Flow equations for the $4$-valent bubbles.} 
The second equation \eqref{flow1} involves two distinct contributions, which we shall compute separately. The first term, $\text{Tr}\big(\partial_s R_s \mathcal{K}_s^{-1} \mathcal{F}_{s,4} \mathcal{K}_s^{-1}\big)$, accounts for the 6-valent necklace bubbles; a typical leading-order contribution is illustrated in Figure \ref{fig4} below.
\begin{center}
\includegraphics[scale=1.1]{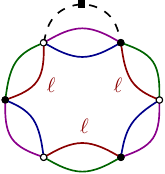} 
\captionof{figure}{Typical LO contribution to the term $\Tr\big(\partial_sR_s\mathcal{K}_s^{-1}\mathcal{F}_{s,4}\mathcal{K}_s^{-1}\big)$}\label{fig4}
\end{center}
These terms only contribute to the flow equation for $\lambda_2$, as their connectivity corresponds to a $4$-valent necklace bubble without derivative coupling. Setting $T = T^{(0)}$, the l.h.s. of the flow equation for the $4$-valent bubbles is given by:
\begin{equation}
\partial_s\Gamma_{s,4}[T^{(0)},\bar{T}^{(0)}]=4\partial_s\lambda_1+3\partial_s\lambda_2 \, ,
\end{equation}

and the contribution to $3\partial_s \lambda_2$ arising from the 6-valent bubbles on the r.h.s. of the flow equation \eqref{flow1} is given by:
\begin{align}
\nonumber\Tr\big(\partial_sR_s\mathcal{K}_s^{-1}\mathcal{F}_{s,4}[T^{(0)},\bar{T}^{(0)}]\mathcal{K}_s^{-1}\big)=&\frac{18\lambda_4}{[Z_1e^{2\eta s}+Z_2e^s+m^{2\eta}]^2}\bigg\{\partial_sZ_1(e^{2\eta s}S_4-S_5)+2\eta Z_1e^{2\eta s}S_4\\
&+\partial_sZ_2(e^{2\eta s}S_4-2S_6)+Z_2e^{s}S_4+2\eta e^{2\eta s} Z_2\big(e^s\partial S_4-2\partial S_6\big)\bigg\}\,.\label{cont1Bis}
\end{align}
A second contribution originates from the term $\text{Tr}\big(\partial_s R_s \mathcal{K}_s^{-1} \mathcal{F}_{s,2} \mathcal{K}_s^{-1} \mathcal{F}_{s,2} \mathcal{K}_s^{-1}\big)$ on the r.h.s. of equation \eqref{flow1}. This term also provides the contributions to the flow of the melonic coupling $\lambda_1$; both contributions can be extracted by projecting onto the subspace generated by $T^{(0)}$. The typical graphs contributing to the flow of $\lambda_2$ are illustrated in Figure \ref{fig5}b (for $i=k$ and $j=l$). Their contributions follow the structure below:
\begin{align}\label{cont2}
\frac{\mathcal{N}_{222}\lambda_2^2\mathcal{S}_{222}+\mathcal{N}_{223}\lambda_2\lambda_3\mathcal{S}_{223}+\mathcal{N}_{233}\lambda_3^2\mathcal{S}_{233}}{[Z_1e^{2\eta s}+Z_2e^s+m^{2\eta}]^3}\,,
\end{align}
where the $\mathcal{N}_{ijk}$ are purely numerical coefficients and $S_{ijk}$ (meaning $S_{ijk}(0)$) are:
\begin{align}
\nonumber\mathcal{S}_{222}(k)=\sum_{\vec{p}}\partial_sR_s(\vec{p}\,)\delta_{p_1k}\delta_{p_2k}=&\partial_sZ_1(e^{2\eta s}S_4-S_5)+2\eta Z_1e^{2\eta s}S_4+\partial_sZ_2(e^{s}S_4-2S_6)\\
&+Z_2e^{s}S_4+2\eta Z_2 e^{2\eta s}\big(e^s\partial S_4-2\partial S_6\big)\,,
\end{align}
\begin{align}
\nonumber\mathcal{S}_{223}(k)=2\sum_{\vec{p}}\partial_sR_s(\vec{p}\,)|p_3|\delta_{p_1k}\delta_{p_2k}=&2\big[\partial_sZ_1(e^{2\eta s}S_6-S_7)+2\eta Z_1e^{2\eta s}S_6+\partial_sZ_2 (e^{s}S_6-S_8)\\
&+Z_2e^{s}S_6+2\eta e^{2\eta s} Z_2\big(e^s\partial S_6-\partial S_8\big)\big]\,,
\end{align}
\begin{align}
\nonumber\mathcal{S}_{233}(k)&=2\sum_{\vec{p}}\partial_sR_s(\vec{p}\,)(|p_3|^2+|p_3||p_4|)\delta_{p_1k}\delta_{p_2k}\\\nonumber
&=2\big[\partial_sZ_1(e^{2\eta s}S_9-S_{10})+2\eta Z_1e^{2\eta s}S_9+\partial_sZ_2 (e^{s}S_9-S_{11})\\
&+Z_2e^{s}S_9+2\eta e^{2\eta s} Z_2\big(e^s\partial S_9-\partial S_{11}\big)\big]\,,
\end{align}
with:
\begin{align}\label{keysums3}
S_9(k)&=\sum_{\sum_{i=3}^4|p_i|^{2\eta}\leq e^{2\eta s}-2k^{2\eta}}(|p_3|^2+|p_3||p_4|),\\
S_{10}(k)&=\sum_{\sum_{i=3}^4|p_i|^{2\eta}\leq e^{2\eta s}-2k^{2\eta}}\sum_{i=3}^4|p_i|^{2\eta}(|p_3|^2+|p_3||p_4|),\\
S_{11}(k)&=\sum_{\sum_{i=3}^4|p_i|^{2\eta}\leq e^{2\eta s}-2k^{2\eta}}\sum_{i=3}^4|p_i|(|p_3|^2+|p_3||p_4|),\\
\partial S_9(k)&=\sum_{\sum_{i=3}^4|p_i|^{2\eta}= e^{2\eta s}-2k^{2\eta}}(|p_3|^2+|p_3||p_4|),\\
\partial S_{11}(k)&=\sum_{\sum_{i=3}^4|p_i|^{2\eta}= e^{2\eta s}-2k^{2\eta}}\sum_{i=3}^4|p_i|(|p_3|^2+|p_3||p_4|)\,.
\end{align}
Then, by counting the number of contractions for each contribution, we find the combinatorial factors to be $\mathcal{N}_{222}=24$, $\mathcal{N}_{223}=48$, and $\mathcal{N}_{233}=24$. Finally, by combining the contributions from \eqref{cont1Bis} and \eqref{cont2}, we obtain:
\begin{align}
\partial_s \lambda_2=&-\frac{6\lambda_4\mathcal{S}_{222}}{[Z_1e^{2\eta s}+Z_2e^s+m^{2\eta}]^2}+\frac{8\lambda_2^2\mathcal{S}_{222}+16\lambda_2\lambda_3\mathcal{S}_{223}+8\lambda_3^2\mathcal{S}_{233}}{[Z_1e^{2\eta s}+Z_2e^s+m^{2\eta}]^3}\,.
\end{align}\label{eqflowlambda2}
\begin{center}
\includegraphics[scale=1]{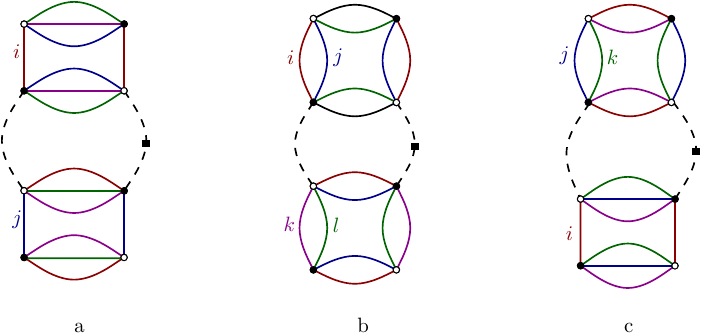} 
\captionof{figure}{Graphical representation of the LO contributions to the flow equations for $\lambda_1$, $\lambda_2$, and $\lambda_3$, arising from the term $\text{Tr}\big(\partial_s R_s \mathcal{K}_s^{-1} \mathcal{F}_{s,2} \mathcal{K}_s^{-1} \mathcal{F}_{s,2} \mathcal{K}_s^{-1}\big)$.}\label{fig5}
\end{center}
In a similar fashion, the r.h.s. of the flow equation for the melonic coupling $\lambda_1$ exhibits the following structure:
\begin{equation}
\frac{\mathcal{N}_{111}\lambda_1^2\mathcal{S}_{111}+\mathcal{N}_{113}\lambda_1\lambda_3\mathcal{S}_{113}+\mathcal{N}_{112}\lambda_1\lambda_2\mathcal{S}_{112}}{[Z_1e^{2\eta s}+Z_2e^s+m^{2\eta}]^3}\,,
\end{equation}
and with $\mathcal{N}_{111}=16$, $\mathcal{N}_{113}=192$, $\mathcal{N}_{112}=96$, we find:
\begin{equation}\label{eqflowlambda1}
\partial_s\lambda_1=\frac{4\lambda_1^2\mathcal{S}_{111}+48\lambda_1\lambda_3\mathcal{S}_{113}+24\lambda_1\lambda_2\mathcal{S}_{112}}{[Z_1e^{2\eta s}+Z_2e^s+m^{2\eta}]^3}\,,
\end{equation}
the numerical coefficients $\mathcal{S}_{113}$ and $\mathcal{S}_{112}$ equaling $\mathcal{S}_{222}$ and $\mathcal{S}_{223}$  respectively, and $\mathcal{S}_{111}$ being:

\begin{align}
\nonumber\mathcal{S}_{111}=&\partial_sZ_1(e^{2\eta s}S_1-S_2)+2\eta Z_1e^{2\eta s}S_1+\partial_sZ_2(e^{s}S_1-2S_3)\\
&+Z_2e^{s}S_1+2\eta Z_2 e^{2\eta s}\big(e^s\partial S_1-2\partial S_3\big)\,.
\end{align}
The flow equation for $\partial_s \lambda_3$ receives contributions exclusively from the term $$\text{Tr}\big(\partial_s R_s \mathcal{K}_s^{-1} \mathcal{F}_{s,2} \mathcal{K}_s^{-1} \mathcal{F}_{s,2} \mathcal{K}_s^{-1}\big)\,.$$ To extract this flow, we employ the test tensor $T^{(k)}$ with $k \neq 0$. In this case, the term involving $\partial_s \lambda_3$ on the l.h.s. is given by:
\begin{align}
24|k|\partial_s\lambda_3.
\end{align}
Hence, the flow equation is obtained by isolating the terms on the r.h.s. that are proportional to $|k|$ and exhibit the same connectivity as the 4-valent necklace bubbles. The relevant contribution to the r.h.s. takes the following form:
\begin{equation}
\frac{\mathcal{N}_{333}\lambda_3^2\mathcal{S}_{223}+\mathcal{N}_{332}\lambda_3\lambda_2\mathcal{S}_{222}}{[Z_1e^{2\eta s}+Z_2e^s+m^{2\eta}]^3}+3\times\frac{8\lambda_2^2\mathcal{S}_{222}^{(1)}+16\lambda_2\lambda_3\mathcal{S}_{223}^{(1)}+8\lambda_3^2\mathcal{S}_{233}^{(1)}}{[Z_1e^{2\eta s}+Z_2e^s+m^{2\eta}]^3}\,,
\end{equation}
where the second term is identical to the contribution for $\partial_s \lambda_2$ but expanded to first order in $|k|$, while the first contribution arises from configurations involving a derivative coupling within an external face. A direct inspection shows that $\mathcal{N}_{333} = \mathcal{N}_{332} = 3 \times 48$. Consequently, we obtain:
\begin{equation}\label{eqflowlambda3}
\partial_s\lambda_3=\frac{6\lambda_3^2\mathcal{S}_{223}+6\lambda_3\lambda_2\mathcal{S}_{222}}{[Z_1e^{2\eta s}+Z_2e^s+m^{2\eta}]^3}+\frac{\lambda_2^2\mathcal{S}_{222}^{(1)}+2\lambda_2\lambda_3\mathcal{S}_{223}^{(1)}+\lambda_3^2\mathcal{S}_{233}^{(1)}}{[Z_1e^{2\eta s}+Z_2e^s+m^{2\eta}]^3}.
\end{equation}

\paragraph{Flow equation for $6$-valent bubbles.}
The first contribution arises from the term $$\text{Tr}\big(\partial_s R_s \mathcal{K}_s^{-1} \mathcal{F}_{s,6} \mathcal{K}_s^{-1}\big)\,,$$ on the r.h.s. of the flow equation \eqref{flow1}, whose leading-order (LO) contributions are illustrated in Figure \ref{fig6}. Up to a numerical factor, this contribution is the same as for Figure \ref{fig3}:
\begin{align}
\nonumber\Tr\big(\partial_sR_s\mathcal{K}_s^{-1}\mathcal{F}_{s,6}[T^{(0)},\bar{T}^{(0)}]\mathcal{K}_s^{-1}\big)=&\frac{24\lambda_5\mathcal{S}_{222}}{[Z_1e^{2\eta s}+Z_2e^s+m^{2\eta}]^2}.
\end{align}\label{cont1}

\begin{figure}[htbp]
  \centering
\includegraphics[scale=1]{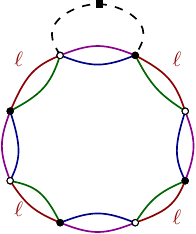} 
\caption{Typical LO contribution to $\Tr\big(\partial_sR_s\mathcal{K}_s^{-1}\mathcal{F}_{s,6}[T^{(0)},\bar{T}^{(0)}]\mathcal{K}_s^{-1}\big)$.}
\label{fig6}
\end{figure}

A second contribution arises from the term $\text{Tr}\big(\partial_s R_s \mathcal{K}_s^{-1} \mathcal{F}_{s,2} \mathcal{K}_s^{-1} \mathcal{F}_{s,4} \mathcal{K}_s^{-1}\big)$, whose leading-order (LO) contribution is illustrated in Figure \ref{fig7} below. In terms of its analytical expression, this corresponds to:
\begin{equation}\label{termref1}
\Tr\big(\partial_sR_s\mathcal{K}_s^{-1}\mathcal{F}_{s,2}\mathcal{K}_s^{-1}\mathcal{F}_{s,4}\mathcal{K}_s^{-1}\big)=\frac{\mathcal{N}_{443}\lambda_4\lambda_3\mathcal{S}_{223}+\mathcal{N}_{442}\lambda_4\lambda_2\mathcal{S}_{222}}{[Z_1e^{2\eta s}+Z_2e^s+m^{2\eta}]^3}+\text{NLO}\,.
\end{equation}
\begin{figure}[htbp]
  \centering\includegraphics[scale=1]{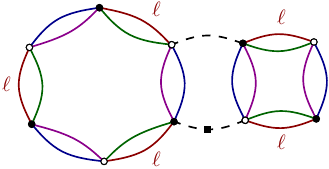} 
\caption{Typical LO contribution of $\Tr\big(\partial_sR_s\mathcal{K}_s^{-1}\mathcal{F}_{s,2}\mathcal{K}_s^{-1}\mathcal{F}_{s,4}\mathcal{K}_s^{-1}\big)$.}
\label{fig7}
\end{figure}
The final contribution arises from the term $\text{Tr}\big(\partial_s R_s \mathcal{K}_s^{-1} \mathcal{F}_{s,2} \mathcal{K}_s^{-1} \mathcal{F}_{s,2} \mathcal{K}_s^{-1} \mathcal{F}_{s,2} \mathcal{K}_s^{-1}\big)$, leading to the diagrams illustrated in Figure \ref{fig8}. Its analytical structure is given by:
\begin{align}
&\nonumber \Tr\big(\partial_sR_s\mathcal{K}_s^{-1}\mathcal{F}_{s,2}\mathcal{K}_s^{-1}\mathcal{F}_{s,2}\mathcal{K}_s^{-1}\mathcal{F}_{s,2}\mathcal{K}_s^{-1}\big)\\
&\qquad =\frac{\mathcal{N}_{4333}\lambda_3^3\mathcal{S}_{333}+\mathcal{N}_{4222}\lambda_2^3\mathcal{S}_{222}+\mathcal{N}_{4233}\lambda_3^2\lambda_2\mathcal{S}_{233}+\mathcal{N}_{4223}\lambda_2^2\lambda_3\mathcal{S}_{223}}{[Z_1e^{2\eta s}+Z_2e^s+m^{2\eta}]^4}\,,
\end{align}
with:
\begin{align}
\nonumber\mathcal{S}_{333}(k)&=2\sum_{\vec{p}}\partial_sR_s(\vec{p}\,)(|p_3|^3+2|p_3|^2|p_4|)\delta_{p_1k}\delta_{p_2k}\\\nonumber
&=2\big[\partial_sZ_1(e^{2\eta s}S_{12}-S_{13})+2\eta Z_1e^{2\eta s}S_{12}+\partial_sZ_2(e^{s}S_{12}-S_{14})\\
&+Z_2e^{s}S_{12}+2\eta Z_2\big(e^s\partial S_{12}-\partial S_{14}\big)\big]\,,
\end{align}
and:
\begin{align}\label{keysums3}
S_{12}(k)&=\sum_{\sum_{i=3}^4|p_i|^{2\eta}\leq e^{2\eta s}-2k^{2\eta}}(|p_3|^3+2|p_3|^2|p_4|)\,,\\
S_{13}(k)&=\sum_{\sum_{i=3}^4|p_i|^{2\eta}\leq e^{2\eta s}-2k^{2\eta}}\sum_{i=3}^4|p_i|^{2\eta}(|p_3|^3+2|p_3|^2|p_4|)\,,\\
S_{14}(k)&=\sum_{\sum_{i=3}^4|p_i|^{2\eta}\leq e^{2\eta s}-2k^{2\eta}}\sum_{i=3}^4|p_i|(|p_3|^3+2|p_3|^2|p_4|)\,,\\
\partial S_{12}(k)&=\sum_{\sum_{i=3}^4|p_i|^{2\eta}= e^{2\eta s}-2k^{2\eta}}(|p_3|^3+2|p_3|^2|p_4|)\,,\\
\partial S_{14}(k)&=\sum_{\sum_{i=3}^4|p_i|^{2\eta}= e^{2\eta s}-2k^{2\eta}}\sum_{i=3}^4|p_i|(|p_3|^3+2|p_3|^2|p_4|)\,.
\end{align}
\begin{figure}[htbp]
  \centering
\includegraphics[scale=1]{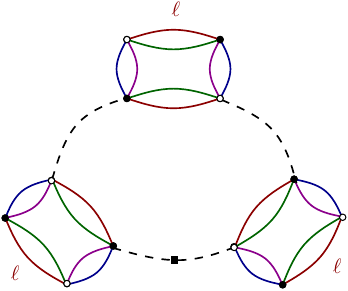} 
\caption{Typical LO contribution arising from the term $\text{Tr}\big(\partial_s R_s \mathcal{K}_s^{-1} \mathcal{F}_{s,2} \mathcal{K}_s^{-1} \mathcal{F}_{s,2} \mathcal{K}_s^{-1} \mathcal{F}_{s,2} \mathcal{K}_s^{-1}\big)$.}
\label{fig8}
\end{figure}
With the combinatorial factors $\mathcal{N}_{443} = 12 \times 3 = \mathcal{N}_{442}$ and $3\mathcal{N}_{4333} = \mathcal{N}_{4332} = \mathcal{N}_{4233} = 3\mathcal{N}_{4222} = 48 \times 3$, we obtain:
\begin{align}\label{flowlambda4}
\nonumber\partial_s\lambda_4=&-\frac{8\lambda_5\mathcal{S}_{222}}{[Z_1e^{2\eta s}+Z_2e^s+m^{2\eta}]^2}+24\frac{\lambda_4\lambda_3\mathcal{S}_{223}+\lambda_4\lambda_2\mathcal{S}_{222}}{[Z_1e^{2\eta s}+Z_2e^s+m^{2\eta}]^3}\\
&\qquad\qquad-16\frac{\lambda_3^3\mathcal{S}_{333}+\lambda_2^3\mathcal{S}_{222}+3\lambda_3^2\lambda_2\mathcal{S}_{233}+3\lambda_2^2\lambda_3\mathcal{S}_{223}}{[Z_1e^{2\eta s}+Z_2e^s+m^{2\eta}]^4}\,.
\end{align}

\paragraph{Flow equation for $8$-valent bubbles.} 
The r.h.s. of the flow equation for $\lambda_5$ involves four traces, as shown in the fourth equation of \eqref{flow1}. Up to a numerical factor reflecting the size of the interaction bubbles involved in $\mathcal{F}_{s,6}$, the term $\text{Tr}\big(\partial_s R_s \mathcal{K}_s^{-1} \mathcal{F}_{s,2} \mathcal{K}_s^{-1} \mathcal{F}_{s,6} \mathcal{K}_s^{-1}\big)$ whose leading-order (LO) contribution involves loops which are identical to \eqref{termref1} (up to the replacement $\bar{\lambda}_4 \to \bar{\lambda}_5$), namely:
\begin{align}\label{1}
\Tr\big(\partial_sR_s\mathcal{K}_s^{-1}\mathcal{F}_{s,2}\mathcal{K}_s^{-1}\mathcal{F}_{s,6}\mathcal{K}_s^{-1}\big)=3\times 16\frac{\lambda_5\lambda_3\mathcal{S}_{223}+\lambda_5\lambda_2\mathcal{S}_{222}}{[Z_1e^{2\eta s}+Z_2e^s+m^{2\eta}]^3}+\text{NLO}.
\end{align}
The second trace $\Tr\big(\partial_sR_s\mathcal{K}_s^{-1}\mathcal{F}_{s,4}\mathcal{K}_s^{-1}\mathcal{F}_{s,4}\mathcal{K}_s^{-1}\big)$ involves LO graphs pictured on \ref{fig9}b, and gives the contribution:
\begin{equation}\label{2}
\Tr\big(\partial_sR_s\mathcal{K}_s^{-1}\mathcal{F}_{s,4}\mathcal{K}_s^{-1}\mathcal{F}_{s,4}\mathcal{K}_s^{-1}\big)=3\times \frac{12\lambda_4^2\mathcal{S}_{222}}{[Z_1e^{2\eta s}+Z_2e^s+m^{2\eta}]^2}+\text{NLO}.
\end{equation}
\begin{center}
\includegraphics[scale=1.1]{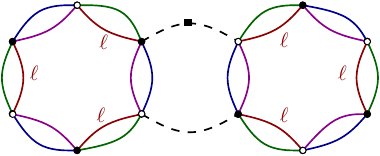} 
\captionof{figure}{Two LO contributions to the flow equation for $\partial_s\lambda_5$}\label{fig9}
\end{center}
The third contribution, $\text{Tr}\big(\partial_s R_s \mathcal{K}_s^{-1} \mathcal{F}_{s,2} \mathcal{K}_s^{-1} \mathcal{F}_{s,2} \mathcal{K}_s^{-1} \mathcal{F}_{s,4} \mathcal{K}_s^{-1}\big)$, yields the leading-order (LO) contribution illustrated in Figure \ref{fig10} below. It shares the same structural topology as the LO graphs shown in Figure \ref{fig7}:
\begin{equation}\label{3}
\Tr\big(\partial_sR_s\mathcal{K}_s^{-1}\mathcal{F}_{s,2}\mathcal{K}_s^{-1}\mathcal{F}_{s,2}\mathcal{K}_s^{-1}\mathcal{F}_{s,4}\mathcal{K}_s^{-1}\big)=3\times 24\lambda_4\frac{\lambda_2^2\mathcal{S}_{222}+\lambda_3^2\mathcal{S}_{233}+2\lambda_3\lambda_2\mathcal{S}_{223}}{[Z_1e^{2\eta s}+Z_2e^s+m^{2\eta}]^3}+\text{NLO}\,.
\end{equation}

\begin{figure}[htbp]
  \centering
\includegraphics[scale=1.2]{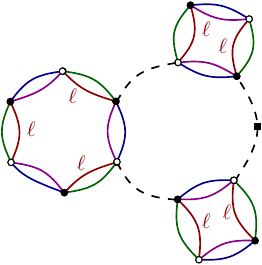} 
\captionof{figure}{Leading order contribution to the flow equation for $\partial_s\lambda_5$ coming from the octic trace: $\Tr\big(\partial_sR_s\mathcal{K}_s^{-1}\mathcal{F}_{s,2}\mathcal{K}_s^{-1}\mathcal{F}_{s,2}\mathcal{K}_s^{-1}\mathcal{F}_{s,4}\mathcal{K}_s^{-1}\big)$.}
\label{fig10}
\end{figure}
The fourth and final contribution arises from the trace $\text{Tr}\big(\partial_s R_s \mathcal{K}_s^{-1} \mathcal{F}_{s,2} \mathcal{K}_s^{-1} \mathcal{F}_{s,2} \mathcal{K}_s^{-1} \mathcal{F}_{s,2} \mathcal{K}_s^{-1} \mathcal{F}_{s,2} \mathcal{K}_s^{-1}\big)$, whose typical LO configurations are illustrated in Figure \ref{fig11}. It shares the same structural topology as the previous contribution $\text{Tr}\big(\partial_s R_s \mathcal{K}_s^{-1} \mathcal{F}_{s,2} \mathcal{K}_s^{-1} \mathcal{F}_{s,2} \mathcal{K}_s^{-1} \mathcal{F}_{s,2} \mathcal{K}_s^{-1}\big)$ computed for $\partial_s \lambda_4$, and we find:
\begin{align}
\nonumber &\Tr\big(\partial_sR_s\mathcal{K}_s^{-1}\mathcal{F}_{s,2}\mathcal{K}_s^{-1}\mathcal{F}_{s,2}\mathcal{K}_s^{-1}\mathcal{F}_{s,2}\mathcal{K}_s^{-1}\mathcal{F}_{s,2}\mathcal{K}_s^{-1}\big)\\
&\qquad \qquad =32\frac{\lambda_3^4\mathcal{S}_{3333}+4\lambda_3^3\lambda_2\mathcal{S}_{333}+12\lambda_3^2\lambda_2^2\mathcal{S}_{233}+4\lambda_2^3\lambda_3\mathcal{S}_{223}+\lambda_2^4\mathcal{S}_{222}}{[Z_1e^{2\eta s}+Z_2e^s+m^{2\eta}]^5}\,,\label{4}
\end{align}
with:
\begin{align}
\nonumber\mathcal{S}_{3333}(k):=2\sum_{\vec{p}}\partial_sR_s(\vec{p}\,)(|p_3|^4+2|p_3|^3|p_4|+6|p_3|^2|p_4|^2)\delta_{p_1k}\delta_{p_2k}.
\end{align}
\begin{center}
\includegraphics[scale=1.3]{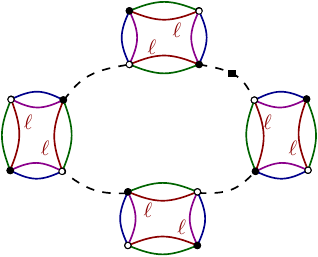} 
\captionof{figure}{LO graph contributing to the flow equation for $\partial_s\lambda_5$ and coming from the trace $\Tr\big(\partial_sR_s\mathcal{K}_s^{-1}\mathcal{F}_{s,2}\mathcal{K}_s^{-1}\mathcal{F}_{s,2}\mathcal{K}_s^{-1}\mathcal{F}_{s,2}\mathcal{K}_s^{-1}\mathcal{F}_{s,2}\mathcal{K}_s^{-1}\big)$.}\label{fig11}
\end{center}
Taking into account all the contributions \ref{1}, \ref{2}, \ref{3}, \ref{4}, we find:
\begin{align}\label{flowlambda5}
\nonumber\partial_s\lambda_5&=32\frac{\lambda_5\lambda_3\mathcal{S}_{223}+\lambda_5\lambda_2\mathcal{S}_{222}}{[Z_1e^{2\eta s}+Z_2e^s+m^{2\eta}]^3}+\frac{12\lambda_4^2\mathcal{S}_{222}}{[Z_1e^{2\eta s}+Z_2e^s+m^{2\eta}]^3}-192\lambda_4\frac{\lambda_2^2\mathcal{S}_{222}+\lambda_3^2\mathcal{S}_{233}+2\lambda_3\lambda_2\mathcal{S}_{223}}{[Z_1e^{2\eta s}+Z_2e^s+m^{2\eta}]^4 }\\
&\qquad \qquad +32\frac{\lambda_3^4\mathcal{S}_{3333}+4\lambda_3^3\lambda_2\mathcal{S}_{333}+12\lambda_3^2\lambda_2^2\mathcal{S}_{233}+4\lambda_2^3\lambda_3\mathcal{S}_{223}+\lambda_2^4\mathcal{S}_{222}}{[Z_1e^{2\eta s}+Z_2e^s+m^{2\eta}]^5}\,.
\end{align}

\subsubsection{Continuum limit and dimensionless renormalized couplings.}

In the large cut-off limit ($e^s \to \infty$) considered in this section, all the sums appearing in the flow equations \eqref{flowm}–\eqref{flowlambda5} can be evaluated using an integral approximation, a standard technique in condensed matter physics. The characteristic melonic sums namely, those arising from graphs with contractions between melonic bubbles in a melonic configuration take the following form:
\begin{equation}
S_s(\alpha, \beta,|k|):=\sum_{\vec{p}\in\mathbb{Z}^3}\Theta \left(e^{2\eta s}-|k|^{2\eta}-\sum_{i=1}^3|p_i|^{2\eta}\right)|p_1|^{\alpha}|p_2|^{\beta}\,.
\end{equation}
And introducing the variables $x_i=|p_i|/e^s$ and $y=|k|/e^s$, this sum can be approached by an integral $S_s(\alpha, \beta,|k|)\approx \mathcal{I}_s(\alpha, \beta,|k|)$, such that:

\begin{align}
\nonumber\mathcal{I}_s(\alpha, \beta,|k|)&:=2^3e^{(3+\alpha+\beta)s}\int_0^{(1-y^{2\eta})^{1/2\eta}} \!\! dx_1 \, x_1^\alpha \int_0^{(1-y^{2\eta}-x_1^{2\eta})^{1/2\eta}} \!\! dx_2 \, x_2^\beta \int_0^{(1-y^{2\eta}-x_1^{2\eta}-x_2^{2\eta})^{1/2\eta}} \!\! dx_3\,.
\end{align}
These integrals can be evaluated by applying two successive changes of variables: $x_2 \to x_2 / (1 - y^{2\eta} - x_1^{2\eta})^{1/2\eta}$ followed by $x_1 \to x_1 / (1 - y^{2\eta})^{1/2\eta}$:
\begin{equation}
\mathcal{I}_s(\alpha, \beta,|k|)= \frac{8 e^{(3+\alpha+\beta)s}}{(2\eta)^2} (1-y^{2\eta})^{\frac{3+\alpha+\beta}{2\eta}} B\left(\frac{\alpha+1}{2\eta}, \frac{\beta+2}{2\eta} + 1\right) B\left(\frac{\beta+1}{2\eta}, \frac{1}{2\eta} + 1\right)\,,
\end{equation}
where $B(a,b)$ denotes the Euler beta function, which is defined as $B(a,b)= \frac{\Gamma(a)\Gamma(b)}{\Gamma(a + b)}$.
In the same way, the boundary sums $\partial S_s(\alpha, \beta,|k|):=\frac{1}{2\eta e^{2\eta s}}\frac{d}{ds}S_s(\alpha, \beta,|k|)$  have integral approximation:
\begin{equation}
\partial\mathcal{I}_s(\alpha, \beta,|k|)=\frac{3+\alpha+\beta}{2\eta e^{2\eta s}}\bigg[1+\frac{y^{2\eta}}{1-y^{2\eta}}\bigg]\mathcal{I}_s(\alpha, \beta,|k|)\,.
\end{equation}

In particular:

\begin{align}\label{keyintsbisbis}
S_3(0)&\approx\frac{12}{5}e^{4s}\Gamma(4/3)\Gamma(2/3),\\
S_3^{(2\eta)}&\approx-\frac{32}{5}e^{(4-2\eta)s}\Gamma(4/3)\Gamma(2/3),\\
S_3^{(1)}&=0,\\
\partial S_1(k)&\approx8e^{(3-2\eta)s}[\Gamma(5/3)]^3,\\
\partial S_1^{(2\eta)}&\approx-8e^{(3-4\eta)s}[\Gamma(5/3)]^3,
\end{align}
\begin{align}
\partial S_1^{(1)}&=0,\\
\partial S_3(0)&\approx \frac{96}{5}e^{(4-2\eta)s}\Gamma(4/3)\Gamma(2/3),\\
\partial S_3^{(2\eta)}&\approx-32e^{4(1-\eta)s}\Gamma(4/3)\Gamma(2/3),\\
\partial S_3^{(1)}&=0.
\end{align}

In the same way, all the necklace sums, i.e. The sums occurring in the graphs involved necklace bubbles are of the form:
\begin{equation}
\mathcal{SN}_s(\alpha,\beta,|k|)=\sum_{(p_1,p_2)\in\mathbb{Z}^2}\Theta(e^{2\eta s}-2|k|^{2\eta}-|p_1|^{2\eta}-|p_2|^{2\eta})|p_1|^{\alpha}|p_2|^{\beta},
\end{equation}
with integral approximation:
\begin{align}
\mathcal{J}_s(\alpha,\beta)&=4e^{(2+\alpha+\beta)s}\int_0^{(1-2y^{2\eta})^{1/2\eta}}x_1^{\alpha}dx_1\int_0^{(1-2y^{2\eta}-x_1^{2\eta})^{1/2\eta}}x_2^{\beta}dx_2\\
&=\frac{4 e^{(2+\alpha+\beta)s}}{2\eta} (1-2y^{2\eta})^{\frac{2+\alpha+\beta}{2\eta}} B\left(\frac{1+\alpha}{2\eta}, \frac{1+\beta}{2\eta} + 1\right)\,.
\end{align}

In the same way, the sum $\partial\mathcal{SN}_s(\alpha,\beta,|k|):=\frac{1}{2\eta e^{2\eta s}}\frac{d}{ds}\mathcal{SN}_s(\alpha,\beta,|k|)$ has integral approximation:
\begin{align}
\partial\mathcal{J}_s(\alpha,\beta)&:=\frac{1}{2\eta e^{2\eta s}}\frac{d}{ds}\mathcal{J}_s(\alpha,\beta)=\frac{2+\alpha+\beta}{2\eta e^{2\eta s}}\bigg[1+\frac{2y^{2\eta}}{1-2y^{2\eta}}\bigg]\mathcal{J}_s(\alpha,\beta)\,.
\end{align}


The flow equations are significantly simplified by noting that all derivatives of the type $S_n^{(1)}$ vanish. In particular, we obtain:

\begin{align}
S_{222}^{(1)}&:=-2 S_4\partial_s Z_2-4 \eta Z_2 \partial S_4e^{2\eta s}\,,\\
S_{223}^{(1)}&:=-4 \partial_s Z_2 S_6 - 8 \eta e^{2 \eta s} Z_2 \partial S_6\,,\\
S_{233}^{(1)}&:=-4 \partial_s Z_2 S_9-8\eta e^{2\eta s} Z_2 \partial S_9\,.
\end{align}

\section{Numerical investigations in the deep UV}\label{numeric}

The flow equations derived in the previous section can be studied numerically. Following the power-counting analysis provided in Section \ref{sectioncanonical}, the dimensionless renormalized couplings are defined as follows:
\begin{equation}
\lambda_1=:Z_1^2 \bar{\lambda}_1\,,\quad \lambda_2=:Z_1^2 e^s \bar{\lambda}_2 \,,\quad \lambda_3=:Z_1^2 \bar{\lambda}_3 \,, \lambda_4=:Z_1^3 e^{s/2} \bar{\lambda}_4 \,, \lambda_5=:Z_1^4 \bar{\lambda}_5 \,,
\end{equation}

and:

\begin{equation}
m^{2\eta}=: Z_1 e^{3 s/2} {\mu} \,,\qquad Z_2=: Z_1 e^{s/2} z_2\,.
\end{equation}
We moreover define the anomalous dimension $\gamma$:

\begin{equation}
\gamma:= \frac{1}{Z_1}\frac{d}{ds} Z_1\,.
\end{equation}

\subsection{Melonic sector}

A simple limiting case is that of the pure melonic sector, for which $\bar{\lambda}_2=\bar{\lambda}_3=\bar{\lambda}_4=\bar{\lambda}_5=z_2=0$. In this case, the equations simplify significantly:

\begin{equation}
\gamma:=-\frac{77760 \bar{\lambda}_1 \Gamma \left(\frac{2}{3}\right)^3 \Gamma \left(\frac{7}{3}\right) \Gamma \left(\frac{10}{3}\right)}{14336 \pi ^2 \bar{\lambda}_1 \Gamma \left(\frac{8}{3}\right)-10935 (\mu +1)^2 \Gamma \left(\frac{7}{3}\right) \Gamma \left(\frac{10}{3}\right)}\,,
\end{equation}

\begin{equation}
\frac{d \mu}{d s}= -\left(\frac{3}{2}+\gamma\right)\mu+\frac{286720 \pi ^2 \bar{\lambda}_1 \left((\mu +1)^2 \Gamma \left(\frac{2}{3}\right)-4 \bar{\lambda}_1 \Gamma \left(\frac{5}{3}\right)^4\right)}{32805 (\mu +1)^4 \Gamma \left(\frac{7}{3}\right) \Gamma \left(\frac{10}{3}\right)-43008 \pi ^2 \bar{\lambda}_1 (\mu +1)^2 \Gamma \left(\frac{8}{3}\right)}\,,
\end{equation}

\begin{equation}
\frac{d \bar{\lambda}_1}{d s}= - 2\gamma \bar{\lambda}_1+ \frac{143360 \pi ^2 \bar{\lambda}_1^2 \left((\mu +1)^2 \Gamma \left(\frac{2}{3}\right)-4 \bar{\lambda}_1 \Gamma \left(\frac{5}{3}\right)^4\right)}{32805 (\mu +1)^5 \Gamma \left(\frac{7}{3}\right) \Gamma \left(\frac{10}{3}\right)-43008 \pi ^2 \bar{\lambda}_1(\mu +1)^3 \Gamma \left(\frac{8}{3}\right)}\,.
\end{equation}

There are three fixed points. One of them is immediately disqualified, as its anomalous dimension falls below the $3/2$ limit imposed by the choice of regulator to ensure that boundary conditions are met in the deep UV. The two remaining fixed points are found at the following values:

\begin{equation}
\text{MFP1} := \{\mu \approx -0.53\,, \bar{\lambda}_1 \approx 0.006\}\,,
\end{equation}
\begin{equation}
\text{MFP2} := \{\mu \approx -0.78\,, \bar{\lambda}_1 \approx 0.004\}\,,
\end{equation}
with anomalous dimensions $\gamma_1 \approx 0.58$ and $\gamma_2 \approx -0.81$. However, the second fixed point is associated with singular critical exponents (of the order of $10^{15}$) and, as such, appears to be pathological. The first fixed point is characterized by the following critical exponents:
\begin{equation}
\{\theta_1,\theta_2 \} \approx \{2.4\,,-0.84\}\,.
\end{equation}
This is clearly a Wilson–Fisher-type fixed point, characteristic of a second-order phase transition. Such fixed points are commonplace in tensor theories and tend to appear in most models. However, this fixed point possesses an intrinsic non-perturbative nature (as it does not appear at the one-loop level) and violates the Ward identities that arise from the breaking of the global unitary symmetry by the propagator. In references \cite{Lahoche_2019bb, Lahoche_2021c}, it is shown that the only way to satisfy the Ward identities is for the anomalous dimension to vanish at the fixed point, leaving $\lambda_1=0$ (the Gaussian fixed point) as the only solution. A general method developed in \cite{Lahoche_2020b} demonstrates that incorporating the constraints from the Ward identities within the melonic sector allows for the closure of the flow equation hierarchy, thereby proving the absence of this fixed point, which appears to be merely an artifact of the vertex expansion. Consequently, this fixed point will be ignored in the remainder of this study.
\subsection{Necklace sector}

Two other fixed points are discovered at the values:
\begin{equation}
\text{NFP1}=\{\mu \approx 0.61\,, z_2\approx 2.32\,, \bar{\lambda}_1 = 0\,,\bar{\lambda}_2\approx -0.05\,,\bar{\lambda}_3\approx 0.0008\,,\bar{\lambda}_4\approx -0.0024\,,\bar{\lambda}_5\approx-0.00038\}\,,
\end{equation}
and:
\begin{equation}
\text{NFP2}=\{\mu \approx -0.53\,, z_2\approx 1.10\,, \bar{\lambda}_1 = 0\,,\bar{\lambda}_2\approx 0.004\,,\bar{\lambda}_3\approx 0.10\,,\bar{\lambda}_4\approx -0.00065\,,\bar{\lambda}_5\approx0.00012\}\,,
\end{equation}
with anomalous dimensions $\gamma_{NFP1} \approx -0.60$ and $\gamma_{NFP2} \approx -0.24$ ; and critical exponents (real parts) given by:
\begin{equation}
\{ \Theta \}_1 = \{-2.35\,,-1.22\,,1.20\,,-1.03\,,-1.03\,,-0.93\,,0.005\,\}\,,
\end{equation}
and:
\begin{equation}
\{ \Theta \}_2 = \{-12.42\,,0.06\,,0.06\,,1.53\,,1.53\,,-1.69\,,-0.81\}\,.
\end{equation}

Many critical exponents share the same real parts and are indeed complex. The same holds for the fourth and fifth exponents in the set ${ \Theta }_1$. The first fixed point possesses two unstable directions, while the second one has four. These are therefore multicritical points, which is unsurprising given the number of relevant parameters in the model. Furthermore, it is easily seen that the $\bar{\lambda}_1$ axis corresponds to an irrelevant direction. Notably, the second-to-last exponent of ${ \Theta }_1$ and the last exponent of ${ \Theta }_2$ correspond to the eigenvector:
\begin{equation}
\textbf{v}:=\{0\,, 0\,,1\,, 0\,,0\,,0\,,0\}\,.
\end{equation}

Hence, some eigendirections escape slowly from the necklace sector. For instance, the last critical exponent of ${ \Theta }_1$, which is a relevant direction, is associated with the eigenvector:
\begin{equation}
\textbf{v}{1,7}:={0.002,, -0.013,,-0.017,, -0.19,,0.04,,-0.49,,-0.85}.
\end{equation}
Note that for both eigenvectors, the components are organized as follows: ${\mu, z_2, \bar{\lambda}_1,\bar{\lambda}_2, \bar{\lambda}_3,\bar{\lambda}_4,\bar{\lambda}_5 }$.

The reliability of the fixed points is difficult to evaluate within such a complex theory space. It should be noted that the sign of the renormalizable coupling $\bar{\lambda}_5$ for the first fixed point leads us to suspect that it is merely a numerical artifact. One way to assess the quality of these fixed points is to examine their stability across lower-order truncations, as summarized in Tables \ref{table1} and \ref{table2}. As shown, the point NFP1 does not appear in a quartic truncation; although such a truncation neglects renormalizable couplings which are relevant in the IR the absence of this fixed point at lower orders provides further evidence of its artificial nature. In the following, we will therefore only retain NFP2, which has the characteristic of having a nearly marginal direction (with critical exponent $0.005$).
\begin{table}[ht]
\centering
\begin{tabular}{|c||c|c|c|}
\hline
 Truncation Order & $4$ & $6$ & $8$  \\ \hline\hline
$\mu$ & $-0.49$ & $-0.58$ & $-0.53$  \\ \hline
$z_2$ & $0.97$ & $1.07$ & $1.10$  \\ \hline
$\bar{\lambda}_2$ & $0.008$ & $0.003$ & $0.004$  \\ \hline
$\bar{\lambda}_3$ & $0.050$ & $0.10$ & $0.101$ \\ \hline
$\bar{\lambda}_4$ & -- & $0.0003$ & $-0.0007$  \\ \hline
$\bar{\lambda}_5$ & -- & -- & $0.0001$ \\ \hline
$\theta_1$ & $-22.52$ & $-37.99$ & $-12.42$ \\ \hline 
$\theta_5$ & $-0.78$ & $-0.77$ & $-1.02$ \\ \hline 
$\gamma$ & $-0.26$ & $-0.28$ & $-0.24$ \\ \hline
\end{tabular}
\caption{Properties of $\text{NFP2}$ for different truncations, up to order $8$.}\label{table1}
\end{table}

\begin{table}[ht]
\centering
\begin{tabular}{|c||c|c|c|}
\hline
 Truncation Order & $4$ & $6$ & $8$  \\ \hline\hline
$\mu$ & -- & $0.57$ & $0.61$  \\ \hline
$z_2$ & -- & $2.85$ & $2.32$  \\ \hline
$\bar{\lambda}_2$ & -- & $-0.05$ & $-0.05$  \\ \hline
$\bar{\lambda}_3$ & -- & $0.00075$ & $0.0008$ \\ \hline
$\bar{\lambda}_4$ & -- & $-0.0022$ & $-0.0024$  \\ \hline
$\bar{\lambda}_5$ & -- & -- & $-0.00038$ \\ \hline
$\theta_1$ & -- & $-1.18$ & $-2.35$ \\ \hline 
$\theta_3$ & -- & $1.16$ & $1.20$ \\ \hline 
$\gamma$ & -- & $-0.58$ & $-0.60$ \\ \hline
\end{tabular}
\caption{Properties of $\text{NFP1}$ for different truncations, up to order $8$.}\label{table2}
\end{table}

Table \ref{table1} summarizes how the fixed point emerges and stabilizes as the truncation order is increased. Interestingly, the anomalous dimension decreases a phenomenon generally associated with robust convergence of the vertex expansion \cite{Balog_2019}. In particular, the table illustrates the crucial role played by $\lambda_4$ and $\lambda_5$ in stabilizing the critical exponents (which are very large in the absence of these couplings, making the RG flow unstable and the fixed point difficult to track). Furthermore, the melonic fixed point discussed in the previous section is already pathological at this level, as its anomalous dimension increases with the truncation order (see \cite{Carrozza_2017a}). As previously discussed, within the melonic sector, Ward identities require the anomalous dimension to vanish at a fixed point; while the results would likely differ for necklaces a topic we will address in a forthcoming paper the fact that the anomalous dimension remains stable is encouraging and suggests that the necklace fixed point may be consistent with the constraints imposed by Ward identities in the infinite truncation limit. 

The behavior of the flow in the vicinity of the fixed point NFP2 is shown in Figure \ref{figflownec}. Figure \ref{landscape} shows the behavior of the anomalous dimension around the fixed point (left), as well as the norm of the 'velocity' vector $\vec{\beta}:={\beta_\mu,\beta_{z_2},\beta_1,\beta_2,\beta_3,\beta_4,\beta_5}$ (right). The top figures show that the anomalous dimension becomes unstable in the $\lambda_1 > 0$ region, such that a fixed point located in this zone of high variation would be considered 'pathological' or an artifact of the vertex expansion. Indeed, this region marks a limit of the approximation considered here. Moreover, in the negative region, the anomalous dimension rapidly reaches the limit $\gamma = -3/2$, below which the boundary condition $R_{s\to \infty}=\infty$ is violated.

The figure on the bottom provides information regarding the stability of the fixed point: no other fixed point is expected in the neighborhood of the NFP2. Furthermore, it should be noted that the landscape surrounding the fixed points, and NFP2 in particular, is extremely flat. Consequently, the precise localization of the fixed point is numerically challenging, and the zeros are, at best, of the order of $10^{-5}$. \\

\begin{figure}[htbp]
\begin{center}
\includegraphics[scale=0.45]{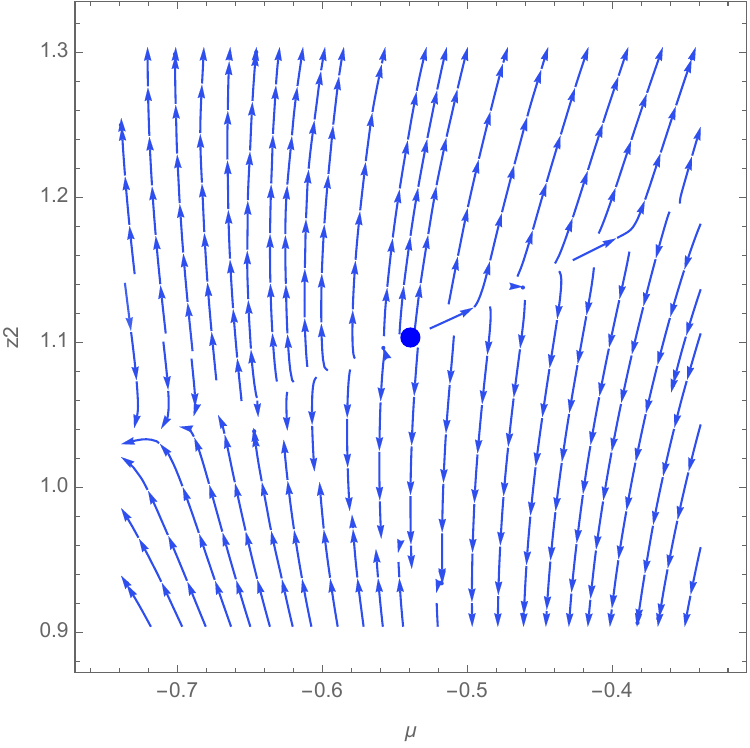}\quad \includegraphics[scale=0.45]{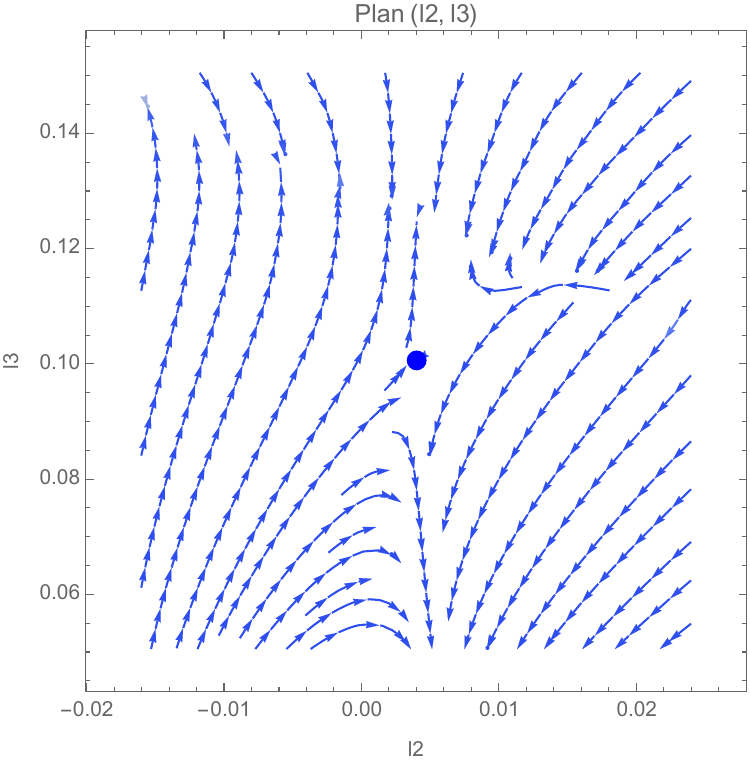}\\
\includegraphics[scale=0.45]{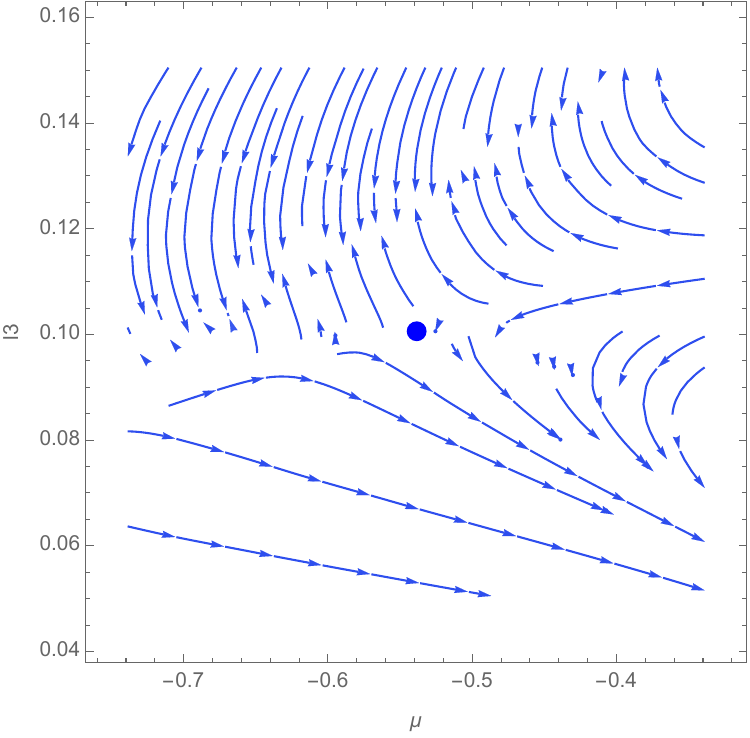}\quad \includegraphics[scale=0.45]{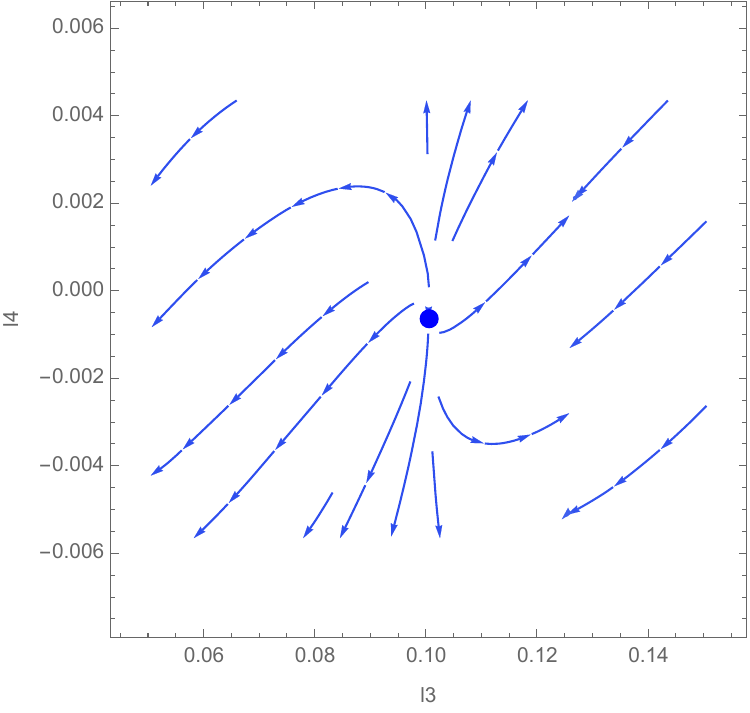}
\end{center}
\caption{Behavior of the RG flow in the vicinity of the non-Gaussian fixed point \text{NFP2}.}\label{figflownec}
\end{figure}

The existence of this fixed point agrees with the findings of \cite{bonzom2015enhancing}: The fixed point identified within the necklace sector delineates a physical scenario in which the theory breaks free from the branched polymer phase, which is typical of the dominant melonic limit. Due to their combinatorial structure, necklaces are akin to the planar graphs of random matrix models, suggesting a transition toward a discretized two-dimensional gravity phase rather than toward one-dimensional tree-like structures\footnote{Branched polymers have Hausdorff dimension 2 (because of the analogy with random walk), but, as branched polymers, they have topological dimension 1 (if we cut an edge of the polymer, we divide him in two parts. Here, we consider the topological dimension.}. This transition from a geometry of isolated 'bubbles' to a more extended structure is supported by the stabilization of the anomalous dimension $\gamma_N$. To summarize: this Necklace fixed point could represent a new universality class within tensorial theories, providing a topological condensation mechanism able of generating spacetime configurations that are closer to geometric continuity.

\begin{figure}
\begin{center}
\includegraphics[scale=0.5]{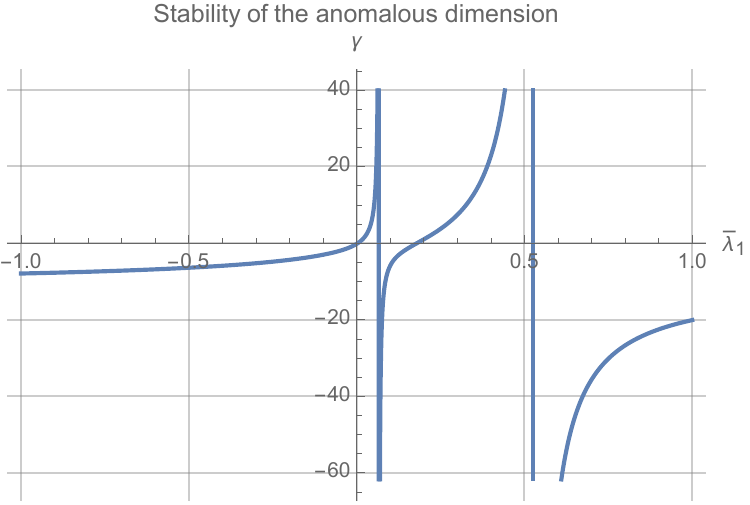}\quad \includegraphics[scale=0.5]{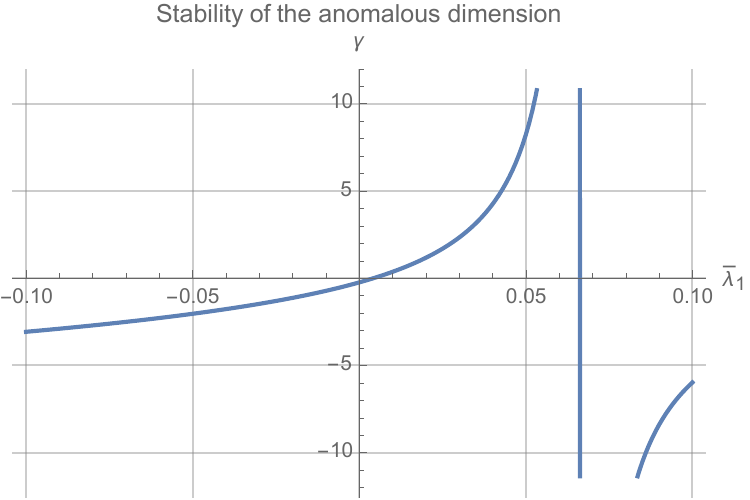} \quad \includegraphics[scale=0.5]{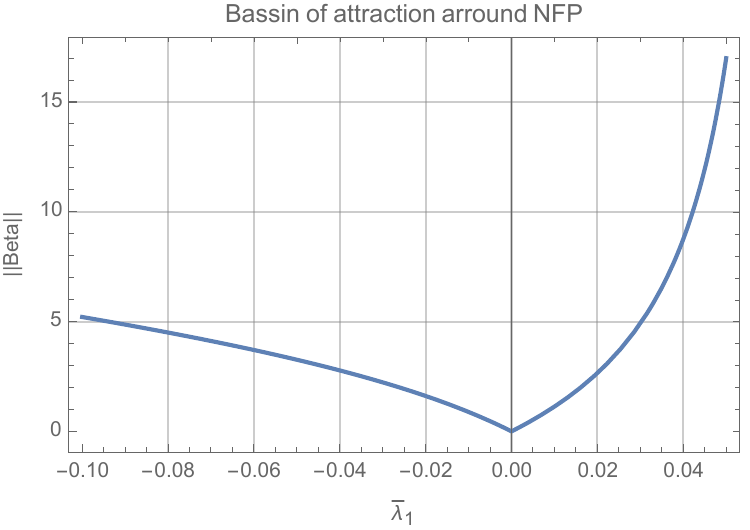}
\end{center}
\caption{On the top: behavior of the anomalous dimension in the vicinity of the non-Gaussian fixed point NFP2. On the bottom: behavior of the norm $\sqrt{\vec{\beta}^2}$.}\label{landscape}
\end{figure}

\section{Conclusion}\label{conclusion}

In this work, we have explored the renormalization group flow structure of a rank-4 TGFT, moving beyond the conventional melonic approximation. By introducing necklace-type interactions enhanced by derivative couplings, we have identified a richer physical landscape characterized by the competition between dominant and subdominant combinatorial structures.

The primary contribution of this study is the discovery of a robust non-Gaussian fixed point (NFP) within the necklace sector. In contrast to the melonic fixed point, which proved pathological under scrutiny with respect to truncation convergence and Ward identities, the necklace fixed point exhibits encouraging signs of numerical stability. We observed that as the truncation order in the vertex expansion increases, the anomalous dimension $\gamma_N$ remains sufficiently stable. This behavior suggests that this fixed point is not merely an artifact of the approximation but may represent a genuine physical phase of the theory.

However, the definitive validity of this fixed point remains contingent upon a rigorous analysis of the modified Ward identities for the non-melonic sector. While the numerical trend of the anomalous dimension is promising, only the integration of symmetry constraints will confirm whether this fixed point survives the breaking of global unitary symmetry induced by the propagator.

Looking forward, this work paves the way for a broader classification of phases in TGFTs. A natural next step will be to extend this analysis to the full flow equations incorporating Ward constraints, in order to verify whether the necklace sector can effectively close the flow equation hierarchy, thereby offering a viable and physically consistent alternative to the standard melonic scenario. Indeed, while the standard melonic limit confines the theory to branched polymer phase structures with a topological dimension of 1 that lack genuine geometric extension, the emergence of the necklace sector allows for the consideration of a richer phase. By drawing closer to the combinatorial structure of matrix models, this fixed point suggests a transition toward 2D planar surfaces, constituting a necessary step toward the generation of a continuous higher-dimensional spacetime.


\pagebreak
\newpage
\appendix


\section{Basics about colored graphs}\label{App0}
This appendix provides definitions and properties of colored graphs. As most of these properties are well-established in the tensor model literature, we shall adopt them here and refer the reader to, for instance, \cite{gurau2017random} for their respective proofs. Following \cite{Carrozza_2014}, we employ the following definitions:
\begin{definition}\label{def1}
\textbf{Contraction operation}. 
Let $\mathcal{G}$ be a Feynman graph and $L_0=\{l_i\} \subset \mathcal{L}(\mathcal{G})$ an ordered subset of dotted (i.e., propagation) lines in $\mathcal{G}$, including tadpole lines. The graph $\mathcal{G}/L_0$ is obtained from $\mathcal{G}$ through the following steps:\\

Let $l_i\in L_0$:\\

\noindent

$\bullet$ Step 1: Delete the line $l_i$ along with its two endpoints (the black and white vertices) and all incident colored edges (other than $l_i$) that join these two vertices.

$\bullet$ Step 2: For each color $c \in \{1, \dots, D\}$, identify the colored edge of color $c$ originally linked to the deleted black vertex with the corresponding edge of color $c$ linked to the white vertex.

$\bullet$ Step 3: Repeat the previous steps for the next line $l_{i+1}$ in the ordered set $L_0$, and continue until all lines in $L_0$ have been processed.
\end{definition}
\begin{definition}
For a connected graph $\mathcal{G}$ with $|L|$ lines and $|V|$ vertices, let $\mathcal{T} \subset \mathcal{G}$ be a spanning tree. We define the tensorial rosette (or simply rosette) as the contracted graph $\mathcal{G}/\mathcal{T}$. This rosette consists of a single vertex and $|L| - |V| + 1$ loops.
\end{definition}
And we have the following lemma: 
\begin{lemma}\label{lemma1}
Consider a connected graph $\mathcal{G}$ possessing $F$ faces. Under the contraction of a spanning tree $\mathcal{T}$, the number of faces $F$ remains invariant. Consequently, the rosette $\mathcal{G}/\mathcal{T}$ preserves the faces number of the original graph.
\end{lemma}
\textit{\textbf{Proof}}: 
Since $\mathcal{T}$ is a spanning tree, it contains no cycles (and thus no self-loops). Every face $f$ of the graph $\mathcal{G}$ must contain at least one edge belonging to the complement of the tree, $\mathcal{G} \setminus \mathcal{T}$. Consequently, the contraction of the edges in $\mathcal{T}$ merely reduces the length of the facial cycles without ever closing or deleting them. The total number of faces $F$ is therefore an invariant of the contraction process $\mathcal{G} \to \mathcal{G}/\mathcal{T}$.
\begin{flushright}
$\square$
\end{flushright}

\begin{definition}\label{coloreddef}
Consider a Feynman graph $\mathcal{G}$. The colored extension $\mathcal{G}_c$ of this graph is the bipartite regular graph constructed as follows:\\ 

\noindent

$\bullet$ The set of vertices is partitioned as $\mathcal{V}(\mathcal{G}_c) = V \cup \bar{V}$, where $V$ and $\bar{V}$ denote the sets of black and white vertices, respectively.

$\bullet$ The set of edges $\mathcal{E}(\mathcal{G}_c)$ consists of all lines (colored and dotted) joining any pair $\{v, \bar{v}\} \in V \times \bar{V}$. By convention, the dotted lines (propagation lines) are assigned the color $0$.

$\bullet$ The set of faces is given by $\mathcal{F}(\mathcal{G}_c) = \mathcal{F}(\mathcal{G}) \cup \mathcal{F}_c^{\neq 0}(\mathcal{G}_c)$. Here, $\mathcal{F}(\mathcal{G})$ represents the set of faces in the original graph $\mathcal{G}$, composed of alternating edges of color $0$ and color $i$ ($i \neq 0$). The set $\mathcal{F}_c^{\neq 0}(\mathcal{G}_c)$ contains the internal faces formed by colors $i$ and $j$ ($i \neq j$ and $i, j \neq 0$).

\end{definition}
\begin{definition}
Consider a colored extension $\mathcal{G}_c$. A $k$-dipole $d_k$ is a set of $k$ edges of distinct colors, necessarily including the color $0$, that link the same two vertices $v$ and $\bar{v}$. Furthermore, these vertices $v$ and $\bar{v}$ must not be connected by any other edges of the remaining $D+1-k$ colors.
\end{definition}
In addition, we recall the following three fundamental definitions regarding the topology and classification of colored graphs:
\begin{definition}
\textbf{(jacket)} 
Consider a colored extension $\mathcal{G}_c$ in dimension $d$. A Jacket $\mathcal{J}$ is a 2-subcomplex of $\mathcal{G}_c$, associated with a cyclic permutation $\tau$ of the color set $\{0, 1, \dots, d\}$. The Jacket $\mathcal{J}$ contains all vertices and edges of $\mathcal{G}_c$, but only the subset of faces $\mathcal{F}_{\mathcal{J}}$ defined by:$$\mathcal{F}_{\mathcal{J}} = \left\{ f \in \mathcal{F}(\mathcal{G}_c) \mid f \text{ is of color type } (\tau^q(0), \tau^{q+1}(0)), \, q \in \mathbb{Z}_{d+1} \right\}\,.$$
\end{definition}
A jacket is a ribbon graph, representing a 2-dimensional sub-manifold. Its Euler-Poincaré characteristic is given by:$$\chi(\mathcal{J}) = |\mathcal{V}_{\mathcal{J}}| - |\mathcal{E}_{\mathcal{J}}| + |\mathcal{F}_{\mathcal{J}}| = 2 - 2g_{\mathcal{J}}\,,$$
where $g_{\mathcal{J}} \geq 0$ denotes the genus of the surface.

\begin{definition}
\textbf{(Gurau degree)} The Gurau degree $\varpi(\mathcal{G}_c)$ of a colored extension $\mathcal{G}_c$ is defined as the sum of the genera of all its jackets:
\begin{equation*}
\varpi(\mathcal{G}_{c})=\sum_{\mathcal{J}}g_{\mathcal{J}} \quad \Rightarrow \quad \varpi(\mathcal{G}_{c})\geq 0\,.
\end{equation*}
\end{definition}
\begin{definition}\label{defmelons}
The graphs whose Gurau degree is zero ($\varpi(\mathcal{G}_c) = 0$) are called melonic graphs. These diagrams represent the leading order (LO) contributions in the large-$N$ expansion of the theory.
\end{definition}
In addition to these definitions, we have the three following lemmas:
\begin{lemma}
The melonic graphs are dual to a $d$-dimensional sphere. 
\end{lemma}
\begin{lemma}\label{propdeg}
In dimension $d$, the Gurau degree $\varpi(\mathcal{G}_c)$ is related to the number of bi-colored faces and the number of black (or white) vertices $p $ by the following two fundamental relations:
\begin{equation*}
|\mathcal{F}(\mathcal{G}_{c})|=\dfrac{d(d-1)}{2}p+d-\dfrac{2}{(d-1)!}\varpi(\mathcal{G}_{c})\,,
\end{equation*}
\begin{equation*}
\varpi(\mathcal{G}_{c})=\dfrac{(d-1)!}{2}(p+d-\mathcal{B}^{[d]})+\sum_{i;\rho}\varpi(\mathcal{B}^{\hat{i}}_{(\rho)})\,.
\end{equation*}
In addition, we can show that $p+d-\mathcal{B}^{[d]} \geq 0$.
\end{lemma}
Note that, in this lemma, the sum over $i$ in the second relation includes the color $0$. Furthermore, $\mathcal{B}^{\hat{i}}_{(\rho)}$ denotes the connected component $\rho$ of the subgraph obtained from $\mathcal{G}_c$ by deleting all edges of color $i$ (including the color $0$). We let $\mathcal{B}^{[d]}$ denote the total number of these connected subgraphs. In the literature, these subgraphs are referred to as $d$-bubbles, which explains the terminology used for the interactions in Figure \ref{fig1}. From this lemma, the following proposition is straightforwardly deduced:
\begin{proposition}
Under any $1$-dipole contraction, the degree of a graph is unchanged.
\end{proposition}
Finally, we close this section by the two following definitions and their corollary:
\begin{definition}
A Feynman graph $\mathcal{G}$ such that its colored extension $\mathcal{G}_c$ has zero Gurau degree is said to be melonic.
\end{definition}
\begin{definition}
The face-connected components of a graph $\mathcal{G}$ are defined by the subsets of edges corresponding to the maximal factorized rectangular blocks of its incidence matrix $\epsilon_{fe}$. The elements of this matrix are defined as $\epsilon_{fe} = \pm 1$ if the edge $e$ belongs to the boundary $\partial f$ (the sign being determined by their relative orientation), and $\epsilon_{fe} = 0$ otherwise.
\end{definition}
\begin{corollary}\label{cor1}
Any melonic graph is face-connected. 
\end{corollary}
\noindent
With this theoretical framework established, we now proceed to the analysis of the divergence degree as characterized by Theorem \ref{th1}.

\section{Perturbative renormalization group}\label{App1}

In the Wilson–Polchinski approach to the RG, high-momentum degrees of freedom are systematically integrated out of the partition function, inducing a scale dependence in the coupling constants. Crucially, even if the initial action is truncated to a finite order, a single step of the integration procedure is sufficient to generate all possible couplings allowed by the symmetries, including those absent at the classical level. The flow of these couplings under scale transformations is governed by a functional differential equation namely, the Wilson–Polchinski equation which serves as a rigorous foundation for non-perturbative RG. While it is often regarded as less convenient for explicit numerical computations than the Wetterich framework, it remains a powerful tool for deriving formal non-perturbative results. In this section, we utilize the Wilson–Polchinski equation within a perturbative framework to cross-validate the non-perturbative findings presented in Section \ref{sectionNP}. For a general review of this topic, see \cite{Zinn-Justin:1989rgp}.

\subsection{Wilson-Polchinski Equation}

The starting point of our approach of Wilson-Polchinski equation in the TGFT context is the following theorem, coming from standard properties of Gaussian integration:
\begin{theorem}\label{th1}
Let two non normalized Gaussian measures $d\mu_{C}$ and $d\mu_{C'}$ such as $C'=C+\Delta$, $C$, $C'$, $\Delta > 0$. Then:
\begin{align}
\int d\mu_{C}(\bar{\psi}_1,\psi_1)d\mu_{\Delta}(\bar{\psi}_2,\psi_2)e^{-S_{int}(\psi_1+\psi_2, \bar{\psi}_1+\bar{\psi}_2)}=\left(\dfrac{\det(\Delta C)}{\det(C')}\right)^{1/2}\int d\mu_{C'}(\bar{\psi},\psi)e^{-S_{int}(\psi, \bar{\psi})}.
\end{align}
\end{theorem}
This theorem relies two Gaussian integration with two propagator, and in the Wilson point of view, it is interpreted as a partial integration over rapid modes, associated to the propagator $\Delta$. In order to make this more concrete, one introduce an arbitrary fundamental UV cut-off $\Lambda_0$ and the regularized propagator:
\begin{equation}\label{regularization}
C_{\Lambda_0}(\vec{p}\,)= \int_{1/\Lambda_0^{2\eta}}^{+\infty} dte^{-t(\sum_{i=1}^4|p_i|^{2\eta}+m^{2\eta})}
\end{equation}
where here $m$ is understood as an IR regulator, which can be removed when any IR divergence occur, as it will be the case in the rest of this section. Then, we introduce a dilatation parameter $s<1$, such as any cut-off $\Lambda<\Lambda_0$ can be written as $\Lambda=s\Lambda_0$ for some $s$. This parameter will be used as a step to the gradual integration of the UV modes. Defining:
\begin{align}
\Delta_{\Lambda}(\sigma,\vec{p}\,)&:=C_{\Lambda}(\vec{p}\,)-C_{(1-\sigma)\Lambda}(\vec{p}\,)=\sigma s\frac{d}{ds}C_{s\Lambda_0}(\vec{p}\,)=\sigma [D_{\Lambda}(\sigma)]_{\vec{p},\vec{p}},
\end{align}
with, using \ref{regularization}
\begin{equation}
[D_{\Lambda}(\sigma)]_{\vec{p},\vec{p}^{\,'}}=\frac{2\eta}{\Lambda^{2\eta}}e^{-\sum_{i=1}^4|p_i|^{2\eta}/\Lambda^{2\eta}}\delta_{\vec{p},\vec{p}^{\,'}},
\end{equation}
we can apply the Theorem \ref{th1} and write the partition function as an integral over two fields, respectively associated to "slow" and "rapid" modes. Starting with the partition function $\Omega[S_{int,\,s}]$ at scale $\Lambda=s\Lambda_0$:
\begin{equation}
\Omega[S_{int,\,s}]:=\int d\mu_{C_{\Lambda}}(\bar{\psi},\psi)e^{-S_{int,\,s}[\psi,\bar{\psi}]}.
\end{equation}
Theorem \ref{th1} allows to decompose it into two Gaussian integrals over two fields, with covariances $\Delta_{\Lambda}(\sigma)$ and $C_{\Lambda}$:
\begin{align}
\Omega[S_{int,\,s}]=&\left(\dfrac{\det(\Delta_{\Lambda} C_{(1-\sigma)\Lambda})}{\det(C_{\Lambda})}\right)^{-1/2}\int d\mu_{C_{\Lambda}}(\bar{\psi}_<,\psi_<)\label{decomp}\int d\mu_{\Delta_{\Lambda}}(\bar{\psi}_>,\psi_>)e^{-S_{int,\,s}(\psi_<+\bar{\psi}_>, \bar{\psi}_<+\bar{\psi}_>)},
\end{align}
and identifying the effective action at scale $\Lambda=s\Lambda_0$ as:
\begin{align}\label{effectiveaction}
&e^{-S_{int,\,s(1-\sigma)}(\psi_<,\bar{\psi}_<)}:=\frac{1}{\sqrt{\det \Delta_{\Lambda}}}\int d\mu_{\Delta_{\Lambda}}(\bar{\psi}_>,\psi_>)e^{-S_{int,\,s}(\psi_<+\psi_>, \bar{\psi}_<+\bar{\psi}_>)},
\end{align}
the decomposition \ref{decomp} becomes:
\begin{equation}\label{effectiveaction2}
\Omega[S_{int,\,s}]=\left(\dfrac{\det C_{(1-\sigma)\Lambda}}{\det C_{\Lambda}}\right)^{-1/2}\int d\mu_{C_{s\Lambda}}(\bar{\psi}_<,\psi_<)e^{-S_{int,s}(\psi_<,\bar{\psi}_<)}.
\end{equation}
This equation can be translated as a differential equation describing the flow of the effective action. Keeping only the leading order terms in $\sigma$, we find:
\begin{align}
&\qquad e^{-\Delta S_{int,\,s}(\psi_<,\bar{\psi}_<)}\label{infinitesimalvar}=1-\sigma\Tr\Big[\Big(\frac{\delta^2 S_{int,s}}{\delta \psi \delta \bar{\psi}}-\frac{\delta S_{int,s}}{\delta \psi} \frac{\delta S_{int,s}}{\delta \bar{\psi}}\Big)D_{\Lambda}\Big]+\mathcal{O}(\sigma^2),
\end{align}
where $\Delta S_{int}^{s}(\psi_<,\bar{\psi}_<):=S_{int,\,(1-\sigma)s}(\psi_<,\bar{\psi}_<)-S_{int,\,s}(\psi_<,\bar{\psi}_<)$. Finally, expanding the left hand side at the same order, and identifying the power of $\sigma$ leads to the differential equation describing how an infinitesimal dilatation of cut-off translates to a modification of couplings:
\begin{equation}
s\dfrac{dS_{int,s}}{ds}=-\Tr\bigg\{\Big(\frac{\delta^2 S_{int,s}^{s}}{\delta \psi \delta \bar{\psi}}-\frac{\delta S_{int,s}}{\delta \psi} \frac{\delta S_{int,s}}{\delta \bar{\psi}}\Big)D_{\Lambda}\bigg\}\label{eqflow}.
\end{equation}
\noindent
\noindent
Choosing to work in momentum representation, the bubble expansion of the action $S_{int,s}[\psi,\bar{\psi}]$
\begin{align}\label{actionLambda}
S_{int,s}[\psi,\bar{\psi}]=\sum_{l=1}\mathcal{W}^{(l)}=\sum_{l=1}\sum_{\{\vec{p}_{i},\vec{\bar{p}}_{i}\}}\mathcal{W}^{(l)\,\vec{\bar{p}}_{1},...,\vec{\bar{p}}_{l}}_{\,\,\,\,\,\,\,\vec{p}_{1},...,\vec{p}_{l}}\prod_{i=1}^{l}T_{\vec{p}_{i}}\bar{T}_{\vec{\bar{p}}_{i}},
\end{align}
where the $T_{\vec{p}_{i}}$ designates the Fourier components of the field $\psi$ : $\psi(\vec{\theta})=\sum_{\mathbb{Z}^4}T_{\vec{p}}e^{i\vec{p}\cdot \vec{\theta}}$, and $\mathcal{V}^{(l)}$ is the sum of all the bubbles of valence $2l$ involved in the effective classical action at scale $\Lambda$. By replacing this expression of the action in the flow equation \ref{eqflow}, we obtain the renormalization group equations (RGEs):
\begin{align}
\label{RGE}s\dfrac{d\mathcal{W}^{(l)}}{ds}=-\sum_{\vec{p}}D_{\Lambda}(\vec{p}\,)\frac{\partial}{\partial \bar{T}_{\vec{p}}}\frac{\partial}{\partial T_{\vec{{p}}}}\mathcal{W}^{(l+1)}+\sum_{m=0}^{l-1}\sum_{\vec{p}}D_{\Lambda}(\vec{p}\,)\frac{\partial\mathcal{W}^{(m+1)}}{\partial \bar{T}_{\vec{p}}}\frac{\partial\mathcal{W}^{(l-m)}}{\partial T_{\vec{\bar{p}}}}.
\end{align}
For our purpose, it is convenient to use of dimensionless couplings. Defining by $d(l,I)$ the canonical dimension of the $I(l)$-th bubble of valence $2l$, $\mathcal{W}^{(l,I)}$, we define the dimensionless bubble $\bar{\mathcal{W}}^{(l,I)}$ as :$\mathcal{W}^{(l,I)}=\bar{\mathcal{W}}^{(l,I)}\Lambda^{d(l,I)}$, and the action \ref{actionLambda} as:
\begin{align}\label{actionLambda2}
S_{int,s}[\psi,\bar{\psi}]=\sum_{l=1}\sum_{I=1}^{N(l)}\sum_{\{\vec{p}_{i},\vec{\bar{p}}_{i}\}}\Lambda^{d(l,I)}\bar{\mathcal{W}}^{(l),I,\,\vec{\bar{p}}_{1},...,\vec{\bar{p}}_{l}}_{\,\,\,\,\,\,\,\vec{p}_{1},...,\vec{p}_{l}}\prod_{i=1}^{l}T_{\vec{p}_{i}}\bar{T}_{\vec{\bar{p}}_{i}},
\end{align}
and the flow equations \ref{RGE} become:
\begin{align}
\nonumber s\dfrac{d\bar{\mathcal{W}}^{(l),I}}{ds}=-d(l,I)\bar{\mathcal{W}}^{(l),I}&-\sum_{\vec{p}}D_{\Lambda}(\vec{p}\,)\bigg\{\sum_{J\in\vartheta_I(l)}\Lambda^{d(l+1,J)-d(l,I)}\frac{\partial^2\bar{\mathcal{W}}^{(l+1),J}}{\partial \bar{T}_{\vec{p}}\partial T_{\vec{{p}}}}\\
&+\sum_{m=0}^{l-1}\sum_{(K,L)\in \vartheta_I(l,m)}\Lambda^{d(m+1,K)+d(l-m,L)-d(l,I)}\frac{\partial\bar{\mathcal{W}}^{(m+1),K}}{\partial \bar{T}_{\vec{p}}}\frac{\partial\bar{\mathcal{W}}^{(l-m),L}}{\partial T_{\vec{\bar{p}}}}\bigg\},
\end{align}\label{RGE2}
where the two sets $\vartheta_I(l)$ and $\vartheta_I(l,m)$ are so that the contracted bubbles in the r.h.s have the same connectivity structure as the bubble on the l.h.s. 

\subsection{Perturbative solution}
\label{section3}
\subsubsection{Ansatz and truncation}
We will apply the general RG formalism to our field theory model, and solve the RGEs equations \ref{RGE2} in a perturbative framework. In this way, we choose the following Ansatz: \\

\begin{ansatz}
All the couplings associated to bubbles of valence $4$ and less are assumed to have the same magnitude $\lambda$, and all the bubbles of valence $l>2$ are assumed to have couplings of magnitude $\lambda^{l-1}$. In addition, the necklace bubbles of valence $l>3$ are those of order $\lambda^{l}$. Then, our truncated action at order $\lambda$ writes as:
\begin{align}
\nonumber S_{int,s}[\psi,\bar{\psi}]=&\sum_{\vec{p}\in\mathbb{Z}^4}\bigg[\bar{\mu}_1\Lambda^{3/2}+\bar{\mu}_2\Lambda^{1/2}\sum_{i=1}^4|p_i|+\delta Z\sum_{i=1}^4|p_i|^{3/2}\bigg]T_{\vec{p}}\bar{T}_{\vec{{p}}}+\sum_{\{\vec{p}_{i},\vec{\bar{p}}_{i}\}}\bigg[\lambda_1\sum_{i=1}^4\bar{\mathcal{W}}^{(2),(i)\,\vec{\bar{p}}_{1},\vec{\bar{p}}_{2}}_{melo\,\,\,\vec{p}_{1},\vec{p}_{2}}\\
&\,\,+\bar{\lambda}_2\Lambda\sum_{i=1}^3\bar{\mathcal{W}}^{(2),(1i)\,\vec{\bar{p}}_{1},\vec{\bar{p}}_{2}}_{neck\,\,\,\vec{p}_{1},\vec{p}_{2}}+{\lambda}_3\sum_{i=1}^3\bar{\mathcal{W}}^{(2),(1i)\,\vec{\bar{p}}_{1},\vec{\bar{p}}_{2}}_{neck\,\,\,\vec{p}_{1},\vec{p}_{2}}\sum_{i=1}^4(|p_{1i}|+|p_{2i}|)\bigg]\prod_{i=1}^{2}T_{\vec{p}_{i}}\bar{T}_{\vec{\bar{p}}_{i}}\,.\label{actionLambda2}
\end{align}
\end{ansatz}

\noindent
Note that this approximation make sense in the vicinity of the Gaussian fixed point, and is motivated by the flow equations \ref{RGE2}, implying that, if the perturbation to the Gaussian fixed point leaves in the coupling subspace of valence $4$, the flow tends to be remain in this subspace, up corrections of order $\lambda^3$. \\

\noindent
After to begin the calculation, we have to precise our regime of approximation. We will consider the ultra-violet limit, and we assume that $s\Lambda$ and $\Lambda$ are large. However, we are interested in a relative infra-red physic compared to the fundamental cut-off $\Lambda$, and more precisely, our approximation can be characterized by $\Lambda_0,\Lambda\gg 1$ but $\Lambda/\Lambda_0=s\ll 1$. 

\subsubsection{$\bar{\mathcal{W}}^{(1)}$ at order $\lambda$}

At order $\lambda$, and retaining only the LO contributions in the UV limit, the flow equation \ref{RGE2} for $\bar{\mathcal{W}}^{(1)}$ writes as:
\begin{align}\label{flow1}
 s\dfrac{d\mathcal{W}^{(1)}}{ds}=-2\sum_{\vec{p},\vec{p}_1,\vec{\bar{p}}_1}&D_{\Lambda}(\vec{p}\,)\bigg[\lambda_1\sum_{i=1}^4\bar{\mathcal{W}}^{(2),(i)\,\vec{\bar{p}}_{1},\vec{{p}}}_{melo\,\,\,\vec{p}_{1},\vec{p}}\cr
&+\bigg(\bar{\lambda}_2\Lambda+\lambda_3\sum_{i=1}^4(|p_{1i}|+|p_{2i}|)\bigg)\sum_{i=1}^3\sym\bar{\mathcal{W}}^{(2),(1i)\,\vec{\bar{p}}_{1},\vec{{p}}}_{neck\,\,\,\vec{p}_{1},\vec{p}}\bigg]T_{\vec{p}_1}\bar{T}_{\vec{\bar{p}}_1},
\end{align}
or explicitly:
\begin{align}
\nonumber s\dfrac{d\mathcal{W}^{(1)}}{ds}=-2\sum_{\vec{p},\vec{p}_1,\vec{\bar{p}}_1}D_{\Lambda}(\vec{p}\,)\bigg[\lambda_1\sum_{i=1}^4\delta_{p_i,p_{1,i}}&\delta_{p_i\bar{p}_{1i}}\prod_{j\neq i}\delta_{p_{1j}\bar{p}_{1j}}+\bigg(\bar{\lambda}_2\Lambda+\lambda_3\sum_{i=1}^4(|p_{1i}|+|p_{2i}|)\bigg)\\
&\quad\times\sum_{i=1}^3\bigg(\prod_{l=1,i}\delta_{p_l,p_{1,l}}\delta_{p_l\bar{p}_{1l}}\prod_{j\neq 1,i}\delta_{p_{1j}\bar{p}_{1j}}+l\leftrightarrow j\bigg)\bigg]T_{\vec{p}_1}\bar{T}_{\vec{\bar{p}}_1}\,.\label{flow11}
\end{align}
\noindent
The right hand side receives two contribution, coming from the melonic and necklace bubbles, and pictured on Figure\ref{figA1}, and we will compute each contributions separately. \\
\begin{center}
\includegraphics[scale=0.9]{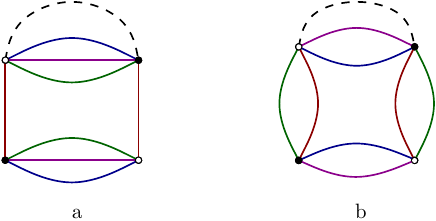} 
\captionof{figure}{The two leading contribution to the RGE for $\mathcal{W}^{(1)}$. The dotted line represent the contraction with $D_{\Lambda}$.}\label{figA1}
\end{center}
\medskip

\noindent
\textit{Contribution of the melonic bubbles.}
The first term of the right hand side involves the typical sum:
\begin{equation}
\mathcal{S}_{melo}=\dfrac{2\eta}{\Lambda^{2\eta}}\sum_{\vec{q}}\delta_{q_1p_1}e^{-\sum_i|q_i|^{2\eta}/\Lambda^{2\eta}}.
\end{equation}
Because we are only interested by the leading order contributions in $\Lambda$, in the large $\Lambda$ limit, we use of the integral approximation:
\begin{equation}
\mathcal{S}_{melo}\approx \mathcal{I}_{melo}:=2\eta\Lambda^{3-2\eta}\int d^3y\prod_{i=2}^4  e^{-|y_i|^{2\eta}}e^{-|y_1|^{2\eta}},
\end{equation}
with $y_i:=p_i/\Lambda$. Making the change of variable $z_i=x|y_i|^{2\eta}\,i=2,3,4$, the integral approximation writes as:
\begin{align}
\mathcal{I}_{melo}=16\eta \Lambda^{3-2\eta}\bigg[\frac{1}{2\eta}\int_0^{\infty}dzz^{1/2\eta-1}e^{-z}\bigg]^3e^{-|p_1|^{2\eta}/\Lambda^{2\eta}},
\end{align}
and expanding the exponential with the value $\eta=3/4$, we find:
\begin{equation}
\mathcal{I}_{melo}=12\big[\Gamma(5/3)\big]^3\Lambda^{3/2}-12[\Gamma(5/3)\big]^3|p_1|^{2\eta}+\text{NLO}\,.\label{melonsum}
\end{equation}

\noindent
\textit{Contribution of the necklace bubbles.}
The necklace bubbles contribution involves the two typical sums:
\begin{align}
\mathcal{S}_{neck,1}&=\dfrac{2\eta}{\Lambda^{2\eta}}\sum_{\vec{q}}\delta_{q_1p_1}\delta_{q_2p_2}e^{-\sum_i|q_i|^{2\eta}/\Lambda^{2\eta}},\\
\mathcal{S}_{neck,2}&=\dfrac{2\eta}{\Lambda^{2\eta}}\sum_{\vec{q}}\delta_{q_1p_1}\delta_{q_2p_2}|q_1| e^{-\sum_i|q_i|^{2\eta}/\Lambda^{2\eta}},
\end{align}
with integral approximations:
\begin{align}
\mathcal{I}_{neck,1}&=2\eta\Lambda^{2-2\eta}\int d^2y \prod_{i=1}^2e^{-|y_i|^{2\eta}}e^{-(|p_1|^{2\eta}+|p_2|^{2\eta})/\Lambda^{2\eta}},\\
\mathcal{I}_{neck,2}&=2\eta\Lambda^{3-2\eta}\int d^2y |y_1|\prod_{i=1}^2e^{-|y_i|^{2\eta}}e^{-(|p_1|^{2\eta}+|p_2|^{2\eta})/\Lambda^{2\eta}},
\end{align}
given:
\begin{align}
\mathcal{I}_{neck,1}&=6\big[\Gamma(5/3)\big]^2\Lambda^{1/2}-6\big[\Gamma(5/3)\big]^2\Lambda^{-1}(|p_1|^{2\eta}+|p_2|^{2\eta})+\text{NLO},\\
\mathcal{I}_{neck,2}&=3\Gamma(5/3)\Gamma(7/3)\Lambda^{3/2}-3\Gamma(5/3)\Gamma(7/3)(|p_1|^{2\eta}+|p_2|^{2\eta})+\text{NLO}\,. \label{necksum}
\end{align}
Denoting with $\mathcal{O}(\vec{p}\,)$ the kernel,
\begin{equation}
\mathcal{W}^{(1)}=:\sum_{\vec{p}\in\mathbb{Z}^4}\mathcal{O}(\vec{p}\,)T_{\vec{p}}\bar{T}_{\vec{{p}}},
\end{equation}
given from equation \ref{actionLambda2}:
\begin{equation}\label{kernel}
\mathcal{O}(\vec{p}\,):=\bar{\mu}_1\Lambda^{3/2}+\bar{\mu}_2\Lambda^{1/2}\sum_{i=1}^4|p_i|+\delta Z\sum_{i=1}^4|p_i|^{3/2},
\end{equation}
and using the two integral approximations \ref{melonsum} and \ref{necksum} in the flow equation, we find, at leading order in $\Lambda$:
\begin{align}
\nonumber s\dfrac{d\mathcal{O}(\vec{p}\,)}{ds}=&-12\big[4\lambda_1[\Gamma(5/3)]^3+3\bar{\lambda}_2[\Gamma(5/3)]^2+3\lambda_3\Gamma(5/3)\Gamma(7/3)\big]\Lambda^{3/2}\\\nonumber
&-
36\bar{\lambda}_2[\Gamma(5/3)]^2\Lambda^{1/2}\sum_{i=1}^4|p_i|\\
&+12[2\lambda_1[\Gamma(5/3)]^3+3\bar{\lambda}_2[\Gamma(5/3)]^2+3\lambda_3\Gamma(5/3)\Gamma(7/3)]\sum_{i=1}^4|p_i|^{2\eta},\label{flow11}
\end{align}
implying, for the dimensionless parameters of equation \ref{kernel}:
\begin{equation}\label{flow1f1}
s\dfrac{d\bar{\mu}_1}{ds}=-\frac{3}{2}\bar{\mu}_1-12\big[4\lambda_1[\Gamma(5/3)]^3+3\bar{\lambda}_2[\Gamma(5/3)]^2+3\lambda_3\Gamma(5/3)\Gamma(7/3)\big],
\end{equation}
\begin{equation}\label{flow1f2}
s\dfrac{d\bar{\mu}_2}{ds}=-\frac{1}{2}\bar{\mu}_2-36\bar{\lambda}_2[\Gamma(5/3)]^2,
\end{equation}
\begin{equation}
s\dfrac{d\delta Z}{ds}=12\big[2\lambda_1[\Gamma(5/3)]^3+3\bar{\lambda}_2[\Gamma(5/3)]^2+3\lambda_3\Gamma(5/3)\Gamma(7/3)\big].
\end{equation}

\subsubsection{$\bar{\mathcal{W}}^{(2)}$ at order $\lambda^2$}

The flow equations for the $4$-valent bubbles involve several contributions. To begin with, the first term on the right-hand side (r.h.s.) generates both 1PI contributions, as shown in Figure \ref{figA21}, and 1PR contributions, typical examples of which are illustrated in Figure \ref{figA22}. The leading orders of these 1PR contributions are, in fact, exactly compensated by the second term on the r.h.s. of the flow equation, which involves the $2$-point interactions at order $\lambda$ calculated in the previous subsection.
Taking this cancellation into account, we obtain the flow equations for each of the three couplings $\lambda_I$ by identifying terms on both sides of the equation that share the same bubble connectivity structure. For instance, the contribution to $\beta_{\lambda_1}$ involves diagram \ref{figA21}a with $i=j$, and diagram \ref{figA21}c for $k=i$. Following the same procedure for all couplings, we find:
\begin{equation}
s\frac{d\lambda_1}{ds}=-4\lambda_1^2\mathcal{S}_{melo,2}-24\lambda_1(2\lambda_3\mathcal{S}_{neck,4}+\bar{\lambda}_2\Lambda\mathcal{S}_{neck,3}),
\end{equation}
\begin{equation}
s\frac{d\bar{\lambda}_2}{ds}=-\bar{\lambda}_2-8\bar{\lambda}_2^2\Lambda\mathcal{S}_{neck,3}-32\lambda_3\bar{\lambda}_2\mathcal{S}_{neck,4}-16\lambda_3^2\Lambda^{-1}\mathcal{S}_{neck,5},
\end{equation}
\begin{equation}
s\frac{d\lambda_3}{ds}=-6\lambda_3^2\mathcal{S}_{neck,4}-6\lambda_3\bar{\lambda}_2\Lambda\mathcal{S}_{neck,3},
\end{equation}
where the four sums are defined as:
\begin{equation}
\mathcal{S}_{melo,2}:=(2\eta)^2\Lambda^{-2\eta}\sum_{q_1,q_2,q_3}\int_{\Lambda_0}^{\Lambda}d\Lambda' \Lambda'^{\,-2\eta-1}\prod_{i=1}^3e^{-\big(\frac{1}{\Lambda^{2\eta}}+\frac{1}{\Lambda'^{2\eta}}\big)|q_i|^{2\eta}}\approx -6[\Gamma(5/3)]^3,
\end{equation}
\begin{equation}
\mathcal{S}_{neck,3}:=(2\eta)^2\Lambda^{-2\eta}\sum_{q_1,q_2}\int_{\Lambda_0}^{\Lambda}d\Lambda' \Lambda'^{\,-2\eta-1}\prod_{i=1}^2e^{-\big(\frac{1}{\Lambda^{2\eta}}+\frac{1}{\Lambda'^{2\eta}}\big)|q_i|^{2\eta}}\approx -18(1-2^{-1/3})[\Gamma(5/3)]^2\Lambda^{-1},
\end{equation}
\begin{equation}
\mathcal{S}_{neck,4}:=(2\eta)^2\Lambda^{-2\eta}\sum_{q_1,q_2}|q_1|\int_{\Lambda_0}^{\Lambda}d\Lambda' \Lambda'^{\,-2\eta-1}\prod_{i=1}^2e^{-\big(\frac{1}{\Lambda^{2\eta}}+\frac{1}{\Lambda'^{2\eta}}\big)|q_i|^{2\eta}}\approx -3\Gamma(5/3)\Gamma(7/3),
\end{equation}
\begin{align}
\nonumber\mathcal{S}_{neck,5}&:=(2\eta)^2\Lambda^{-2\eta}\sum_{q_1,q_2}(|q_1|^2+|q_1||q_2|)\int_{\Lambda_0}^{\Lambda}d\Lambda' \Lambda'^{\,-2\eta-1}\prod_{i=1}^2e^{-\big(\frac{1}{\Lambda^{2\eta}}+\frac{1}{\Lambda'^{2\eta}}\big)|q_i|^{2\eta}}\\
&\,\,\approx -\frac{9}{5}(1-2^{-5/3})\bigg(\frac{1}{2}[\Gamma(7/3)]^2+3\Gamma(5/3)\bigg)\Lambda,
\end{align}
giving:
\begin{equation}\label{flow2f1}
s\frac{d\bar{\lambda}_1}{ds}=24[\Gamma(5/3)]^3\bar{\lambda}_1^2+144\lambda_1\Gamma(5/3)\big[\bar{\lambda}_3\Gamma(7/3)+3(1-2^{-1/3})\bar{\lambda}_2\Gamma(5/3)\big]\,,
\end{equation}
\begin{equation}\label{flow2f2}
s\frac{d\bar{\lambda}_2}{ds}=-\bar{\lambda}_2+144\bar{\lambda}_2^2(1-2^{-1/3})[\Gamma(5/3)]^2+96\lambda_3\bar{\lambda}_2\Gamma(5/3)\Gamma(7/3)+\frac{144}{5}\bar{\lambda}_3^2(1-2^{-5/3})\bigg(\frac{1}{2}[\Gamma(7/3)]^2+3\Gamma(5/3)\bigg)\,,
\end{equation}
\begin{equation}\label{flow2f3}
s\frac{d\bar{\lambda}_3}{ds}=18\bar{\lambda}_3^2\Gamma(5/3)\Gamma(7/3)+108(1-2^{-1/3})[\Gamma(5/3)]^2\bar{\lambda}_3\bar{\lambda}_2\,.
\end{equation}
\begin{center}
\includegraphics[scale=1]{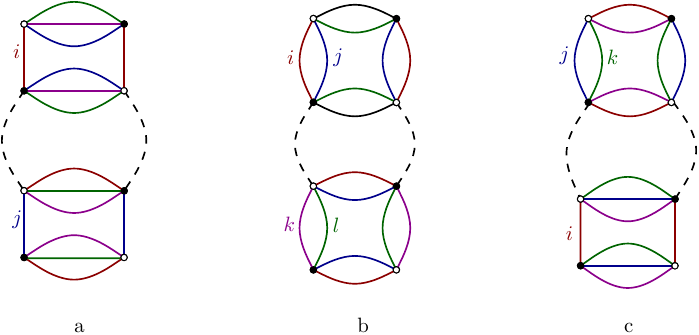} 
\captionof{figure}{Typical 1PI contributions to the equation for $\bar{\mathcal{W}}^{(2)}$}
\end{center}\label{figA21}
\begin{center}
\includegraphics[scale=1]{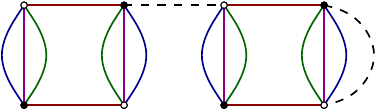} 
\captionof{figure}{1PR contribution comping from the first term on the r.h.s of the flow equation for $\bar{\mathcal{W}}^{(2)}$}
\end{center}\label{figA22}

\subsubsection{RGEs for dimensionless renormalized couplings}

The renormalized couplings $\bar{\lambda}_I^r$, $\bar{\mu}_I^r$ are defined as:
\begin{align}
\bar{\mu}_I^r&=Z^{-1}\bar{\mu}_I,\\
\bar{\lambda}_I^r&=Z^{-2}\bar{\lambda}_I.
\end{align}
With these definitions, and using the flow equations \ref{flow1f1},    \ref{flow1f2},   \ref{flow2f1}, \ref{flow2f2},   \ref{flow2f3}, we find:
\begin{equation}\label{flow1ff1}
s\dfrac{d\bar{\mu}_1^r}{ds}=-\frac{3}{2}\bar{\mu}_1^r-12\big[4\lambda_1^r[\Gamma(5/3)]^3+3\bar{\lambda}_2^r[\Gamma(5/3)]^2+3\bar{\lambda}_3^r\Gamma(5/3)\Gamma(7/3)\big]\,,
\end{equation}
\begin{equation}\label{flow1ff2}
s\dfrac{d\bar{\mu}_2^r}{ds}=-\frac{1}{2}\bar{\mu}_2^r-36\bar{\lambda}_3^r[\Gamma(5/3)]^2\,,
\end{equation}
\begin{equation}\label{flow2ff1}
s\frac{d\bar{\lambda}_1^r}{ds}=-24\bar{\lambda}_1^{r\,2}[\Gamma(5/3)]^3+72\bar{\lambda}_1^{r}\bar{\lambda}_3^r\Gamma(5/3)\Gamma(7/3)+\bar{\lambda}_1^r\bar{\lambda}_2^r[432(1-2^{-1/3})-72][\Gamma(5/3)]^2\,,
\end{equation}
\begin{align}\label{flow2ff2}
\nonumber s\frac{d\bar{\lambda}_2^r}{ds}=-\bar{\lambda}_2^r-72 (1- 2^{2/3})\bar{\lambda}_2^{r\,2}[\Gamma(5/3)]^2&-48\bar{\lambda}_1^r\bar{\lambda}_2^r[\Gamma(5/3)]^3+24\bar{\lambda}_2^r\bar{\lambda}_3^r\Gamma(5/3)\Gamma(7/3)\\
&+\frac{144}{5}\bar{\lambda}_3^{r\,2}(1-2^{-5/3})\big([\Gamma(7/3)]^2+6\Gamma(5/3)\big)\,,
\end{align}\label{flow2ff2}
\begin{equation}\label{flow2ff3}
s\frac{d\bar{\lambda}_3^r}{ds}=-54\bar{\lambda}_3^2\Gamma(5/3)\Gamma(7/3)+[108(1-2^{-1/3})-72][\Gamma(5/3)]^2\bar{\lambda}_3\bar{\lambda}_2-48\bar{\lambda}_3\bar{\lambda}_1[\Gamma(5/3)]^3\,.
\end{equation}

\begin{remark}
As a consistency check, the one-loop contributions of renormalizable interactions match those obtained by expanding the non-perturbative equations derived in this paper around the Gaussian fixed point, as required by universality.
\end{remark}

\subsection{Gaussian fixed point}

The flow around the Gaussian fixed point is illustrated in Figure \ref{figgauss}. It is immediately apparent that the theory is not asymptotically free, which is also confirmed by calculation. A direct method consists in examining the evolution of the norm of the vector $\vec{\lambda} := (\bar{\lambda}_1^r, \bar{\lambda}_2^r, \bar{\lambda}_3^r)$,
\begin{equation}
s\frac{d}{ds}\, \vec{\lambda}^2= 2 \sum_i \bar{\lambda}_i^r \beta_i\, ,
\end{equation}
where $\beta_i := s \, d\bar{\lambda}_i^r / ds$. The theory is asymptotically free if $\sum_i \bar{\lambda}_i^r \beta_i < 0$. The $\beta$-functions for the renormalizable couplings are all of the form:

\begin{equation}
\beta_i := \sum_{j,k}\,A^{(i)}_{jk} \bar{\lambda}_j^r\bar{\lambda}_k^r\,.
\end{equation}

The presence of at least one positive eigenvalue for any of the matrices $A^{(i)}$ is sufficient for asymptotic freedom to be violated; these eigenvalues are given by:

\begin{equation}
\text{vp} [A^{(1)}]:=\{-98.25,62.94,0.\}\,,
\end{equation}
\begin{equation}
\text{vp} [A^{(3)}]:=\{-137.15,21.06,0.\}\,.
\end{equation}
There are clearly several directions of instability; the theory is not asymptotically free.

\begin{figure}[htbp]
\begin{center}
\includegraphics[scale=0.35]{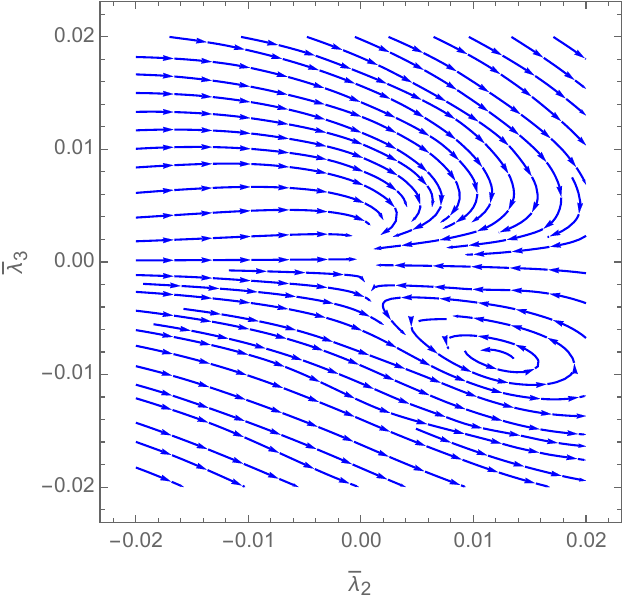}\quad \includegraphics[scale=0.35]{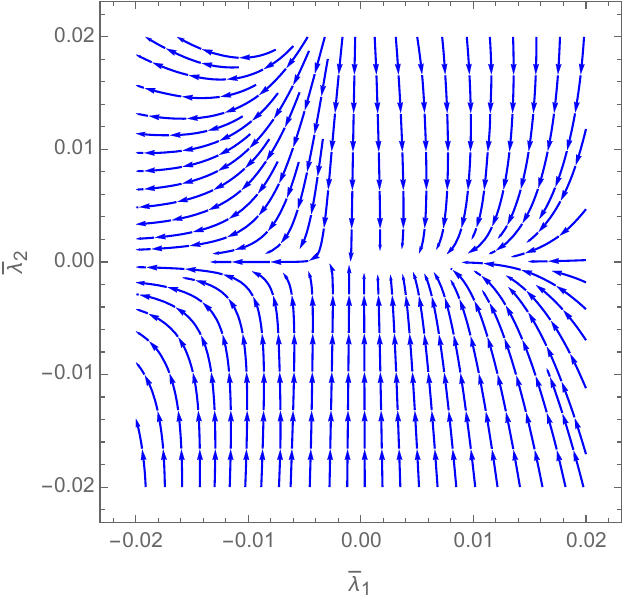}\quad \includegraphics[scale=0.35]{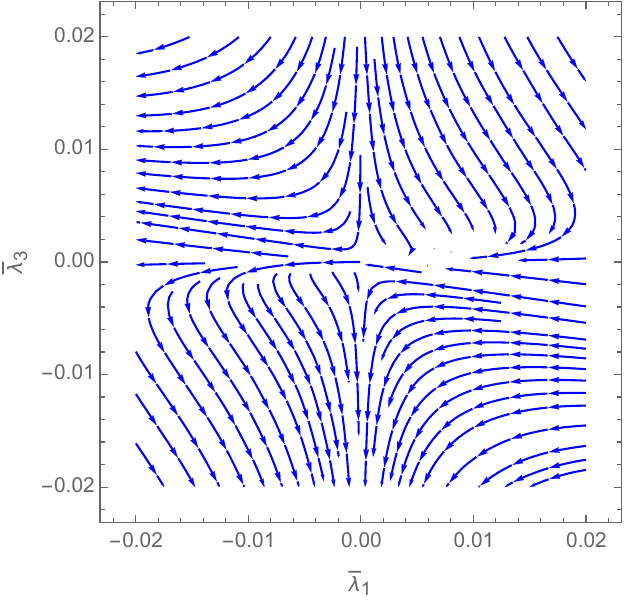}
\end{center}
\caption{Behavior of the RG flow in the vicinity of the gaussian fixed point.}\label{figgauss}
\end{figure}

\subsection{Non-gaussian fixed points}

Numerical analysis of the flow equations reveals the emergence of three non-trivial fixed points in addition to the Gaussian fixed point\footnote{A fourth one is observed, characterized by very large critical exponents and anomalous dimension. We discard it based on the principle that the perturbative regime should not deviate significantly from the Gaussian theory.}.

A first fixed point appears for the following values:

\begin{equation}
\text{FP1}:=\{\bar{\mu}_1^r,\bar{\mu}_2^r,\bar{\lambda}_1^r,\bar{\lambda}_2^r, \bar{\lambda}_3^r\} \,\approx \, \{0.053, 0.304, -0.01, 0.016, -0.005\}\,,
\end{equation}

for which the (real part of) critical exponents are:

\begin{equation}
\{\theta_i\}_1 \, \approx \, \{1.5, 0.5, -0.121, -0.121, -0.57\}\,,
\end{equation}

and the anomalous dimension is $\eta_1 \approx 0.096$. This fixed point, which couples melons and necklaces, possesses three irrelevant directions in the IR, spanning a 3D stable surface, and two relevant (critical) directions.  Two other 'pure necklace' fixed points emerge at the values:

\begin{equation}
\text{FP2}:=\{\bar{\mu}_1^r,\bar{\mu}_2^r,\bar{\lambda}_1^r,\bar{\lambda}_2^r, \bar{\lambda}_3^r\} \,\approx \, \{-0.02, 0.5, 0., 0.012, -0.008\}\,,
\end{equation}
\begin{equation}
\text{FP3}:=\{\bar{\mu}_1^r,\bar{\mu}_2^r,\bar{\lambda}_1^r,\bar{\lambda}_2^r, \bar{\lambda}_3^r\} \,\approx \, \{-0.57, 0, 0., 0.029, 0.\}\,,
\end{equation}
with the following critical exponents:
\begin{equation}
\{\theta_i \}_2 \, \approx \, \{1.5, 0.5, 0.488, -0.056, -0.056 \}\,,
\end{equation}
\begin{equation}
\{\theta_i\}_3\, \approx \,\{1.5, 1.176, 0.5, -0.404, -1.\}\,,
\end{equation}
and anomalous dimensions $\eta_2 \approx 0.03$ and $\eta_3 \approx 0.85$, respectively. These two fixed points possess two irrelevant directions in the IR, spanning a two-dimensional stable space. Physically, they correspond to a non-trivial summation of planar graphs. Figures \ref{figflow1} and \ref{figflow2} illustrate the flow behavior around $\text{FP1}, \text{FP2}$ and $\text{FP3}$. It should be noted that many values of the couplings are beyond the perturbative regime, highlighting the importance of non-perturbative methods for their study. Moreover, none of the discovered fixed point in the non-perturbative regime match with the perturbative ones, seeming to indicate that the UV completion issue is a non perturbative problem.

\begin{figure}
\begin{center}
\includegraphics[scale=0.32]{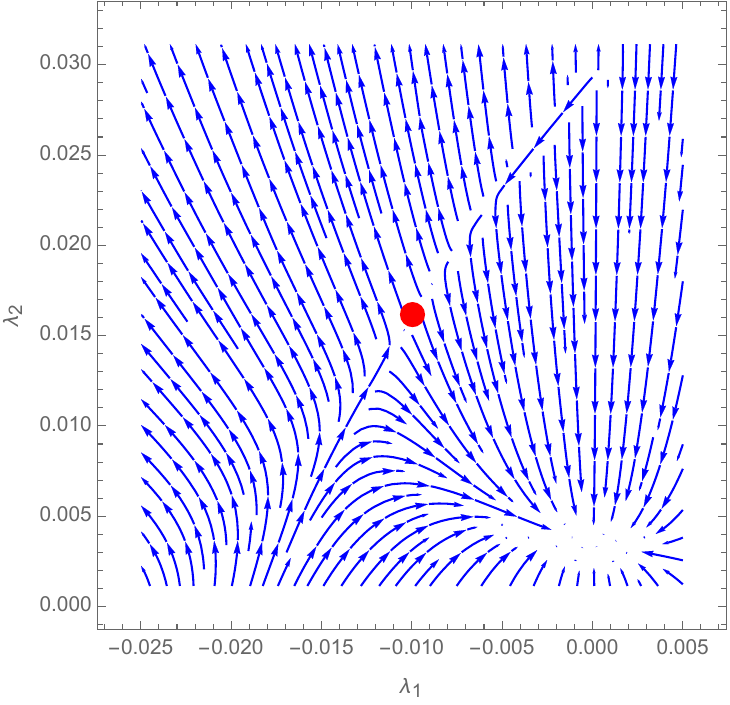}\quad \includegraphics[scale=0.32]{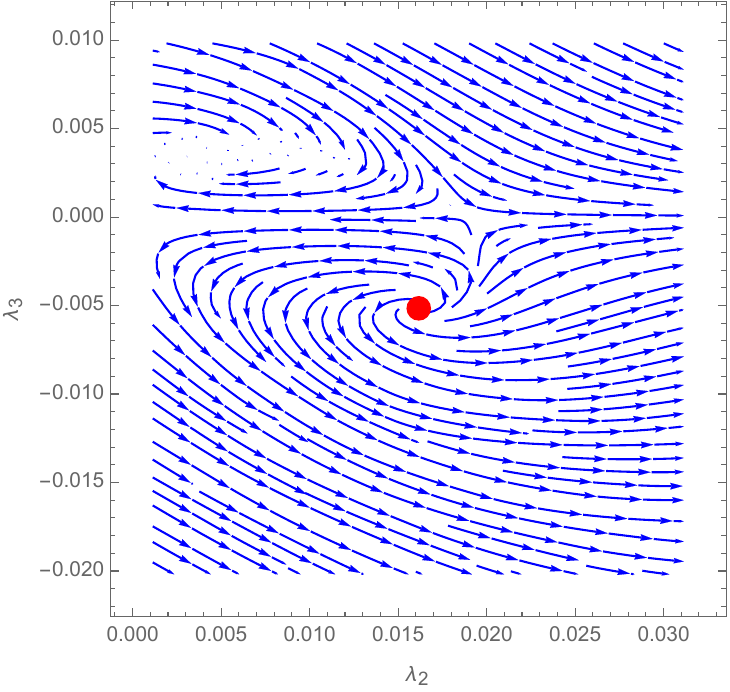}\quad \includegraphics[scale=0.32]{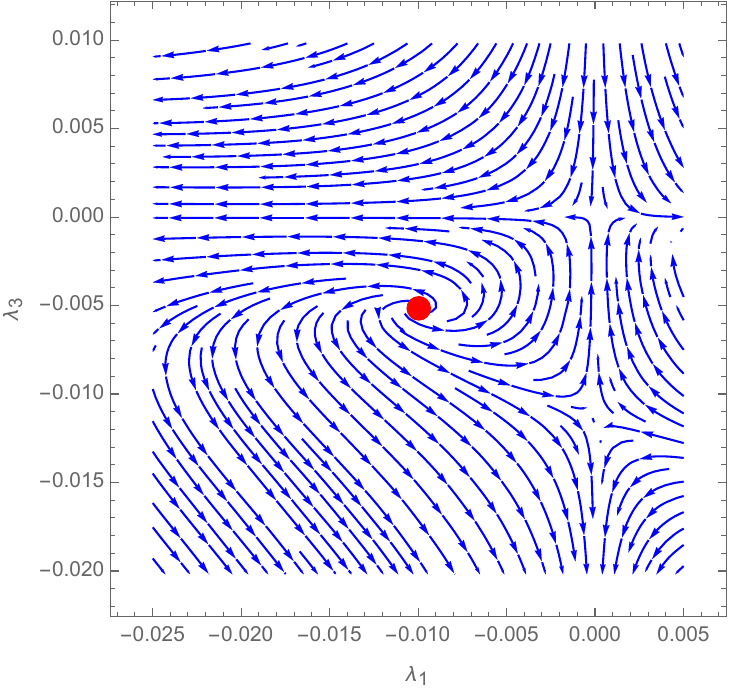}
\end{center}
\caption{Behavior of the RG flow in the vicinity of $\text{FP1}$.}\label{figflow1}
\end{figure}

\begin{figure}
\begin{center}
\includegraphics[scale=0.5]{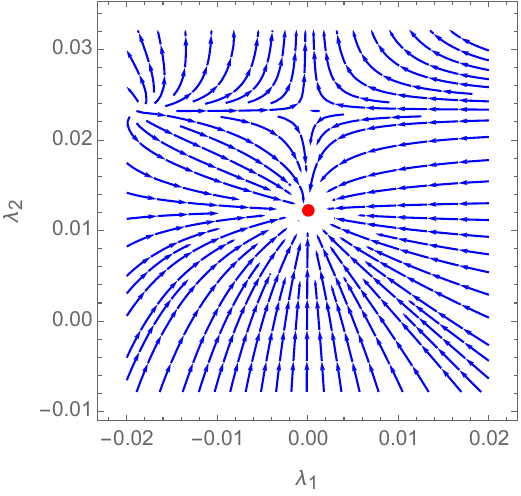}\qquad \includegraphics[scale=0.5]{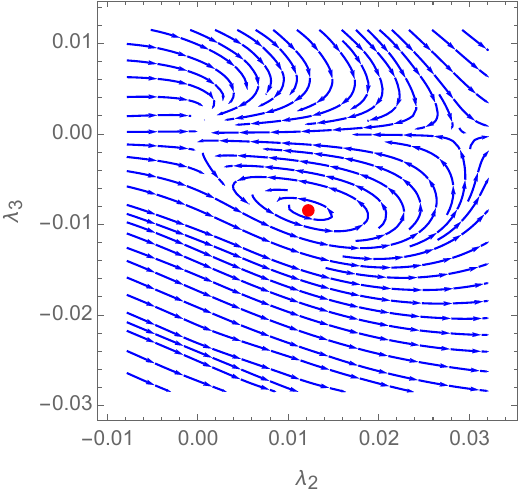}\\
\includegraphics[scale=0.5]{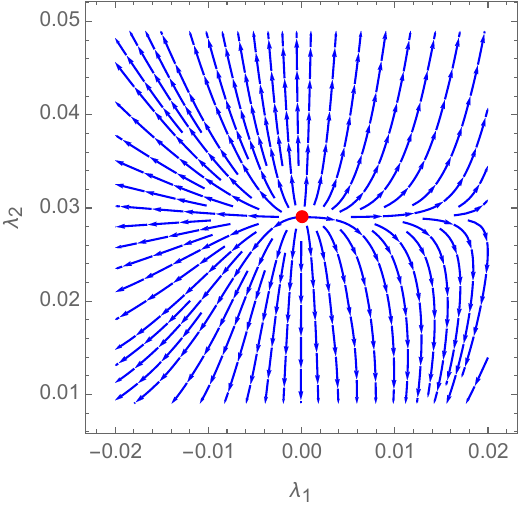}\qquad \includegraphics[scale=0.5]{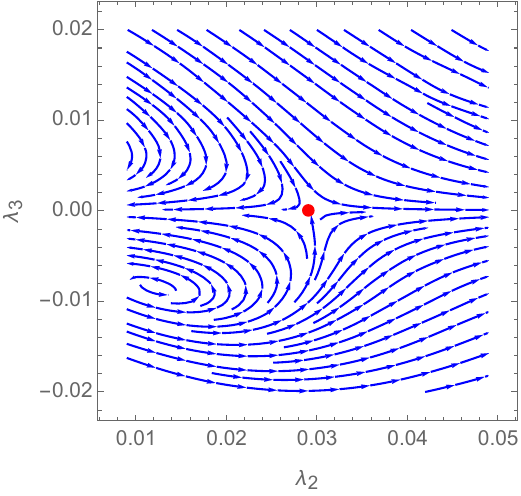}
\end{center}
\caption{Behavior of the RG flow in the vicinity of $\text{FP2}$ (on the top) and $\text{FP3}$ (on the bottom).}\label{figflow2}
\end{figure}

\pagebreak



\clearpage
\printbibliography[heading=bibintoc]

\clearpage

\end{document}